\pdfoutput=1
\documentclass[a4paper,oneside,reqno]{amsart}

\usepackage[english]{babel}
\usepackage[utf8]{inputenc}		
\usepackage[scaled]{helvet}		
\usepackage{courier}			
\usepackage{eulervm}			
\usepackage[bb=ams,bbscaled=1.05,scr=dutchcal]{mathalpha}	%
\usepackage[T1]{fontenc}
\normalfont

\usepackage[a4paper,margin=2.5cm]{geometry}
\usepackage[dvipsnames]{xcolor}
\usepackage[hypertexnames=false,colorlinks=true,linkcolor=blue,%
citecolor=purple,filecolor=magenta,urlcolor=cyan]{hyperref}                      
\usepackage{amssymb,anyfontsize,enumerate,mathrsfs,mdwlist,stmaryrd,tikz}
\usepackage{dsfont} 
\usepackage[textsize=footnotesize]{todonotes}
\presetkeys%
{todonotes}%
{color=Apricot}{}%
\usepackage[style=alphabetic,maxnames=99,maxalphanames=5]{biblatex}
\usepackage{csquotes}
\usepackage{mathtools}
\usepackage{braket}

\usepackage[english,noabbrev]{cleveref}
\crefname{lemma}{lemma}{lemmata}
\Crefname{lemma}{Lemma}{Lemmata}

\allowdisplaybreaks 

\theoremstyle{plain}                          
\newtheorem{theorem}{Theorem}[section]
\newtheorem{proposition}[theorem]{Proposition}    
\newtheorem{lemma}[theorem]{Lemma}
\newtheorem{corollary}[theorem]{Corollary}
\newtheorem{conjecture}[theorem]{Conjecture}
\theoremstyle{definition}
\newtheorem{definition}[theorem]{Definition}
 
\newtheorem{example}[theorem]{Example}
\theoremstyle{remark}
\newtheorem{remark}[theorem]{Remark}

\renewcommand{\theta}{\vartheta}
\renewcommand{\phi}{\varphi}
\renewcommand{\epsilon}{\varepsilon}

\newcommand{\bs}[1]{\ensuremath{\boldsymbol{#1}}}
\newcommand{\mb}[1]{\mathbb{#1}} 
\newcommand{\mf}[1]{\mathfrak{#1}}
\newcommand{\mc}[1]{\mathcal{#1}}

\newcommand{\R}{\mb{R}} 
\newcommand{\N}{\mb{N}} 
\newcommand{\C}{\mb{C}} 
\newcommand{\Z}{\mb{Z}} 
\newcommand{\Q}{\mb{Q}}
\renewcommand{\P}{\mb{P}}

\DeclareMathOperator{\Log}{Log}
\DeclareMathOperator{\spa}{span}
\DeclareMathOperator{\Sym}{Sym}
\DeclareMathOperator{\csch}{csch}

\DeclareMathOperator{\Aut}{Aut}
\DeclareMathOperator{\sgn}{sgn}
\DeclareMathOperator{\Erd}{Erd}
\DeclareMathOperator{\Ord}{Ord}
\DeclareMathOperator{\Mult}{Mult}

\DeclareMathOperator*{\Res}{Res}

\newcommand{\raisingfactorial}[1]{%
	^{\mspace{2mu}\overline{\mspace{-2mu}#1\mspace{-2mu}}\mspace{2mu}}%
}

\begingroup\expandafter\expandafter\expandafter\endgroup
\expandafter\ifx\csname pdfsuppresswarningpagegroup\endcsname\relax
\else
\fi
{}

\begin{document}
	
\title{Topological recursion on spectral curves with infinite ramification loci}

\author[Q.~Weller]{Quinten Weller}
\email{qgw1@nyu.edu}
\address{New York University, Department of Physics, 726 Broadway, New York, NY 10003, USA}	
	
\thanks{}

\begin{abstract}
Recently in the physics literature there have been a small, but growing, number of papers that have applied the Eynard-Orantin topological recursion procedure to spectral curves with ramification loci of infinite cardinality. However, there is a corresponding gap in the mathematics literature; indeed, it has not been checked rigorously whether the resulting infinite sums converge, or whether the resulting correlators have all the desired properties of the topological recursion. The present work aims to bridge this gap by defining topological recursion on a suitably broad class of spectral curves to cover the aforementioned physics applications and to rigorously prove that this definition has the properties one would intuitively expect. As a by-product, some interesting coincidences involving quantum curves are discovered, where two different spectral curves yield the same quantum curve, a situation that has, hitherto, not appeared in the literature.
\end{abstract}
	
\maketitle
	
\tableofcontents


\section{Introduction}
\subsection{Motivation}\label{ss:mot}

Topological recursion (TR) was originally developed to solve loop equations that arose from matrix models; the solution of these loop equations yielded an all orders solution of the large-$N$ expansion of the matrix model correlators \cite{E04,CEO06}. However, in \cite{EO07} it was realised that one could consider the recursion as a mathematical formalism independent of its physical origins, thereby employing it as a general method of associating a collection of symmetric multidifferentials (the `correlators') to an object called a spectral curve. Since the key insight of \cite{EO07}, the topological recursion has found uses across many fields of mathematics and physics including, but not limited to: mirror symmetry \cite{BM08,GS11,BS12,Z12}; knot theory \cite{ABM12,BE12,GJKS14,DPSS19}; Gromov--Witten theory \cite{BKMP08,EO15,DOSS14,NS14,FLZ17,FLZ20,GKLS22}; intersection theory on the moduli space of curves \cite{EO07,E11,DOSS14}; several kinds of integrable systems \cite{BEM17,EGMO21,BDKS20a}.

A spectral curve, the initial data of TR, is a quadruple $\mc{S} = \left(\Sigma,x,y,\omega_{0,2}\right)$ where: $\Sigma$ is a Riemann surface; $x,y:\Sigma\to \P^1$ are holomorphic except, possibly, on a finite set of points; $\omega_{0,2}$ is a symmetric $2$-differential with poles only on the diagonal. Previously, it has always been assumed that $x$ has a finite ramification locus, i.e. the set of zeroes and non-simple poles of $dx$ was finite. However, in the physics literature a few papers have recently appeared: \cite{CEMR24,CEMR25}, which study the worldsheet theory of what the authors dub the `complex Liouville string'; \cite{CER26}, which studied the $c=1$ string from the context of matrix modles and topological recursion; \cite{DLY25} which studies a supersymmetric variant of the complex Liouville String; \cite{BERW26}, which studies the related, so-called `Runkel-Watts' string; \cite{A25}, which studies the implications of the $x$-$y$ duality (see \Cref{ss:xy}) in TR on minimal strings, and, in a limiting case, JT gravity; in all of these works computations were performed on curves where the ramification locus of $x$ was infinite. In \cite{CEMR24,CEMR25}, $x(z) = -2 \cos(\pi b^{-1}z)$ ($b\in\C^*$ is a fixed parameter), in \cite{CER26} $x(z) = 2\sqrt{2}\cos(z)$, in \cite{DLY25} $x(z) = 2 \sin(\pi b^{-1} z)$, in \cite{BERW26} $x(z) = -2 \cos(\pi z / \sqrt{qq'})$ ($q,q'\in \C^*$ are fixed parameters), and in \cite{A25} $x(z) = (4\pi)^{-1}\sin(2\pi z)$.

The correct definition of topological recursion with infinitely many ramification points is more or less obvious (although one may worry about possible contributions from the inevitable essential singularities, such as those that appear in \cite{BKW24,AH26}): the original Eynard-Ornatin (or Bouchard-Eynard in the case of higher order ramification) should be applied locally around each ramification point of $x$ and the infinite number of contributions summed, and this was the approach applied in \cite{CEMR24,CEMR25,CER26,A25,DLY25,BERW26}. However, to make this intuitively obvious definition mathematically sound one should check two things: first, does this infinite sum converge for every $g$ and $n$ (in the physics papers this was only checked for those correlators which were actually computed), and, second, are the key properties of the correlators maintained?

The approach adopted here is as follows. First, the space of all possible $x$ with infinitely many ramification points is much too broad for the present work's techniques to tackle. To restrict the domain of curves down to one that can be solved, two classes of curves will be studied. The first are those where $x$ is a periodic function when written in some global coordinate, and such curves will be called \emph{periodic spectral curves}; this class will encompass the computations performed in the physics literature \cite{CEMR24,CEMR25,CER26,A25,DLY25,BERW26}. The second are those curves of the form $P(x,e^{xy})=0$, where $P$ is a polynomial,\footnote{Note that here it is assumed that $x$ and $y$ are parameterised \emph{without} any branch cuts, but \emph{with} essential singularities.} and such curves will be called \emph{exponential spectral curves}. The motivation for considering the second class is twofold: first, it provides many examples of higher genus curves with infinite ramification loci; second, circumstantial evidence will be presented that the quantisation of such curves is related to the quantisation of curves of the form $P(e^x,e^y) = 0$, a problem that is intimately related to knot theory \cite{ABM12,BE12,GJKS14,DPSS19} and mirror symmetry \cite{BKMP08,BM08,GS11,BS12,Z12,DOSS14,NS14,EO15,FLZ17,FLZ20,GKLS22}. 

Second, with the arena of action delineated, for every periodic/exponential curve $\mc{S}$, it will be argued that there exists a sequence of spectral curves $\mc{S}_N$, defined within the well-known Bouchard-Eynard formalism, that converges to $\mc{S}$ such that the correlators $\omega_{g,n}^N$ produced from $\mc{S}_N$ converge to the correlators of $\mc{S}$ in the limit as $N\to\infty$. From this one can immediately deduce the limiting correlators inherit the key properties of the $\omega_{g,n}^N$. In particular, employing the methodology developed in \cite{BE17,BKW24} one can qauntise many curves with an infinite ramification loci. Curiously, in some of the examples considered here one gets a quantisation that has already appeared in the literature, that is, a quantisation of a spectral curve where $x$ has only a finite number of ramification points. It must be stressed that this seeming coincidence is not a trivial consequence of one set of correlators being pullbacks of the others or any other more-or-less obvious relation; \textit{au contraire} it relies on seemingly non-trivial cancellations between different correlators of different $g$ and $n$ but the same $2g+n-2$. Indeed, such cancellations are explicitly checked for and found for different curves for low $g$ and $n$.

\subsection{Main results}

Two extensions of the topological recursion are proposed. The first is for spectral curves which are dubbed \emph{periodic}; these are curves of the form $\mc{S} = \left(\Sigma,x,y,\omega_{0,2}\right)$ where $\Sigma$ is an open subset of $\P^1$ and $x$ is a periodic function in some affine coordinate on $\P^1$. Second, are a class of curves dubbed \emph{exponential}; these are curves such that $x$ and $y$ satisfy identically a relation of the form $P(x,e^{\kappa xy}) = 0$, where $\kappa\in\C$ and $P$ is a polynomial.\footnote{These two cases are not disjoint; for example, the curve $ P(x,y) = -e^{(f+1)xy} + e^{fxy} - x $ for an integer $f\in\Z\setminus\{0,-1\}$, falls under both categories.} In both these cases, the difference with the usual Bouchard-Eynard formalism is the fact that the sum over ramification points becomes infinite (see: \Cref{d:per} and \Cref{d:expc}).

It is shown that, for both classes of curves, the sum over ramification points will always converge. Furthermore, for any spectral curve $\mc{S}$ in either class, the existence of a sequence of spectral curves $\mc{S}_N$, such that $\mc{S}_N$ is a spectral curve for which topological recursion has been previously defined and $\mc{S}_N\to\mc{S}$ as $N\to\infty$, is proven. Furthermore, denoting $\omega_{g,n}^N$ as the correlators constructed from $\mc{S}_N$, and $\omega_{g,n}$ the correlators constructed from $\mc{S}$, the relation
\begin{equation}
	\lim\limits_{N\to\infty}\omega_{g,n}^N = \omega_{g,n}\,,
\end{equation}
is established in \Cref{t:seqper,t:seqexp}. This allows for the deduction, in \Cref{c:perprop,c:expprop}, that the $\omega_{g,n}$ of $\mc{S}$ satisfy the key properties of the Bouchard-Eynard formalism; in particular, the correlators of \cite{CEMR24,CEMR25,CER26,A25,DLY25,BERW26} satisfy these very same properties.

Furthermore, the construction of the sequence $\mc{S}_N$ and the proof that the limiting correlators agree with that of $\mc{S}$ allows for the establishment of the quantum curve/topological recursion (QC/TR) connection for genus zero periodic and exponential curves. In addition to general results, two specific examples are studied and quantum curves are obtained.

First, consider the curve $x(z) = 2\cosh(z)$, $y(z) = z$ so that $x$ and $y$ satisfy $P(x,y) = x - 2\cosh(y) = 0$. The quantum curve is shown to be (\Cref{p:QCGWP1})
\begin{equation}
	\left[\hat{x} - 2\cosh(\hat{y})\right]\psi = 0\,,
\end{equation}
where $\hat{x} \coloneq x\cdot$ and $\hat{y}\coloneq \hslash d/dx$. Alternatively, the relation $P(x,y) = x - 2\cosh(y) = 0$ can be parametersied as $x=z+1/z$, $y(z) = \log(z)$; this parameterisation is known to calculate the Gromov-Witten invariants of $\P^1$ \cite{NS14}. Surprisingly, the quantum curve computed here agrees with the result computed for the second parametrisation \cite{M17,DMNPS17}.

The second example is the curve $x(z)=e^{fz}(1-{\rm e}^{z})$, $\omega_{0,1} = z/x(z)$ so $x$ and $y$ satisfy identically $P(x,y)=-e^{(f+1)xy}+e^{fxy}-x=0$, for an integer $f\in\Z\setminus\{0,-1\}$ called the framing. A quantum curve is computed for the first parametrisation (\Cref{p:QCGWtor})
\begin{equation}
	\left[-e^{(f+1)\hat{x}\hat{y}}+e^{f\hat{x}\hat{y}}-\hat{x}\right]\psi = 0\,.
\end{equation}
A different parametrisation of $P$, $x(z)=\log(z^f-z^{f+1})$, $y(z) = \log(z)$, is known to calculate the Gromov-Witten invariants of $\C^3$ \cite{BM08,BS12,Z12}. This does not agree with the quantum curve obtained for the second parametrisation in \cite{Z12} for any integer $f\in\Z$, but the two wavefunctions are shown to satisfy the simple relation 
\begin{equation}
  \psi(x) = (\pi^*\tilde{\psi})(e^{\hslash(f+1/2)}x)\,.
\end{equation}
where $\pi = \exp$ is the exponential map in $z$ and $\tilde{\psi}$ is the wavefunction for the parameterisation known to calculate the Gromov-Witten invariants of $\C^3$. If one relaxes the condition $f\in\Z$ and sets $f=-1/2$, then the two quantum curves and wavefunctions agree. Although half-integer framing was not considered in \cite{Z12}, half-integer framing has appeared in the physics literature \cite{DF05}.

In both cases where the quantum curves agree it is checked explicitly in \Cref{Comp} that the two wavefunctions agree up to $\mc{O}(\hslash^2)$.

\subsection{Outline}

The work begins with a review of the necessary background in \Cref{s:back}. Spectral curves and their various properties are defined and discussed, and the topological recursion of \cite{BBCCN24,BKW24} is explained. The relevant information and theorems about the QC/TR connection from \cite{BE17,BKW24} are also touched upon. This section primarily contains information that can already be found in the literature, with some minor new results and reformulations (there is one new result, \Cref{t:TR/QCess}, but this has an identical flavour to results that already appeared in \cite{BKW24}).

Periodic spectral curves are studied in \Cref{s:per}; in particular, how topological recursion can take such curves as initial data is explained rigorously. In \Cref{ss:defper}, topological recursion is defined on such curves, and it is shown that the sum over the infinitely many ramification points converges. Building upon this, \Cref{ss:propper} establishes that for any periodic curve, there exists a sequence of transalgebraic curves (transalgebraic in the sense of \cite{BKW24}; see \Cref{d:specprop}) that converge to the given periodic curve, and that topological recursion commutes with this limit. This theorem is used to establish general properties about periodic curves that mirror that of the Eynard-Orantin formalism, including the existence of quantum curves.

The structure of \Cref{s:exp} is analogous to that of \Cref{s:per}, except that \Cref{s:exp} deals with exponential, rather than periodic, spectral curves. Topological recursion on exponential curves is defined in \Cref{ss:defexp}, and the sum over ramification points is shown to converge. In \Cref{ss:propper} it is shown that, for any exponential curve there exists a sequence of algebraic curves converging to the given exponential curve and this limit commutes with the topological recursion. As in \Cref{s:per}, this is used to establish the key properties one would expect from the topological recursion formalism.

Next, the results of \Cref{s:per} and \Cref{s:exp} are put to work on examples in \Cref{s:app}. In each of \Cref{ss:P1,ss:C3}, a particular spectral curve is studied, and quantum curves are computed for each such spectral curve. The spectral curve of interest in \Cref{ss:P1} is
\begin{equation}
  \left(\P^1,\, x(z)=2\cosh(z),\, y(z)=z,\, \omega_{0,2}(z_1,z_2) = \frac{dz_1dz_2}{(z_1-z_2)^2}\right)\,,
\end{equation}
whereas in \Cref{ss:C3} the relevant curve is
\begin{equation}
  \left(\P^1,\, x(z)=e^{fz}(1-{\rm e}^{z}),\, y(z)=z/x(z),\, \omega_{0,2}(z_1,z_2) = \frac{dz_1dz_2}{(z_1-z_2)^2}\right)\,.
\end{equation}

In \Cref{s:pullback} the obvious notion of pullback is defined for a spectral curve $\mc{S}$,\footnote{\cite{E19} already defined pullbacks of spectral.} and various basic properties of this notion are demonstrated. Indeed, this notion is used to prove various relations between correlators of periodic/exponential curves and curves that fall under the traditional Bouchard-Eynard framework.

In the first appendix, \Cref{Npol}, the ideas of \Cref{s:pullback} are used to comment upon non-canonically polarised spectral curves, which have not yet appeared in the literature (see \Cref{d:specprop} for the definition of non-canonically polarised). Finally, the $\omega_{g,n}$ of topological recursion are computed for low values of $g,n$ and various periodic and exponential spectral curves in \Cref{Comp}.

\subsection{Relation to recent generalisations of TR}
The most obvious antecedent to the present work is the authors master's thesis \cite{W22}. In \cite{W22} many of the results of \Cref{s:per} are already present in some, albeit less polished\footnote{The author discourages anyone from actually reading \cite{W22}; the results are completed superseded by \cite{BKW24} and the present work.} and general, form. Furthermore, it has been brought to the author's attention that in the recent master's thesis \cite{D26} it was established rigorously that the sum over ramification points in the definition of TR converged for the spectral curve studied by physicists in \cite{CEMR24,CEMR25}; this is a direct presaging of \Cref{l:converge} and the results of \Cref{ss:propper} in one, quite important, special case.

Generalising TR in a seemingly different direction, \cite{ABDKS25} developed a new approach to topological recursion called \emph{generalised TR}, or Gen-TR for short. Here, the integrand is globally defined without the use of deck transformations (however, the integrand does involve more involved combinatorics). For curves that fall under the class originally studied in \cite{EO07}, Gen-TR and the original TR definition coincide. Howbeit, both Gen-TR and the earlier Bouchard-Eynard framework \cite{BE13} yield correlators for curves with higher ramification orders and these correlators do not always agree (see \cite{ABDKS25} for a discussion of this fact). The present work takes as its starting point the more traditional Bouchard-Eynard approach, rather than the Gen-TR approach. As the two frameworks give the same result for simple ramification points, and Gen-TR is very well behaved when taking limits of curves where ramification points collide, many of the results in the present work should be essentially the same if Gen-TR was taken as a starting point.

However, in \cite{AH26}, Gen-TR was used to define topological recursion for curves where $x$ or $y$ have essential singularities. Due to the fact that the integrand in Gen-TR is globally defined, one can flip the residue computation and sum over the poles of the integrand that aren't the ramification points of $x$. This insight was used by \cite{AH26} to avoid the essential singularities and define TR based on a different residue computation. Employing this strategy allows one to define topological recursion for many of the curves considered in the present work (see \cite{AH26}) and has the distinct advantage that the proofs (and often: computations) are much more elegant than the ones given here. However, the approach taken here has two points in its favour.
\begin{itemize}
  \item The naïve formula for computing TR correlators for curves with infinitley many ramification points is shown to be correct. As this is what is being done in the physics literature \cite{CEMR24,CEMR25,CER26,A25,DLY25,BERW26}, and it is \textit{a priori} unclear whether this is equivalent to the approach in \cite{AH26} (in the cases when the approach in \cite{AH26} is well-defined; see the next item), this is an important point.
  \item The approach of \cite{AH26} fails in cases when both $x$ and $y$ both have infinitely many ramification points, whereas the approach here still succeeds. This is not purely theoretical question: such curves find applications \cite{CEMR24,CEMR25,CER26,DLY25}.
\end{itemize}
For periodic curves where all the ramification points of $x$ and $y$ are simple and $y$ has finitely many ramification points, it is clear that the definition of TR used here will coincide with the one used in \cite{AH26} as both satisfy the $x$-$y$ duality (see \Cref{c:perx-y}) and the dual side is manifestly identical for both definitions.

Finally, \Cref{s:app} provides a partial answer to the question of quantum curves posed in the `Further directions' section of \cite{AH26}.

\subsection{Acknowledgements}

The author would like to thank Yifan Wang for useful discussions. The nascent portion of this work was completed while the author was a master's student of Vincent Bouchard's \cite{W22}, and without his invaluable guidance and mentorship the present work would surely never have come to fruition.

\section{Background}\label{s:back}
\subsection{The geometry of spectral curves}
Central to the notion of a spectral curve is the basic algebrogeometric notion of a branched covering. Consider a meromorphic function $x$ from a dense, open subset $D$ of a Riemann surface $\Sigma$ to $\P^1$, where it is assumed that there is no $D'\supset D$ such that $x$ can be meromorphically extended to $D'$. Near generic points, by the open mapping theorem, one expects $x$ to behave as a covering map. This expectation fails at two distinct classes of points \cite{I1914} (see also the introduction of \cite{BE95} for a brief explanation in English and more references).
\begin{enumerate}
	\item \emph{Critical Values:} points $x_0\in x(D)$ such that $\exists$ $z_0\in x^{-1}(x_0)$ with $dx(z_0)=0$.
	\item \emph{Asymptotic Values:} points $x_0 \in \P^1$ such that $\exists$ a continuous curve $\gamma:[0,1) \to D$ with $x\circ \gamma (t) \to x_0$ as $t\to 1^{-}$, but $\gamma(1^-) \notin D$.
\end{enumerate}
Collectively these two types of points are called \emph{Singular Values} and are denoted by $S = C \cup A \subset \P^1$, where $C$ denotes the critical values and $A$ the asymptotic values. In this work singular values shall usually be referred to as \emph{branchpoints}.

Then, the map $x|: D \setminus x^{-1}(S) \to x\left(D \setminus x^{-1}(S)\right),\, z\mapsto x(z)$ is an honest covering map, and, if $x(D \setminus x^{-1}(S))$ is dense in $\P^1$, $x$ itself shall be called a \emph{branched covering}. The set $R = R_\infty \sqcup R_0$, where $R_\infty = D^c$ and $R_0 = x^{-1}(S)$, is called the \emph{ramification locus}, and elements $a\in R$ are called \emph{ramification points}. The points in $R_\infty$ are called \emph{infinite}, whereas the points in $R_0$ are called \emph{finite}, ramification points; note that the points in $R_\infty$ correspond to the (not necessarily isolated) essential singularities of $x$. A point in $a\in R_0$ is said to lie over the branchpoint $x(a)$ and a point $a\in R_\infty$ is said to lie over the branchpoint $x\circ\gamma(1^{-})$, where $\gamma$ is a path as in the definition of asymptotic values such that $\gamma(1)=a$. Alternatively, the ramification locus $R$ can be decomposed into isolated and non-isolated ramification points via $R = R_I \sqcup R_N$, where $R_N$ are the non-isolated (limit points) of $R$ and $R_I$ are the isolated points of $R$; the following obvious relations hold: $R_N \subseteq R_\infty$; $R_0 \subset R_I$; $\overline{R}_I = R$.

Around a ramification point $a \in R$ define the positive integer $ r_a $ as the ramification order of $ x $ at $a$ if $a$ is a finite ramification point and put $r_a=\infty$ if $a$ is an infinite ramification point. For a point $ z \in \Sigma$ denote $ \mf{f}(z) = x^{-1}(x(z))$ as the fibre and $ \mf{f}'(z) = \mf{f}(z) \setminus \{ z\}$. Next, if $ a \in R_I$ and $ z$ is close to $ a$, write $ \mf{f}_a(z)$ for the local Galois conjugates of $ z$,\footnote{In particular, $z$ should be in an open neighbourhood $U$ of $a$ such that the natural action of the monodromy group on $ \mf{f}(z)$ restricts to a well-defined action on $\mf{f}_a(z)$.} and again $ \mf{f}'_a(z) = \mf{f}(z) \setminus \{ z\}$. If $a \in R_N$ is a non-isolated member of the ramification locus it will be considered, by convention, to have no local Galois conjugates. Note that $\mathfrak{f}_a(z)$ is always of cardinality $r_a$; in particular, for finite ramification points, it is finite. 

On a small open neighbourhood $U_a$ of each ramification point $a\in R_0$ one can define a local coordinate $\zeta_a$ through the relation $x = x(a) + \zeta_a^{r_a} $ if $ x (a) \neq \infty $ and through the relation $x =  \zeta_a^{-r_a} $ if $ x(a) = \infty$. In this coordinate the action of a local deck transformation $\sigma_a$ amounts to multiplication by an $r_a$th root of unity $\zeta\mapsto e^{2\pi im/r_a}\zeta$ where $m = 0,1\dots,r_a-1$.

Now examine a point $a\in \R_{\infty}\setminus R_N$; it was argued in \cite[Proposition 3.3]{BKW24} that for any such isolated infinite ramification point one may choose an open neighbourhood $U_a$ such that $x|_{U_a}$ has precisely two lacunary values $p_1,p_2\in\P^1$. Defining an affine coordinate $t$ on $\P^1$ such that $t(p_1)=0=1/t(p_2)$, one may then define the \emph{exponential order} of $x$ at $a$ as 
\begin{equation}
\Erd_x(a) \coloneq \inf\left\{ \rho>0\ \big|\ \log|t(x(z))|=\mc{O}\big(|z-a|^{-\rho}\big),\, z\to a \right\}\,. 
\end{equation}
Letting $m_0 \coloneq \Res_{a} d\log(x) $ and $ m_1 \coloneq \Erd_x(a) $ one can define a local coordinate $\zeta_a$ through $ t(x) = \zeta_a^{m_0}\exp(-(m_0/m_1)\zeta_a^{-m_1}) $ if $m_0\neq 0$ or $t(x)=\exp(\zeta_a^{-m_1})$ if $m_0=0$ \cite{BKW24}. In the coordinate $\zeta_a$, the local deck transformations can be solved for in terms of elementary functions if $m_0=0$ and in terms of the Lambert $W$ function if $m_0\neq 0$. For more information on the local geometry around these points refer to \cite{BKW24} and the references contained therein.

\begin{remark}\label{r:erd_fin}
	The astute reader may have noticed that it was implicitly assumed that $\Erd_x(a)$ was finite. The fibres of functions with $\Erd_x(a) = \infty$ are extraordinarily complicated and it seems exceedingly unlikely to the author that they will be pertinent to the topological recursion; therefore, in this work, such pathologies will be excluded by assumption.
\end{remark} 

Next, in a small open neighbourhoods $U_a$ of each ramification point $a$, it will be necessary to imbue the space of locally exact meromorphic $1$-forms
\begin{equation}
	V_a \coloneq \{df \,|\, f\in \mc{M}(U_a)\}\,,
\end{equation}
with a symplectic structure. To do so define the symplectic pairing
\begin{equation}
	\Omega_a(\omega_1,\omega_2) \coloneq \Res_{a}\omega_2\int\omega_1\,,
\end{equation}
and note that the space of holomorphic $1$-forms on $U_a$
\begin{equation}
	V_a^+ \coloneq \{\omega\in V_a \,|\, \lim\limits_{a} \int\omega < \infty\}\subset V_a\,,
\end{equation}
is then a Lagrangian subspace of $V_a$. The next step is to extend this structure to a set of ramification points $A \subset R$. Begin by making the following definitions
\begin{equation}
	V\coloneq \bigoplus_{a\in A} V_a,\quad \Omega \coloneq \sum_{a\in A} \Omega_a,\quad V^+\coloneq \bigoplus_{a\in A} V_a^+\,. 
\end{equation}
Now, $V$ is a symplectic vector space with symplectic form $\Omega$, and $V^+$ is a Lagrangian subspace. A \emph{choice of polarisation} corresponds to a choice of Lagrangian subspace $V^-$ complimentary to $V^+$.\footnote{It is important to note, for each $a\in A$, $V_a^-$, $V_a^+$, and $V_a$ are recoverable just from the data of $V^-$. This is because $V_a$ and $V_a^+$ depend on $V^-$ only implicitly through $A$, and $V_a^- = V^-/\sim_a$, where $\omega_1\sim_a\omega_2$ iff $\omega_1$ and $\omega_2$ define the same germ at $a$.} However, the most convenient way to encode a choice of polarisation, for the purposes of topological recursion, is through a symmetric bidifferential.

To this end, start by assuming $A\subset R_0$. Take the local coordinate $\zeta_a$ on each $U_a$\footnote{as $R_0$ is a discrete set, it is now assumed that these open neighbourhoods are chosen to be disjoint, \textit{i.e.} $U_a\cap U_b=\emptyset$ for all $b\in A\setminus\{a\}$.} Then, for $l\in\Z_{>0}$ define the $1$-forms
\begin{equation}
	d\xi_a^l(z) \coloneq \bs{1}_{U_a}(z)\zeta_a^{l-1}(z)d\zeta_a(z)\,,
\end{equation}
where $\bs{1}_{U_a}$ is the indicator function for the set $U_a$. The space of holomorphic $1$-forms on $U_a$ is given by $V_a^+ \coloneq \spa \{d\xi^{l}_a \,|\, l\in\Z_{>0}\}$. Then, for each $a\in A\subseteq R$ choose $1$-forms $d\xi_a^l$ with $l\in\mathbb{Z}_{< 0}$
\begin{equation}
	d\xi_a^{-l}(z) \coloneq \bs{1}_{U_a}(z)\zeta_a^{-l-1}(z)d\zeta_a(z) + \frac{1}{l}\sum_{b\in A}\sum_{k=0}^{\infty} \phi_{l,k}^{a,b} d\xi^b_k\,,
\end{equation}
which completes the basis of $V_a^+$ into a basis of $V_a = \spa\{d\xi_a^l \,|\, l\in\Z\setminus\{0\}\}$; furthermore $V = \spa\{d\xi_a^l \,|\, l\in\Z\setminus\{0\} \wedge a \in A\}$. Notice that, if the coefficients $\phi_{l,k}^{a,b}$ are chosen to be symmetric in $l,k$
\begin{equation}
	\Omega(d\xi_a^l,d\xi_b^{k}) = \frac{1}{l}\delta_{a,b}\delta_{l+k,0}\,,
\end{equation}
so the space $V^- = \spa\{d\xi_a^l \,|\, l\in\Z_{<0} \wedge a \in A\}$ is a Lagrangian subspace complimentary to $V^+$. Indeed, the choice of polarisation can then be encoded in a symmetric bidifferential
\begin{equation}
	B(z_1,z_2) = \sum_{a\in A}\left(\frac{\bs{1}_{U_a}(z_1)\bs{1}_{U_a}(z_2)d\zeta_a(z_1) d\zeta_a(z_2)}{( \zeta_a(z_1)-\zeta_a(z_2) )^2} + \sum_{b\in A}\sum_{l,k=0}^{\infty} \phi_{l,k}^{a,b} d\xi_a^l(z_1) d\xi_b^k(z_2)\right)\,,\footnote{For notational succinctness, the tensor product between differentials will be suppressed.}
\end{equation}
where the basis of $V^-$ can be recovered via
\begin{equation}
	d\xi_a^{-l}(z) = \Omega\left(B(z,\cdot),\zeta^{-l-1}(\cdot)d\zeta(\cdot)\right)\,,\qquad l>0\,.
\end{equation}

For an infinite ramification point, such a simple local picture may not exist. However, one can still encode the choice of polarisation in a bidifferential
\begin{equation}\label{e:bidiff}
	B \in \Sym^2 V^+ + \sum_{a\in A}\frac{\bs{1}_{U_a}(z_1)\bs{1}_{U_a}(z_2)d\zeta_a(z_1) d\zeta_a(z_2)}{( \zeta_a(z_1)-\zeta_a(z_2) )^2}\,,
\end{equation}
where $\Sym^n$ is the $n$-fold symmetric tensor product\footnote{Unless otherwise stated, all the tensor products are assumed to be completed.} and $\zeta_a$ is any coordinate choice on $U_a$ (the resulting space is independent of this choice), by defining $V^-$ to be the invariant subspace of the projection
\begin{equation}
	P: V \to V, \, \omega\mapsto \Omega(B,\omega).
\end{equation}
Thus, it is evident that there are three equivalent ways of specifying a choice of polarisation:
\begin{itemize}
\item a choice of Lagrangian subspace $V^-\subset V$ that is complimentary to $V^+$;
\item a choice of symmetric bidifferential $B$, locally taking the form \eqref{e:bidiff}, along with the set $A$;
\item a choice of projection $P$ with $\ker P = V^+$.
\end{itemize}
It will be most convenient to specify a polarisation using a symmetric bidifferential $B$. Of particular interest will be those $B$ that can be considered as an element of $\Sym^2\Omega^1(\Sigma\setminus R_N)$ that locally takes the form \eqref{e:bidiff} whenever its two arguments are near elements of $A$. This is formalised in the following definition.

\begin{definition}
	Let $\Sigma = \sqcup_\alpha \Sigma_\alpha$ be a finite disjoint union of Riemann surfaces and let $x$ be a branched covering from $\Sigma$ to $\P^1$. A \emph{choice of polarisation} $(A,B)$ is a collection of ramification points $A\subseteq R$, subject to $\overline{R\setminus A}\cap A = \emptyset$, along with a residueless bidifferential $B\in\Sym^2\Omega^1(\Sigma\setminus R_N)$ such that, for $a_1,a_2\in A\cap R_I$
	\begin{equation}
		B(z_1,z_2)\bs{1}_{U_{a_1}}(z_1)\bs{1}_{U_{a_2}}(z_2) \in \Sym^2 V^+ + \sum_{a\in A}\frac{\bs{1}_{U_a}(z_1)\bs{1}_{U_a}(z_2)d\zeta_a(z_1) d\zeta_a(z_2)}{( \zeta_a(z_1)-\zeta_a(z_2) )^2}\,,
	\end{equation}
	\textit{i.e.} $B$ locally restricts to the form \eqref{e:bidiff}.
\end{definition}

Equipped with this knowledge of polarisation, it is time to define what a spectral curve is. 
 \begin{definition} \label{d:SC}
	A \emph{spectral curve} is a quintuple $ \mc{S} = (\Sigma,\, x,\, (A,B),\, y,\, \omega_{0,2})$, where:
	\begin{enumerate}
		\item $\Sigma = \sqcup_\alpha \Sigma_\alpha$ is a finite disjoint union of Riemann surfaces;
		\item $x$ is a branched covering from $\Sigma$ to $\P^1$;
		\item $(A,B)$ is a choice of polarisation;
		\item $y$ is a function on $\Sigma\setminus R_N$ such that, for each local Galois transformation $\sigma_a$ about a ramification point $a\in A$
		\begin{equation}
			\left(y-\sigma_a^*y\right)\bs{1}_{U_a}dx \in V_a^+\setminus\{0_{V_a^+}\}\,;
		\end{equation}
		\item $\omega_{0,2}$ is a symmetric residueless bidifferential on $\Sigma^2$ with the following two properties:
		\begin{itemize}
			\item for any ramification point $a \in A$ $\omega_{0,2}(z_1,z_2)$ may be approximated by a sequence of elements in the uncompleted tensor product $V^+ \otimes V^-$ converging compactly on the subset of $U_a^2$ defined by $|z_1-a|<|z_2-a|$;
			\item for any ramification point $a\in R_I$, and any two distinct local Galois transformations $\sigma_{a,1}$ and $\sigma_{a,2}$, the mapping
			\begin{equation}
				z\mapsto \frac{\omega_{0,2}(\sigma_{a,1}(z),\sigma_{a,2}(z))}{d\sigma_{a,1}(z)d\sigma_{a,2}(z)}\,,
			\end{equation}
			must be a well-defined meromorphic function on $U_a$.
		\end{itemize}
	\end{enumerate}
\end{definition}

In addition to the defining properties above, spectral curves can have a myriad of additional properties that will facilitate nicer results. Many of these properties are listed in the following definition.
\begin{definition}[properties of spectral curves]\label{d:specprop}
	A spectral curve $\mathcal{S} = (\Sigma,\, x,\, (A,B),\, y,\, \omega_{0,2})$ is called:
	\begin{itemize}
		\item \emph{condensed} if $|A|<\infty$ (notice this implies $A \cap R_N = \emptyset$);
		\item \emph{quasi-meromorphic} if it is condensed and for every point $a\in A\cap R_\infty$, it follows that $\Erd_x(a) = 1$.
		\item \emph{meromorphic} if it is quasi-meromorphic and $A\subseteq R_0$ (so $x$ defines a meromorphic function on each $U_a$);
		\item \emph{simple} if $a\in A \Rightarrow r_a=2$;
		\item \emph{connected} if $\Sigma$ is connected;
		\item \emph{compact} if $\Sigma$ is compact;
		\item \emph{finite} if for all $z\in\mathbb{P}^1$, we have $|x^{-1}(z)|<\infty$;
		\item \emph{globally polarised} if the only pole of $B$ is the double pole on the diagonal and this double pole has biresidue one;
		\item \emph{canonically polarised} if $B=\omega_{0,2}$ and $A=R$;
		\item \emph{global} if the curve is globally and canonically polarised, compact, connected, and $A = R$;
		\item \emph{algebraic} if it is global and finite;
		\item \emph{transalgebraic} if it is global and condensed, but not algebraic.
	\end{itemize}
\end{definition}

\begin{remark}
	A spectral curve may given as a quadruple $\mathcal{S} = (\Sigma,\, x,\, y,\, \omega_{0,2})$ rather than a quintuple $\mathcal{S} = (\Sigma,\, x,\, (A,B),\, y,\, \omega_{0,2})$. If so, it is implicitly assumed that locally $\omega_{0,2}$ takes the form \eqref{e:bidiff} and the choice of polarisation is given by $(A,B) = (R,\omega_{0,2})$, \textit{i.e.} that the curve is canonically polarised.
\end{remark}

\begin{remark}
	The presentation of spectral curves given here closely resembles that of \cite{BBCCN24} with a few differences. First, in \cite{BBCCN24}, all objects were defined as formal expansions in local coordinates. This, strictly speaking, results in slightly more generality as the expansions need not converge. However, the reason such local coordinates were always available was the fact that \cite{BBCCN24} only considered meromorphic curves. Furthermore, all curves in \cite{BBCCN24} were canonically polarised. 
\end{remark}

\begin{remark}
	Although, to the knowledge of the author, non-canonically polarised curves have not previously appeared in the literature, these curves are not the focus of the present work. Rather, they appear almost accidentally when it becomes necessary to consider certain transformations of spectral curves. Although several examples of curves that are non-canonically polarised but for which the topological recursion is still well-defined will be given, the exposition of a general theory of non-canonically polarised curves will not be undertaken.
\end{remark}

The topological recursion was first discovered and studied for simple algebraic spectral curves in \cite{EO07}; this has come to be refereed to as the Eynard-Orantin topological recursion. The generalisation to simple, canonically polarised curves is straightforward from the work of \cite{EO07}. The study of general meromorphic curves that were also canonically polarised was conducted in \cite{BE13}, and the resulting formalism is referred to as the Bouchard-Eynard topological recursion. Recently, in \cite{BKW24}, the assumption that the curve must be meromorphic was dropped and transalgebraic curves were studied.\footnote{Although it wasn't done in \cite{BKW24}, the topological recursion is essentially local in nature, and so the story of general condensed and canonically polarised curves is almost certainly entirely similar to the presentation in \cite{BKW24}.} Here, the requirement that the curve be condensed is loosened and the topological recursion generalised to two new categories of curves: \emph{periodic} curves and \emph{exponential} curves. 

\begin{definition}\label{d:per}
	 A spectral curve $ \mc{S} = (\Sigma,\, x,\, (A,B), y,\, \omega_{0,2})$ will be referred to as \emph{periodic} if the following conditions hold\footnote{The author is unsure if the second condition follows from the first and third.}
	 \begin{enumerate}
	 	\item $\Sigma$ is genus zero and there exists an affine coordinate $z$ on $ \Sigma \cong \mathbb{P}^1 \cong (\C\cup\{\infty\}) $ such that the meromorphic function $x$, when written in this coordinate, is a periodic entire function;
	 	\item As an entire function, $x$ has finite order;\footnote{The \emph{order} of an entire function $x$ is defined as the infimum of all $\rho>0$ such that $x(z)=\mathcal{O}\big(\exp(|z|^\rho)\big)$. This is just the exponential order of $x$ at infinity.}
	 	\item denoting $p$ as the period of $x$ in coordinate $z$, we have that $x$ has finitely many ramification points on the set $\C/p\Z$.
	 \end{enumerate}
\end{definition}

\begin{definition}\label{d:expc}
	 A global spectral curve $ \mc{S} = (\Sigma,\, x,\, y,\, \omega_{0,2})$ will be referred to as \emph{exponential} if $x$ and $y$ satisfy identically an equation $P(x, e^{\kappa xy}) = 0$ where $P$ is an irreducible polynomial and $\kappa\in\C^*$.
\end{definition}

Below are some examples of spectral curves which will be given starring roles in the present work. All of them are connected.


\begin{example}\label{ex:GWnorm}
	Consider the spectral curve $\mathcal{S}$ with $\Sigma = \mathbb{C} \setminus [0,-i\infty]$,\footnote{Oftentimes, curves like this are sloppily written without properly specifying the branch cut needed to define the logarithm. For this curve, one should require \cite{NS14}
	\begin{equation}
		y(z) = \log(z) = -\frac{1}{2}\sum_{k=0}^{\infty}\frac{(1-z^2)^k}{k}\,,\quad |1-z^2|<1\,.
	\end{equation}
} $x = z+z^{-1}$, $y = \log(z)$, and $B = \omega_{0,2} = \frac{d z_1 dz_2}{(z_1-z_2)^2}$. $x$ and $y$ satisfy identically the equation
	\begin{equation}
		x-2\cosh(y)=0\,,
	\end{equation}
	which is not algebraic despite the fact that $x$ is a finite degree branched covering. Despite not being algebraic, this curve is meromorphic, simple, and finite. It is known that the $\omega_{g,n}$ of this curve are generating functions of the Gromov-Witten invariants of the complex projective line \cite{Z12,NS14,DOSS14}.	
\end{example}

\begin{example}\label{ex:GWus}
	Consider the spectral curve $\mathcal{S}$ with $\Sigma = \mathbb{P}^1$, $x = 2\cosh(z)$, $y = z$, and $B = \omega_{0,2} = \frac{d z_1 dz_2}{(z_1-z_2)^2}$. $x$ and $y$ satisfy identically the equation
	\begin{equation}
		x-2\cosh(y)=0\,.
	\end{equation}
	which is the same equation as in the prior example despite the fact $x$, $y$, and $\Sigma$ contain significant differences.\footnote{It is important to note that, as $\omega_{0,2}$ is the same for both curves, this curve is \emph{not} just the pullback of the prior curve.} While the former example was meromorphic and finite, this example is periodic and global; it is also a non-condensed spectral curve where $R_N=\{z=\infty\}$ is non-empty. 
\end{example}


\begin{example}\label{ex:GWtorus}
	Consider the spectral curve $\mathcal{S}$ with $\Sigma = \mathbb{P}^1$, $x = e^{fz}(1-e^z)$, $y = z / x$, and $B=\omega_{0,2} = \frac{d z_1 dz_2}{(z_1-z_2)^2}$ where $f \in \mathbb{Z}\setminus\{0,-1\}$ is an integer called the framing. Note that $x$ and $y$ satisfy identically the equation
	\begin{equation}
		-e^{(f+1)xy}+e^{fxy}-x=0\,,
	\end{equation}
	which can be obtained from $-e^{(f+1)y}+e^{fy}-e^{x}=0$ via the transformation $(x,y) \mapsto (\log(x),xy)$. This curve is periodic, and, as we can see from the polynomial equation, it is also an exponential curve.
\end{example}

\begin{example}\label{ex:Lstring}
  Consider the spectral curve $\mc{S}$ with $\Sigma = \mathbb{P}^1$, $x = -2\cos(\pi b^{-1} z)$, $y = 2\cos(\pi b z)$, $A = b\Z\setminus\{0\}$, $B = \frac{d z_1 dz_2}{(z_1-z_2)^2}$, $\omega_{0,2} = \frac{4z_1z_2dz_1dz_2}{(z_1^2-z_2^2)^2}$ where $b\in\C\setminus (\R\cup i\R)$. This curve is periodic, but, unlike the previous examples, is not canonically polarised. This curve calculates the string amplitudes of the complex Liouville String when $b\in e^{i\pi/4}\R$. \cite{CEMR24,CEMR25}.\footnote{The spectral curve is written here in a form slightly differently from that of \cite{CEMR25}. That the two forms calculate the same correlators is a consequence of the results in \Cref{s:pullback}.}
\end{example}

\begin{example}\label{ex:JTdual}
	Consider the spectral curve $\mc{S}$ with $\Sigma = \mathbb{P}^1$, $x = (4\pi)^{-1}\sin(2\pi z)$, $y = z^2$, $B = \omega_{0,2} = \frac{d z_1 dz_2}{(z_1-z_2)^2}$ where $b\in\C\setminus\{0,1,-1\}$. This curve is periodic, and is related to JT gravity via the $x$-$y$ swap \cite{A25}. It is argued in \Cref{ss:propexp} that the $x$-$y$ swap extends to periodic curves.
\end{example}

\subsection{Systems of correlators and topological recursion}

Begin by defining the notion of a system of correlators.
\begin{definition}
	Given a spectral curve $ \mc{S} = (\Sigma,\, x,\, y,\, (A,B),\, \omega_{0,2}) $ a \emph{system of correlators} is a collection of symmetric multidifferentials, $\{\omega_{g,n}\}_{(g,n)\in\Z_{\geq 0}\times\Z_{\geq 1}}$, on $\Sigma^n$, called \emph{correlators}. The correlators $\omega_{0,1} \coloneq ydx$ and $\omega_{0,2}$ are given by $\mc{S}$ and are referred to as \emph{unstable}, whereas the rest of the correlators are referred to as \emph{stable}.
\end{definition}

Pause now to establish some prerequisite notation. Given an index set $ N $ denote $ z_{N}\coloneq\{z_n|n\in N\}\subset\Sigma $ as a collection of $|N|$ points on $\Sigma$ indexed by the elements of $N$. Then define, for $n \in \mathbb{Z}_{\geq 1}$, the set $\llbracket n \rrbracket\coloneq\{1,2,\dots,n\}$ so, for example, $z_{\llbracket n \rrbracket} = \{z_1,\dots,z_n\}$.

\begin{definition}
	Given a spectral curve $ \mc{S} = (\Sigma,\, x,\, y,\, (A,B),\, \omega_{0,2}) $ and a system of correlators $\{\omega_{g,n}\}_{(g,n)\in\Z_{\geq 0}\times\Z_{>0}}$ on $\mc{S}$, then the system of correlators is said to be \emph{polarised}\footnote{Sometimes in the literature polarisation is referred to by saying the system of correlators satisfy the projection property.} if all stable correlators $\omega_{g,n}$ satisfy
	\begin{equation}
		\omega_{g,n+1}(z_0,z_{\llbracket n \rrbracket}) = \Omega\big(B(z_0,\cdot),\omega_{g,n+1}(\cdot,z_{\llbracket n \rrbracket})\big) = \sum_{a\in A} \Res_{z=a} \omega_{g,n+1}(z,z_{\llbracket n \rrbracket}) \int_a^z B(z_0,\cdot)\,.
	\end{equation}
\end{definition}

\begin{definition}\label{d:combcomb}
	Given a system of correlators $\{\omega_{g,n}\}_{(g,n)\in\Z_{\geq 0}\times\Z_{>0}}$ on a spectral curve $\mc{S}$ define
	\begin{equation}
		\begin{split}
			\mc{E}_{g,n,|Z|} (Z \mid z_{\llbracket n \rrbracket}) &\coloneqq \sum_{\substack{\mu \vdash Z\\ \bigsqcup_{k=1}^{\ell (\mu)} N_k = \llbracket n \rrbracket\\ \sum g_k = g + \ell (\mu) - |Z|}} \prod_{k=1}^{\ell (\mu)} \omega_{g_k,|\mu_k| + |N_k|}(\mu_k, z_{N_k})\,, \\
			\mc{W}_{g,n,|Z|} (Z \mid z_{\llbracket n \rrbracket}) &\coloneqq \sum_{\substack{\mu \vdash Z\\ \bigsqcup_{k=1}^{\ell (\mu)} N_k = \llbracket n \rrbracket\\ \sum g_k = g + \ell (\mu) - |Z| }}' \prod_{k=1}^{\ell (\mu)} \omega_{g_k,|\mu_k| + |N_k|}(\mu_k, z_{N_k})\,,
		\end{split}
	\end{equation}
	where the prime on the summation indicates that one should omit any term containing a factor of $ \omega_{0,1}$, the notation $\mu \vdash Z$ means that $\mu$ is a set partition of a set $Z$, and $l(\mu)$ is used to denote the length of the partition $\mu$. 
\end{definition}

With this definition one can define the notion of higher abstract loop equations. Solving these equations is what leads to topological recursion.

\begin{definition}\label{d:loop}
	Given a meromorphic spectral curve $\mc{S} = (\Sigma,\, x,\, y,\, (A,B),\, \omega_{0,2})$, a system of correlators $\{\omega_{g,n}\}_{(g,n)\in\Z_{\geq 0}\times\Z_{>0}}$ on $\mc{S}$ is said to satisfy the \emph{higher abstract loop equations} if for every $a\in A$, $(g,n)\in\Z_{\geq 0}\times\Z_{\geq 1}$ with $2g+n-2>0$, and $i\in\Z_{\geq 1}$
	\begin{equation}
		\sum_{\substack{Z\subseteq f_a(z) \\ |Z|=i}} \mc{E}_{g,n,i} (Z \mid z_{\llbracket n \rrbracket}) = \mc{O}\left(x(z)^{-\mf{d}_a(i)}\big(dx(z)\big)^{i}\right)\,,\quad z\to a \, ,
	\end{equation}
	where,
	\begin{equation}
		\mf{d}_a(i) \coloneq i-1-\left\lfloor \frac{s_a(i-1)}{r_a} \right\rfloor\,,\qquad s_a \coloneq 1 + \max\left\{ \Ord_{(\omega_{0,1}-\sigma_a^*\omega_{0,1})}(a) \, | \, \sigma_a(z)\in f_a(z) \right\}\,.
	\end{equation}
	The notation $\sigma_a(z) \in f_a(z)$ means that $\sigma_a$ is a local deck transformation about $a$; the maximum is taken over all such local deck transformations. Notice that, by \Cref{d:SC}, $s_a$ must be a positive number.
\end{definition}

Now ask the following question: given a spectral curve $\mc{S}$ does there exists a polarised system of correlators that satisfy the higher abstract loop equations\footnote{For a non-meromorphic curve, it is unknown how to formulate the higher abstract loop equations so one would have to ask some as-of-yet undiscovered generalisation of this question. See \cite{BKW24} for more details and speculation.} and, if there is such a system, is it unique? The answer is provided by, at least in part, the topological recursion.

\begin{definition} \label{d:CndTR}
	Given a condensed quasi-meromorphic spectral curve $\mathcal{S} = (\Sigma,\, x,\, y,\, (A,B),\, \omega_{0,2})$, the \emph{topological recursion} is a procedure to define a polarised system of correlators $ \{ \omega_{g,n} \}_{g\geq 0, n \geq 1}$, recursive on $ -\chi_{g,n}=2g-2+n$, as follows
	\begin{equation} \label{TR-RecursiveStep}
		\omega_{g,n+1}(z_0, z_{\llbracket n\rrbracket}) \coloneqq \sum_{a \in A \cap R_0} \Res_{z = a}\sum_{\emptyset \neq Z \subseteq \mf{f}'_a(z)} K_{|Z|+1}(z_0,z, Z) \mc{W}_{g,n,|Z|+1}(z, Z \mid z_{\llbracket n \rrbracket})\,,
	\end{equation}
	where
	\begin{equation}
		K_{|Z|+1} (z_0, z, Z )  \coloneqq \frac{\int_{a}^z B(z_0, \cdot)}{\prod_{z' \in Z} (y (z') -y(z))dx(z)}\,,
	\end{equation}
	is the so-called \emph{recursion kernel}.
\end{definition}

\begin{theorem}\label{t:loop}
	Given a meromorphic canonically and globally polarised spectral curve $\mc{S}$, if there exists a polarised system of correlators satisfying the higher abstract loop equations, then this solution is constructed by the topological recursion, and therefore is unique. Conversely, if the topological recursion constructs a system of correlators, then these correlators are polarised and satisfy the higher abstract loop equations.
\end{theorem}
\begin{proof}
	See \cite{BBCCN24}.
\end{proof}

\begin{remark}
	The proof of the non-canonically polarised case should carry over \textit{mutatis mutandis} to this more general case as the required properties of $\omega_{0,2}$ are the ones listed in \Cref{d:SC}.		
\end{remark}

Ignoring the ramification points in $R_\infty$ is justified as follows. In \cite[Corollary 4.11]{BKW24} it was argued, given a transalgebraic spectral curve and an essential singularity $a$ of $x$, if $\Erd_x(a) = 1$ then the ramification point $a$ may be ignored for the purposes of topological recursion. Furthermore, the proof of this fact only relied upon the limiting behaviour of certain local loop equations and so should be valid in a more general setting than the transalgebraic case. The definition of the topological recursion at essential singularities is more involved when $\Erd_x(a) \geq 2$ and the reader is invited to look at \cite{BKW24} for details. For both exponential and periodic curves it will be demonstrated that $\Erd_x(a) \leq 1$ (for exponential curves this is immediately clear from the definition).

It is important to realise that the topological recursion will always construct polarised multidifferentials. The question of whether these multidifferentials form a system of correlators is a question of whether these polarised multidifferentials are symmetric. Indeed, in \Cref{d:CndTR} the variable $z_0$ plays a special role and it is a priori unclear how symmetric multi-differentials will be produced from this equation. Sufficient conditions for symmetry are commonly referred to as admissibility.

\begin{definition}\label{d:CondSCAdm}
	A condensed, quasi-meromorphic spectral curve is \emph{admissible} if: 
	\begin{enumerate}
		\item for every $a\in A$ either $s_a \leq -1$, or $1 \leq s_a \leq r_a + 1$ with $ r_a = \pm 1 \pmod{s_a}$;
		\item for every point $a\in R_\infty$ the meromorphic function $xy$ has a pole at $a$;
		\item the curve is globally and canonically polarised.
	\end{enumerate}
\end{definition}

The following theorem lists some well-known properties of the correlators.

\begin{theorem}\label{t:origprop}
	The correlators $\omega_{g,n}$ constructed by the topological recursion from an admissible spectral curve $\mc{S}$ defined by \Cref{d:TRinflocus} enjoy the following properties.
	\begin{itemize}
		\item \emph{Symmetry:} the $\omega_{g,n}$ are symmetric in their $n$ variables;\\
		\item \emph{Residueless:} for $2g+n-2\geq 0$ the $\omega_{g,n}$ have vanishing residue; \\
		\item \emph{Homogeneity:} under the rescaling $y \mapsto fy$ it follows that $\omega_{g,n} \mapsto f^{2-2g-n} \omega_{g,n}$;\\
		\item \emph{Pole Structure:} if $\mc{S}$ is globally polarised, then the $\omega_{g,n}$ only have poles in $R_0$, and for every $a\in R_0$ and $2g+n-2\geq 0$, $\omega_{g,n}$ will have a pole of order no more than $(s_a-1)(2g+n-2)+2g$ at $a$ in each of its $n$ variables, where $s_a$ is defined in \Cref{d:CondSCAdm}; \\
		\item \emph{Formula for $\omega_{0,3}$:} the simplest stable correlator $\omega_{0,3}$ is given by the following formula
		\begin{equation}
			\omega_{0,3}(z_1,z_2,z_3) = \sum_{a\in A \cap R_0} \Res_{z=a} \frac{ \omega_{0,2}(z,z_1) \omega_{0,2}(z,z_2) \omega_{0,2}(z,z_3) } { dx(z) dy(z) }\,.
		\end{equation}
	\end{itemize}
\end{theorem}
\begin{proof}
	See \cite{EO07,BE13,BBCCN24,BKW24}.
\end{proof}

\subsection{Quantum curves}

%

Topological recursion first appeared in the context of $1$-matrix models \cite{E04}, which are zero dimensional quantum field theories with one field (a matrix) where gauge transformations correspond to unitary change of basis transformations of the matrix. Observables, therefore, are expectation values of basis independent objects associated to the matrix: the multidifferentials $ \omega_{g,n} $, for example, are generating functions for the large-$N$ expansion of expectation values of traces \cite{E04}. 

But the trace is only one basis-independent object one can form from a matrix. Arguably, the most fundamental basis-independent object one can associate to a matrix is its characteristic polynomial. The expectation value of the characteristic polynomial is, perforce, an object of considerable interest in a matrix model. Indeed, this expectation value enjoys intimate connections with the $\omega_{g,n}$ themselves through what are known as determinantal formulae \cite{M90}. In the context of topological recursion, one expression of this relation is the topological recursion/quantum curve connection.

To be more precise, let $ \mathcal{S} = \left( \Sigma,\, x,\, y,\, B\right) $ be a global spectral curve (assume  here that $\Sigma$ is genus zero for simplicity). The functions $x$ and $y$ then satisfy a relation $P(x,y)=0$; if the spectral curve is algebraic then $P$ is a polynomial, but more generally it may be an entire function.

Define the \emph{wavefunction} $\psi$ (in a matrix model $\psi$ is proportional to an asymptotic expansion of the expectation value of the characteristic polynomial) associated to the spectral curve $\mathcal{S}$ as
\begin{equation}\label{e:wfuncconj}
	\psi(x(z);b) \coloneq \exp\left[ \sum_{n=1}^{\infty} \sum_{g=0}^{\infty} \frac{\hslash^{2g+n-2}}{n!} \int^z_b\cdots\int^z_b \left( \omega_{g,n}  - \delta_{g,0} \delta_{n,2} \frac{dx(z_1) dx(z_2)}{(x(z_1) - x(z_2))^2} \right) \right] \,,
\end{equation}
here $ \hslash $ is a formal expansion parameter, there are $n$ integrations in each term, $b\in\Sigma$ is some basepoint. Whether $\psi$ is written as function of $x$ or $z$ is conventional: the physics motivation leans towards the former, but with the latter $\psi$ is single-valued. The exact nature of the integration should be defined carefully (see \cite{BE17}). Then, a \emph{quantisation} of the spectral curve $\mc{S}$ is an operator $\hat{P}$ satisfying the following definition.

\begin{definition}\label{d:Quant}
	Let $\mathcal{S}$ be a global spectral curve, with $x$ and $y$ satisfying the relation $ P(x,y) = 0 $. $\hat{P}(\hat{x},\hat{y};\hslash) $ is called a quantisation of $P(x,y)$ if the following expansion exists for some $ m\in\mathbb{N}\cup\{\infty\} $:
	\begin{equation}
		\hat{P}(\hat{x},\hat{y};\hslash) = P(\hat{x},\hat{y}) + \sum_{i=1}^{m} \hslash^i\hat{P}_i(\hat{x},\hat{y})\,,
	\end{equation}
	where $P(\hat{x},\hat{y})$ is taken to be normally ordered (in each term all the $\hat{x}$ are put to the left of the $\hat{y}$) and the $ \hat{P}_i $ are  normal ordered polynomials of degree at most $ \deg P - 1 $. The quantisation is referred to as \emph{simple} if $ m < \infty $.
\end{definition}

Finally, one can now state the topological recursion/quantum curve correspondence.

\begin{conjecture}\label{con:GS}
	Let $\mathcal{S}$ be a global spectral curve, with $x$ and $y$ satisfying the relation $ P(x,y) = 0 $. Let $\psi(x(z))$ be the wavefunction (see \Cref{e:wfuncconj}) associated to $\mathcal{S}$, with the $\omega_{g,n}$ constructed from topological recursion. Then there exists a quantisation $\hat{P}(\hat{x},\hat{y}; \hslash)$ of $P(x,y)$ such that
	\begin{equation}
		\hat{P}(\hat{x},\hat{y}; \hslash) \psi(x(z)) = 0\,.
	\end{equation}
	$\hat{P}(\hat{x},\hat{y}; \hslash) $ is called the \emph{quantum curve}.
\end{conjecture}

This conjecture has been proven for a wide class of algebraic genus zero spectral curves \cite{BE17,W24,HS25}, every simple algebraic spectral curve \cite{EG19,MO20,EGMO21,MO22}, and a broad range of genus zero transalgebraic curves \cite{BKW24}.

This section will now review some results from \cite{BKW24,BE17}. To set the stage let $ P(x,y)=0 $ be the equation corresponding to a compact spectral curve $ \mathcal{S} $
\begin{equation}\label{e:Pexpansion}
	P(x,y) = \sum_{i=0}^{d} q_i (x) y^i = \sum_{(i,j) \in \mathbb{N}^2} \alpha_{i,j} y^i x^j \,,
\end{equation}
where $d$ is the degree of the curve (which may be infinite) and proceed with the following few definitions which will be necessary to explain the construction of quantum curves.
\begin{definition}
	The \emph{Newton polygon} $\Delta$ of $P$ is the convex hull of the exponents in $P$, i.e., the convex hull in $ \R^2 $ of $ A \coloneqq \{(i,j)\in\mathbb{N}^2 \, | \, \alpha_{i,j} \neq 0\}$.
\end{definition}

\begin{definition}
	A global spectral curve is called \emph{regular}\footnote{More simply, but less practically, a curve is regular iff the interior of it's Newton polygon contains no integral point.} if 
	\begin{enumerate}
		\item $P$ is linear in $x$;
		\item the Newton polygon of $P$ is the convex hull of $\{(0,0),(0,2),(2,0)\}$;
		\item it is related to one of the first two cases by a transformation of the form $(x,y)\to(x^ay^b,x^cy^d)$ with $ad-bc=1$ and $a,b,c,d\in\Z$.
	\end{enumerate}
\end{definition}

\begin{definition}\label{d:QCpoly}
	For $m=0,\dots,d$ define
	\begin{equation}
		Q_m(x,y) \coloneq \sum_{i=1}^{d-m} q_{m+i}(x) y^i.
	\end{equation}
\end{definition}

\begin{definition}\label{d:alpha}
	Given $m=0,\dots,d$ denote
	\begin{equation}
		\alpha_m \coloneqq \inf \{a \, | \, (a,m) \in \Delta \}, \qquad \beta_m \coloneqq \sup \{a \, | \, (a,m) \in \Delta\}.
	\end{equation}
\end{definition}
\begin{definition}\label{d:diffoperate}
	Let $b\in\mathbb{C}$ be a pole of $ dx$ where all the $\omega_{g,n}$ are holomorphic and $x$ is meromorphic. Define
	\begin{equation}
		E_i \coloneqq - \lim_{z_1 \to b} \frac{Q_{i}(x(z_1),y(z_1))}{x(z_1)^{\lfloor\alpha_{i}\rfloor+1}} \,,
		\qquad
		F_i \coloneqq \hslash\frac{x^{\lfloor\alpha_i\rfloor}}{x^{\lfloor\alpha_{i-1}\rfloor}} \frac{d}{ dx},
	\end{equation}
	where $\lfloor\cdot\rfloor$ is the floor function. For $b\in\mathbb{C}$ a zero of $q_d(x(z))$ that is not in the ramification locus of $x$ put
	\begin{equation}
		G_{i} \coloneqq \lim\limits_{z_1\rightarrow b} \frac{1}{x(z_1)^{\lfloor \alpha_{i}\rfloor}} Q_{i}(x(z_1),y(z_1)), \quad H_i \coloneqq \hslash\frac{x^{\lfloor\alpha_i\rfloor}}{x^{\lfloor\alpha_{i-1}\rfloor}} \left( \frac{d}{dx}-\frac{1}{x-x(b)} \right).
	\end{equation}
\end{definition}

With these definitions two theorems from \cite{BE17,BKW24}, may be stated.

\begin{theorem}(\cite[Lemma 5.14]{BE17}; \cite[Theorem~A.5]{BKW24}) \label{t:TR/QCpole}
	Let $\mathcal{S} = (\Sigma,\, x,\, y,\, B)$ be an algebraic or transalgebraic, admissible, regular spectral curve. Let $b$ be a pole of ${\rm d}x$ at which the $\omega_{g,n}$ are holomorphic. Then the wave-function \eqref{e:wfuncconj} satisfies the differential equation
	\begin{equation}
		\left(\frac{q_0(x)}{x^{\lfloor\alpha_0\rfloor}} + \sum_{i=1}^{d}F_1F_2\cdots F_{i-1}\frac{q_i(x)}{x^{\lfloor\alpha_i\rfloor}}F_i + \hslash\sum_{i=1}^{d-1}E_iF_1F_2\cdots F_{i-1}\frac{x^{\lfloor\alpha_i\rfloor}}{x^{\lfloor\alpha_{i-1}\rfloor}}\right)\psi(z;b)=0.
	\end{equation}
	(Note that $d$ may be infinite if the curve is transalgebraic.)
\end{theorem}

\begin{theorem}(\cite[Theorem~A.7]{BKW24})\label{t:TR/QCzero}
	Let $\mathcal{S} = (\Sigma,\, x,\, y,\, B)$ be a transalgebraic, admissible, regular spectral curve. Let $b$ be a zero of $q_d(x)$ for $d<\infty$ or, if $d=\infty$, a zero of $x$, with $b$ not in the ramification locus of $x$. Then the wave-function \eqref{e:wfuncconj} satisfies the differential equation
	\begin{equation}
		\left( \frac{q_0(x)}{x^{\lfloor\alpha_0\rfloor}} + \sum_{i=1}^{d}H_1 \cdots H_{i-1}\frac{q_i(x)}{x^{\lfloor\alpha_i\rfloor}}F_i + \hslash\sum_{i=1}^{d-1}G_iH_1 \cdots H_{i-1}\frac{x^{\lfloor\alpha_i\rfloor}}{x^{\lfloor\alpha_{i-1}\rfloor}(x-x(b))} \right)\psi(z;b)=0.
	\end{equation}
\end{theorem}


There is one other choice of basepoint which will be useful here, but was not considered in \cite{BKW24}. For a transalgebraic admissible regular spectral curve the polynomial $P(x,y)=0$ takes the form \cite[Proposition 5.8]{BKW24}
\begin{equation}
	P(x,y) = x P_2(xy) - P_1(xy) e^{P_0(xy)},
\end{equation}
up to rescaling by $x$ to get an irreducible equation, for some polynomials $P_0,P_1,P_2$ where it is assumed $P_1$ and $P_2$ are coprime. Furthermore, if $P_0$ is linear, take as an affine coordinate $z = P_0(xy)$, and note that the stable correlators $\omega_{g,n}$ are regular at $z=\infty$; thus, it seems the basepoint $z=\infty$ would be a natural choice. The next theorem makes this precise.

\begin{theorem}\label{t:TR/QCess}
	Let $\mathcal{S} = (\P^1,\, x,\, y,\, B)$ be a transalgebraic admissible regular spectral curve such that $\Erd_x(xy=\infty)=1$; following the above discussion denote the corresponding irreducible polynomial equation as $P(x,y) = P_2(xy) x - P_1(xy) e^{axy}$ for some $a\in\C^*$.\footnote{One would have to divide by $x$ if $P_1(0)=0$.} Then, taking the affine coordinate $z=xy$ and the basepoint $b=\infty$ the wavefunction \eqref{e:wfuncconj} satisfies the differential equation
	\begin{equation}
		\left(P_2\big(\hslash x \frac{d}{dx}\big) x - P_1\big(\hslash x \frac{d}{dx}\big) e^{a\hslash x\frac{d}{dx}}\right)\psi(z;\infty)=0.
	\end{equation}
\end{theorem}
\begin{proof}
	First assume $P_1(0) \neq 0$. Define a sequence of spectral curves via
	\begin{equation}\label{e:P_N}
		P_N(xy) = xP_2(xy) - P_1(xy)\left(1+a\frac{xy}{N}\right)^N.
	\end{equation}
	Now notice that
	\begin{equation}
		Q_i(x,y) = \sum_{m=1}^{d-i}q_{m+i}(x) y^m = -\sum_{m = 0}^{i} q_m(x) y^{m-i},
	\end{equation}
	as $P_N(x,y) = 0$. Then, from \eqref{e:P_N}, it follows that $q_m(x) = t_m x^{m+1} - u_m x^m$ and $\alpha_i = i$. Using these two facts
	\begin{equation}
		E_i = - \lim_{z \to \infty} \frac{Q_{i}(x(z),y(z))}{x(z)^{\lfloor\alpha_{i}\rfloor+1}} = t_i.
	\end{equation}
	Finally, from \Cref{t:TR/QCpole} one obtains, taking the limit as $N\to\infty$
	\begin{equation}
		\begin{split}
			&\left( t_0 x - u_0 + \sum_{i=1}^\infty \left(\hslash x \frac{d}{dx}\right)^{i-1} (t_i x - u_i) \left(\hslash x \frac{d}{dx}\right) + \hslash\sum_{i=1}^{\infty}t_i \left(\hslash x \frac{d}{dx}\right)^{i-1} x \right)\psi_\infty(z;\infty)=0, \\
			\Rightarrow & \left(P_2\big(\hslash x \frac{d}{dx}\big) x - P_1\big(\hslash x \frac{d}{dx}\big) e^{a\hslash x\frac{d}{dx}}\right)\psi_\infty(z;\infty)=0,
		\end{split}
	\end{equation}
	where $\psi_\infty(z;\infty) \coloneq \lim\limits_{N\to\infty} \psi_N(z;\infty)$ is the limiting wavefunction.
	
	Now assume the $P_1(0)=0$. Define the sequence of spectral curves
	\begin{equation}
		P_N(xy) = P_2(xy) - \frac{1}{x}P_1(xy)\left(1+a\frac{xy}{N}\right)^N.
	\end{equation}
	Here $\alpha_i = i-1+\delta_{i,0}$ and on can define $q_m(x) = t_m x^m - u_m x^{m-1}$ with $u_0=0$. It still follows that
	\begin{equation}
		E_i = - \lim_{z \to \infty} \frac{Q_{i}(x(z),y(z))}{x(z)^{\lfloor\alpha_{i}\rfloor+1}} = t_i,
	\end{equation}
	and \Cref{t:TR/QCpole} yields here, taking the limit as $N\to\infty$
	\begin{equation}
		\begin{split}
			&\left( t_0 + \frac{1}{x}\sum_{i=1}^\infty \left(\hslash x \frac{d}{dx}\right)^{i-1} (t_i x - u_i) \left(\hslash x \frac{d}{dx}\right) + \hslash\frac{1}{x}\sum_{i=1}^{\infty}t_i \left(\hslash x \frac{d}{dx}\right)^{i-1} x \right)\psi_\infty(z;\infty)=0, \\
			\Rightarrow & \left(\frac{1}{x}P_2\big(\hslash x \frac{d}{dx}\big) x - \frac{1}{x}P_1\big(\hslash x \frac{d}{dx}\big) e^{a\hslash x\frac{d}{dx}}\right)\psi_\infty(z;\infty)=0,
		\end{split}
	\end{equation}
	from which the stated result follows after multiplying by $x$.
	
	Lastly, it is not a priori clear if $\psi_\infty(z;\infty)=\psi(z;\infty)$; it will now be argued they are equal up to an irrelevant multiplicative constant. The potential issue is that the $\omega_{g,n}^N$ correlators constructed from the spectral curve $P_N$ may have contributions from the ramification point at $z=N$. These will vanish in the limit, but this vanishing may not commute with integration if the integration contour passes through infinity, which is the case here.
	
	However, \cite[(78)]{BKW24} says that these contributions must take the form
	\begin{equation}
		\frac{\omega_{g,n+1}^N(z,z_{\llbracket n \rrbracket})}{dzdz_1\cdots dz_n} = N^{1-2g-2n}\sum_{l=2}^{2g}\frac{A_l(z_{\llbracket n \rrbracket})}{(z/N-1)^l} + \mathcal{O}\left((z/N-1)^0\right),
	\end{equation}
	where $A_l(z_{\llbracket n \rrbracket}) = \mathcal{O}(N^0)$. Each integral from the basepoint infinity has the potential to add one power of $N$ so the principal part of $\int_\infty^z \cdots \int_\infty^z \omega_{g,n}^N(z_{\llbracket n \rrbracket})$ at $z=N$ is $\mathcal{O}(N^{3-2g-n})$. Thus, the only possible contributions are from the $2g+n-2=1$ correlators. However, the order of the pole at this point is bounded by $2g$, so the only contribution will be the $(g,n) = (1,1)$ correlator. This will take the form
	\begin{equation}
		\frac{1}{N}\int_\infty^z\frac{A_2dz}{(z/N-1)^2} = \frac{A_2}{1-z/N},
	\end{equation}
	which just converges to a constant as $N\to\infty$ as $A_2 = \mathcal{O}(N^0)$.
\end{proof}

\begin{remark}
	The above theorem could have been amended to quantise the curve $P(x,y) = P_2(xy) e^{a(\tau-1)xy} x - P_1(xy) e^{a\tau xy}$ for $\tau \in \C$ by using \Cref{t:TR/QCpole} when $\deg P_1 > \deg P_2$ and \Cref{t:TR/QCzero} when $\deg P_1 < \deg P_2$ (curiously, there is no analogous way to proceed when $\deg P_1 = \deg P_2$ for general $\tau$). It is easy to check that the result is as expected
	\begin{equation}
		\left(P_2\big(\hslash x \frac{d}{dx}\big) e^{a(\tau-1)\hslash x\frac{d}{dx}} x + P_1\big(\hslash x \frac{d}{dx}\big) e^{a\tau\hslash x\frac{d}{dx}}\right)\psi(z;\infty)=0 \,.
	\end{equation}
\end{remark}

\subsection{$x$-$y$ swap} \label{ss:xy}

One can solve certain $2$-matrix models with two fields (matrices) $M_1$ and $M_2$ using the topological recursion \cite{CEO06}. Interchanging $M_1$ and $M_2$ produces a natural duality of the matrix model; then, it is a natural to ask whether the topological recursion has a general manifestation of this duality. The answer is yes: switching $M_1$ and $M_2$ interchanges $x$ and $y $, one can switch $x$ and $y$ for any spectral curve, and the so-called $x$-$y$ swap is the correct abstraction of this duality to general spectral curves, regardless of whether or not they came from a $2$-matrix models \cite{EO08}.

\begin{definition}\label{d:x-yswap}
	Given a spectral curve $ \mathcal{S} = \left(\Sigma,\, x,\, y\, \omega_{0,2}\right) $ define the $x$-$y$ swap dual spectral curve
	\begin{equation}
		\mathcal{S}^{\vee} = \left( \Sigma,\, x^{\vee} \coloneq y,\, y^{\vee} \coloneq x,\, \omega_{0,2}^{\vee}\coloneq \omega_{0,2} \right)\,.
	\end{equation} 
\end{definition}

Critically ${\mc{S}^{\vee}}^{\vee} = \mc{S}$; but it's not just that the double dual is the same as the original, the two curves $\mathcal{S}$ and $\mathcal{S}^{\vee}$ are related in diver ways \cite{EO07,EO08,ABDKS22,ABDKS23,ABDKS24,H23b,H23a,H23c,W24,HS25}. Here, the relations most relevant to the present work are reviewed.

The fundamental result is the following \cite{ABDKS22}.
\begin{theorem}\label{t:x-yswap} 
	Let $\mc{S}$ be an admissible meromorphic simple global spectral curve and denote by $\omega_{g,n}$ the corresponding system of correlators. If $\mc{S}^{\vee}$ is also admissible, denote by $\omega_{g,n}^{\vee}$ its corresponding system of correlators. Then
	\begin{equation}
		\omega_{g,n}^{\vee}(z_{\llbracket n \rrbracket}) = \mathrm{Expr}_{g,n}\left(\left\{\omega_{g',n'}\right\}_{2g'+n'\leq 2g+n},\,\left\{dx(z_i),dy(z_i)\right\}_{i\in \llbracket n \rrbracket},\,\left\{\frac{dx(z_i)dx(z_j)}{\big( x(z_i) - x(z_j) \big)^2}\right\}_{i,j\in \llbracket n \rrbracket}\right),
	\end{equation}
	where $\mathrm{Expr}_{g,n}$ is some explicit expression that, for each $g,n$, is polynomial in its arguments and their derivatives with respect to the $z_i$.
\end{theorem}

The explicit form of $\mathrm{Expr}_{g,n}$ involves intricate sums over graphs and will not be needed; the interested reader can find it in \cite{ABDKS22}. Furthermore, it is slightly sloppy to say the only condition for the $x$-$y$ swap is the admissibility of $S$ and $S^{\vee}$; the truth is slightly more subtle. See \cite{ABDKS25} for an extensive discussion of exactly how general the $x$-$y$ swap formalism is. In the following, it will be sufficient that it always holds when both $dx$ and $dy$ have simple zeros and their poles do not contribute to the topological recursion; this is essentially the generic case and the case for all examples studied here. 

One relevant consequence of the $x$-$y$ swap is it implies a relation between the wavefunction $\psi(x;b)$ of \eqref{e:wfuncconj} and it's dual $\psi^{\vee}(y^{\vee},b^{\vee})$, which is the wavefunction defined from the dual spectral curve of \Cref{d:x-yswap}.

\begin{corollary}\label{c:wavex-y}
	Let $\mc{S}$ be an admissible spectral curve such that $\mc{S}^{\vee}$ is also admissible and take two base points $b\in \Sigma$. Then, for the wavefunctions
	\begin{equation}
		\begin{split}
		\psi(z;b) &= \exp\left[ \sum_{n=1}^{\infty} \sum_{g=0}^{\infty} \frac{\hslash^{2g+n-2}}{n!} \int^z_b\cdots\int^z_b \left( \omega_{g,n}  - \delta_{g,0} \delta_{n,2} \frac{dx(z_1) dx(z_2)}{(x(z_1) - x(z_2))^2} \right) \right] \,, \\
		\psi^{\vee}(z;b) &= \exp\left[ \sum_{n=1}^{\infty} \sum_{g=0}^{\infty} \frac{(-\hslash)^{2g+n-2}}{n!} \int^z_{b}\cdots\int^z_{b} \left( \omega^{\vee}_{g,n}  - \delta_{g,0} \delta_{n,2} \frac{dy(z_1) dy(z_2)}{(y(z_1) - y(z_2))^2} \right) \right] \,,
		\end{split}
	\end{equation}
	the following relation holds \cite{HS25}
	\begin{equation}
		\frac{\psi(z;b)}{x(z)-x(b)} = \frac{i}{2\pi\hslash} \iint \frac{dy(w)dy(a)}{y(a)-y(w)} e^{\frac{x(z)y(w) - x(b)y(a)}{\hslash}} \psi^{\vee}(w;a) \,,
	\end{equation}
	where the integrations are understood to be taken formally, order by order in $\hslash$ (see \cite{HS25} for details). 
\end{corollary}

This theorem yields an explicit relation between a wavefunction and it's dual and is a highly effective tool in the construction of quantum curves. If both $dx$ and $dy$ have a pole at a mutual basepoint $b$, the above result simplifies further \cite{W24}. In the present work, quantum curves will be constructed using the tools of \cite{BE17,BKW24}, but the results should also be obtainable via the $x$-$y$ duality.

\section{Topological recursion on periodic curves}\label{s:per}
\subsection{Definition}\label{ss:defper}
To begin the examination of periodic spectral curves, start by arguing that periodicity implies a very specific form for the differential $dx$.

\begin{proposition}\label{p:perform}
	Let $\mathcal{S} = \left(\Sigma,\, x,\, y,\, (A,B),\, \omega_{0,2}\right)$ be a periodic spectral curve (recall \Cref{d:per}). Let $z$ be an affine coordinate on $\Sigma\cong \P^1$ such that $x$ is a periodic function in $z$ with period $p$. By \Cref{d:per} the ramification locus of $x$ may be written as
	\begin{equation}\label{e:finitegen}
		R=R_0\cup R_\infty=\left\{a_j+kp|k\in\mathbb{Z},1\leq j\leq r\right\}\cup\{\infty\}\,,
	\end{equation}
	where $r \in \mathbb{Z}_{\geq 0}$ and $a_{j_1}\neq a_{j_2}+kp$ for all distinct $1\leq j_1, j_2 \leq r$ and integers $k\in\mathbb{Z}$. Then, there exists positive integers $b_1,\dots,b_r \in \mathbb{Z}_{\geq 1}$ and constants $\lambda p\in 2\pi{\rm i}\mathbb{Z}$ and $C_0\in\mathbb{C}^*$ such that
	\begin{equation}
		dx(z) = C_0 e^{\lambda z} \prod_{j=1}^{r} (z-a_j)^{b_j} \prod_{k=1}^{\infty} \left[ 1-\left( \frac{z-a_j}{kp} \right)^2 \right]^{b_j} dz \,.
	\end{equation}
\end{proposition}
\begin{proof}
	As $x$ is periodic and has only finitely many singularities (by assumption), it's only singularity is an essential singularity at $z=\infty$. Let $b_j = \Mult_x(a_j) \geq 1$ and consider the expression
	\begin{equation}
		\frac{d x}{d z}\prod_{j=1}^{r}(z-a_j)^{-b_j}\prod_{k=1}^{\infty}\left[1-\left(\frac{z-a_j}{kp}\right)^2\right]^{-b_j} \,,
	\end{equation}
	which is a never zero entire function with period $p$. Thus, it is equal to $\exp(f(z))$ for some entire function $f(z)$ so
	\begin{equation}
		dx(z) = e^{f(z)} \prod_{j=1}^{r} (z-a_j)^{b_j} \prod_{k=1}^{\infty} \left[ 1-\left( \frac{z-a_j}{kp} \right)^2 \right]^{b_j} dz.
	\end{equation}

	As $x$ must be of finite order, $f$ can not have an essential singularity at infinity and is therefore a rational function. However, as $x$ is periodic with period $p$, we know that $f(z+p)=f(z)+2\pi i m$ for some $m\in\Z$. The only rational functions with this property are those of the form $f(z) = \lambda z + C$ with $C \in \mathbb{C}$ and $\lambda p\in 2\pi{\rm i}\mathbb{Z}$.
\end{proof}

\begin{example}\label{ex:infprod}
	Consider the differential $dx(z) = \sinh(z)dz$. Here $r=2$ with $b_1=b_2=1$ and one can choose $a_1=0$ and $a_2=i\pi$ so the expansion reads (notice $\lambda = 0$)
	\begin{equation}
		dx(z) = \frac{2i}{\pi^2} z(z-i\pi) \prod_{k=1}^{\infty} \left[ 1 + \left( \frac{z}{2\pi k }\right)^2 \right] \left[ 1 + \left( \frac{z-i\pi}{2\pi k} \right)^2 \right] dz \,.
	\end{equation}
	This is clearly not the well-known infinite product expansion of $\sinh z$. To obtain the regular infinite product expansion of $\sinh$ first observe
	\begin{equation}
		1+\left(\frac{z-i\pi}{2\pi k}\right)^2 = \left(\frac{4k^2}{(2k-1)(2k+1)}\right)\left(1-\frac{iz}{(2k+1)\pi}\right)\left(1+\frac{iz}{(2k-1)\pi}\right) \,,
	\end{equation}
	and then note the following re-arrangement of factors
	\begin{equation}
		(z-i\pi)\prod_{k=1}^{\infty}\left[1+\left(\frac{z}{2\pi k}\right)^2\right]\left[1-\frac{iz}{(2k+1)\pi}\right]\left[1+\frac{iz}{(2k-1)\pi}\right] = -i\pi\prod_{k=1}^{\infty}\left[1-\left(\frac{z}{\pi k}\right)^2\right] \,.
	\end{equation}
	Finally note
	\begin{equation}
		\prod_{k=1}^{N} \frac{4k^2} {(2k-1)(2k+1)} = \frac{((2N)!!)^2} {(2N-1)!!(2N+1)!!} \stackrel{N\to\infty}{\sim} \frac{ 2\pi N(2N/e)^{2N} } {4N((2N-1)/e)^{N-1/2}((2N+1)/e)^{N+1/2}} \stackrel{N\to\infty}{\to} \frac{\pi}{2} \,,
	\end{equation}
	where Stirling's approximations for the double factorial $(2N)!!\sim\sqrt{2\pi N}(2N/e)^{N}$ and $(2N+1)!!\sim\sqrt{4N+2}((2N+1)/e)^{2N+1/2}$ were applied. So indeed, the well-known infinite product expansion is obtained
	\begin{equation}
		\frac{2i}{\pi^2} z(z-\pi) \prod_{k=1}^{\infty} \left[1+\left(\frac{z}{2\pi k}\right)^2\right]\left[1+\left(\frac{z-i\pi}{2\pi k}\right)^2\right] = z\prod_{k=1}^{\infty}\left[1+\left(\frac{z}{\pi k}\right)^2\right] \,,
	\end{equation}
	as expected.
\end{example}

To define admissibility, we first need the following two definitions. For the convergence of the infinite sums in the topological recursion, we want $y$, if meromorphic, to have a pole at infinity.
If $y$ has an essential singularity at infinity, this essential singularity needs to somehow behave as if it was pole.
\begin{definition}\label{d:pole-like}
  A meromorphic function $y(z)$ on a domain $\Sigma\subseteq \C_\infty$ is said to have a \emph{pole-like} singularity along a set $R_0 \subset \Sigma$ if $\exists$ $\epsilon,c_y >0$ and $\phi:R_0\to [0,2\pi)$ such that 
  \begin{equation}
    \Re\left\{ e^{i \phi(a)} y'(z+a) \right\} \geq c_y \quad \forall \ a\in R_0 \text{ and } |z|<\epsilon \,.
  \end{equation}
  Notice that if $y$ has poles at all accumulation points of $R_0 \subset \C_\infty$, then it automatically has a pole-like singularity along $R_0$.
\end{definition}

An analytic condition on $\omega_{0,2}$ will also be needed for convergence.
\begin{definition}\label{d:tamed}
  A symmetric meromorphic bidifferential on a domain $\C_\infty^2$, $D \subseteq \C_\infty$, $\omega_{0,2}$ is said to be \emph{tamed} on the set $R_0$ if it takes the following form
  \begin{equation}
    \omega_{0,2}(z_1,z_2) = \frac{dz_1dz_2}{(z_1-z_2)^2} + \phi_{0,2}(z_1,z_2)dz_1dz_2 \,,
  \end{equation}
  where $\exists$ $\epsilon,c_{0,2}>0$ such that
  \begin{equation}
    |\phi_{0,2}(z_1,z_2)| < c_{0,2} \quad \forall \ z_1,z_2 \in \Sigma \text{ and } a\in R_0 \text{ such that } |z_1 - a| < \epsilon \,.
  \end{equation}
\end{definition}

Now, for periodic curves, admissibility is defined as follows; this definition will be justified a posteriori.
\begin{definition}\label{d:admissper}
	A periodic canonically polarised spectral curve $\mathcal{S} = \left(\Sigma,\, x,\, y,\, \omega_{0,2}\right)$ is called \emph{admissible} if it is locally admissible at every finite ramification point $a\in R_0$ in the sense of \Cref{d:CondSCAdm}, $\omega_{0,2}$ is tamed on the set $R_0$, and the curve falls into one of the following two cases:
	\begin{enumerate}
    \item $y$ has a pole like singularity at the essential singularity of $x$ along $R_0$ and $\lambda=0$ (where $\lambda$ is as in \Cref{p:perform});
		\item $xy$ is meromorphic on $ \Sigma $ with a pole at the essential singularity of $ x $.
	\end{enumerate}
	When it is necessary to distinguish between the two cases curves that satisfy (1) will be called $\mathbb{C}$-admissible and curves that satisfy (2) $\mathbb{C}^*$-admissible.
\end{definition}

\begin{definition}\label{d:regper}
  A global periodic spectral curve $\mathcal{S} = \left(\Sigma,\, x,\, y,\, \omega_{0,2}\right)$ is called \emph{regular} if the curve satisfies one of the following two conditions:
	\begin{enumerate}
		\item $z \coloneq y$ is an affine coordinate on $ \Sigma \cong \mathbb{P}^1 $, $ \lambda = 0 $, and the essential singularity of $x$ is located at $ z = \infty $;
		\item $z \coloneq xy$ is an affine coordinate on $ \Sigma \cong \mathbb{P}^1 $ and the essential singularity of $x$ is located at $ z = \infty $.
	\end{enumerate}
	When it is necessary to distinguish between the two cases curves that satisfy (1) will be called $\mathbb{C}$-regular and curves that satisfy (2) $\mathbb{C}^*$-regular.
\end{definition}

Note that regularity implies admissibility, as in \Cref{d:CondSCAdm} one will always get $s_a = 1$ or $s_a = r_a + 1$ for all $a\in R_0$ and $z$ automatically has a simple pole at $z=\infty$. The justification for this notion of regularity can be found in the discussion preceding Definition 5.6 in \cite{BKW24}. Next, define the topological recursion for periodic spectral curves.

\begin{definition}\label{d:TRinflocus}
	Given a spectral curve admissible in the sense of \Cref{d:admissper}, define the topological recursion to be exactly that of definition \Cref{d:CndTR} except the sum over $R_0$ is now infinite. The ramification point at infinity is ignored.
\end{definition} 

Of course, the above definition only makes sense if the sum over $R_0$ indeed converges, and it is only a natural generalisation of the topological recursion if the key properties of the correlators are maintained.

\begin{lemma}\label{l:converge}
  Given an admissible periodic spectral curve, the sum over $R_0$ in the definition of the $\omega_{g,n}$ (\Cref{d:TRinflocus}) converges absolutely.
\end{lemma}

\begin{proof}
	For notational convenience, up to periodicity, assume there is only one ramification point $a$ (it will be clear that the proof easily generalises) and rescale the coordinate $z$ so that $p=1$. Begin with an examination of the local deck transformations to find (perhaps unsurprisingly) that they obey a periodicity-like property. Let $\theta$ be a primitive $r_a$'th root of unity. Fixing $k\in\mathbb{Z}$, every $ \sigma_{a+k}^m(z) \in f_a(z) $ has an expansion of the form
	\begin{equation}
		\sigma_{a+k}^m(z)=a+k+\theta^m(z-a-k)+\mathcal{O}(z-a-k)^2 \,.
	\end{equation}
	Then compute
	\begin{equation}
		x(\sigma_{a+k}^m(z+1))=x(z+1)=x(z)\,,
	\end{equation}
	so $\sigma(z):=\sigma_{a+k}^m(z+1)$ is also a deck transformation. One can then observe, using a prime to denote a $z$-derivative
	\begin{equation}
		\sigma'(a+k-1)=\theta^m \,,
	\end{equation}
	to conclude that, using the uniqueness of local deck transformations upon fixing their first derivative,
	\begin{equation}
		\sigma^m_{a+k}(z+a)=\sigma^m_{a+k-1}(z) = a+k-1+\theta^m(z-a-k-1)+\mathcal{O}(z-a-k-1)^2 \,.
	\end{equation}
	Finally, use induction to establish that, for any $k-l$ with $l\in\mathbb{Z}$, there are expansion coefficients $s_l^m$ not depending on $k$ such that
	\begin{equation}\label{e:sigexp}
		\sigma_{a+k}^m(z)=a+k+\sum_{l=1}^{\infty}s_l^m(z-a-k)^l \,.
	\end{equation}

  Now proceed via induction on $2g+n-2$. For each $g$ and $n$ with $2g+n-2 \geq 1$ claim that the $\omega_{g,n}$ have the following expansion
  \begin{equation} \label{e:convergeform}
    \omega_{g,n+1}(z_0,z_{\llbracket n \rrbracket}) = \sum_{k_0,\dots,k_n = -\infty}^{\infty} \sum_{l_0,\dots,l_n = 1}^{N_{g,n}} a^{k_0;k_1,\dots,k_n}_{l_0;l_1,\dots,l_n} d\xi_{a+k_0}^{-l_0}(z_0) \cdots d\xi_{a+k_n}^{-l_n}(z_n) \,,
  \end{equation}
  where $a^{k_0;k_1,\dots,k_n}_{l_0;l_1,\dots,l_n} = \mc{O}(k_0^0\cdots k_n^0)$, is symmetric in the indices $l_1,\dots,l_n$ and $d\xi_{a+k}^l$ is defined to be, for $l\in\Z_{\geq 1}$,
  \begin{equation}\label{e:perbasis}
    \begin{split}
      d \xi^{-l}_{a+k}(z) &= \Res_{z' = a+k} (z'-a-k)^{-l-1}dz' \int_{a+k}^{z'} \omega_{0,2}(\cdot,z) \,, \\ 
      d\xi^l_{a+k}(z) &= (z-a-k)^{l-1} dz \,.
    \end{split}
  \end{equation}
  Notice that, by the fact that $\omega_{0,2}$ is tamed on $R_0$, there exists the following expansions
  \begin{equation}
    d\xi^{-l}_{a+k}(z) = \sum_{\substack{s=-l \\ s\neq -1}}^{\infty} \Xi_{a+k}^{l,s} (z-a-k)^s dz \,, \quad (z-a-k)^{-l} = \sum_{\substack{s=-l \\ s\neq 0}}^{\infty} \tilde{\Xi}_{a+k}^{l,s} d\xi_{a+k}^{s} \,, 
  \end{equation}
  with $\Xi^{l,s}_{a+k},\tilde{\Xi}^{l,s}_{a+k} = \mc{O}(k^0)$, so checking that there exists an expansion with coefficients bounded in $k$ is the same whether one uses the $d\xi^{l}_{a+k}$ basis or the $(z-a-k)^{l}dz$ basis.

  The induction beginning can be taken to be $\omega_{0,2}$; as there is no claim about the form of $\omega_{0,2}$, we can skip to the induction step. The strategy will be to examine each object in the recursive definition of the correlator $\omega_{g,n+1}$ and prove that it has an expansion about $a+k$ with coefficients bounded in $k$ or the sufficient condition that it is bounded in small ball about each $a+k$ in a manner such that the radius of the ball and the bound are independent of $k$.
	
  Put $M_2 \coloneq y$ if the curve is $\C$-admissible and $M_2 \coloneq xy$ if the curve is $\C^*$-admissible. In either case, $M_2$ has a pole-like singularity along $R_0$. Take $\phi:R_0 \to [0,2\pi)$ and $\epsilon,c_{y}>0$ as in \Cref{d:pole-like} and take $z$ such that both $|z-a-k|< \epsilon$ and $|\sigma_{a+kp}^m(z) - a-k| < \epsilon$; by \cref{e:sigexp} this can be done by choosing one $\epsilon' \leq \epsilon$ and taking $|z-a-kp| < \epsilon'$, where only one $\epsilon'$ is needed for every $k$. Let $\theta_{a}^m(z) = \arg(\sigma^m_{a+k}(z) - z)$ and compute
  \begin{equation}\label{e:y-kernel-bound}
    \begin{split}
      \left| M_2(\sigma^m_{a+k}(z)) - M_2(z) \right| &= \left| e^{i(\phi(a)-\theta_a^m(z))} \int_{0}^1 M_2'(z + t[\sigma^m_{a+k}(z)-z]) [\sigma^m_{a+k}(z)-z] dt \right| \\ 
      &\geq \left| \int_0^1 \Re\left\{ e^{i\phi(a)} M_2'(z + t[\sigma^m_{a+k}(z)-z]) \right\} \left|\sigma^m_{a+k}(z)-z\right| dt \right| \\
      &\geq c_y \left|\sigma^m_{a+k}(z)-z\right| \,.
    \end{split}
  \end{equation}
  Now, by \cref{e:sigexp}, $|\sigma^m_{a+k}(z)-z|$ is bounded below independently of $k$ so the expression
  \begin{equation}
    \left| \frac{1}{M_2(\sigma^m_{a+k}(z)) - M_2(z)} \right|
  \end{equation}
  may be uniformly bounded in $k$ in a sufficiently small, but with $k$-independent radius, circle about each point in $a+k \in R_0$.

  Now briefly consider $1/dx$ or $x/dx$ (depending on whether the curve is $\C$ or $\C^*$ admissible). As $x$ is periodic with period $p=1$, these expressions clearly have expansions about $a+k$ with coefficients independent of, and therefore bounded in, $k$.
	
  Next examine $\omega_{0,2}$. $\omega_{0,2}$ appears in two ways: first, as $\omega_{0,2}(\sigma^m_{a+k}(z),z_i)$, for $i=0,\dots,n$; second, as $\omega_{0,2}(\sigma^{m_1}_{a+k}(z),\sigma^{m_2}_{a+k}(z))$. Write
  \begin{equation}
    \omega_{0,2}(z_1,z_2) = \frac{dz_1dz_2}{(z_1-z_2)^2} + \phi_{0,2}(z_1,z_2) dz_1dz_2 \,,
  \end{equation}
  as in \Cref{d:tamed}. As $\phi_{0,2}$ is uniformly bounded whenever one of it's arguments approaches an element of $R_0$ by assumption, there is nothing that needs to be said of it. For the first term with the pole on the diagonal, note that
  \begin{equation}
    \frac{dz_1dz_2}{(z_1-z_2)^{2}} = \frac{d(z_1-a-k)d(z_2-a-k)}{\big((z_1-a-k)-(z_2-a-k)\big)^{2}} \,,
  \end{equation}
  so $\omega_{0,2}(\sigma^{m_1}_{a+k}(z),\sigma^{m_2}_{a+k}(z))$ will have expansions with coefficients bounded in $k$ by \cref{e:sigexp}. Finally, observe that
  \begin{equation}\label{e:02exp}
\omega_{0,2}( \sigma^m_{a+k}(z), z_i ) = \sum_{l\neq 0} l d\xi^{-l}_{a+k}(\sigma_{a+k}^m(z)) d\xi^l_{a+k}(z_i) \,.
  \end{equation}

  Now examine the full TR integrand about a ramification point $a+k \in R_0$
	\begin{equation}
		\Res_{z=a+k}\sum_{\emptyset \neq Z \subseteq \mf{f}'_{a+k}(z)} K_{|Z|+1}(z_0,z, Z) \mc{W}_{g,n,|Z|+1}(z, Z \mid z_{\llbracket n \rrbracket})\,,
	\end{equation}
  The recursion kernel $K_{|z|+1}$ is taken care of by the arguments for $\omega_{0,2}$, $1/dx$ (or $x/dx$), and the difference of $M_2$s. The $\mc{W}_{g,n,|Z|+1}$ factor is therefore what remains. However, the results for $\omega_{0,2}$, the specific form of the deck transformations \eqref{e:sigexp}, and the induction assumption will yield the result. Look first at the variable $z_0$. This expansion is obtained by taking the integrand 
  \begin{equation}
    \sum_{\emptyset \neq Z \subseteq \mf{f}'_{a+k}(z)} \frac{1}{dx(z)^{|Z|+1} \prod_{z'\in Z}\big(M_2(z')-M_2(z)\big)} \mc{W}_{g,n,|Z|+1}(z, Z \mid z_{\llbracket n \rrbracket}) \,,
  \end{equation}
  which is has $k$-bounded expansion coefficients about the expansion point $z=a+k$ by the induction assumption and prior results, multiplying by
  \begin{equation}
    \int^z_{a+k} \omega_{0,2}(\cdot, z_0) = \sum_{l} (z-a-k)^l d\xi_{a+k}^{-l}(z_0)\,,
  \end{equation}
  and taking the residue at $z=a+k$, which will give an expansion of the claimed form. That the other variables also expand as desired follows immediately from the expansion \eqref{e:02exp} and the induction assumption.
\end{proof}

\subsection{Properties}\label{ss:propper}

Here it is first shown that there always exists a sequence of condensed spectral curves $\mathcal{S}_N$ converging to any periodic spectral curve $\mathcal{S} = \left( \Sigma, \, x,\, y,\,\omega_{0,2}\right)$, and the properties of the topological recursion on $\mc{S}$ are then deduced from the properties of the topological recursion on $\mc{S}_N$.

\begin{theorem}\label{t:seqper}
  Given a periodic admissible spectral curve, $\mathcal{S} = \left( \Sigma,x,y,\omega_{0,2} \right)$, there exists a sequence of global condensed admissible spectral curves $\mathcal{S}_N = \left( \Sigma,x_N,y_N,\omega_{0,2} \right)$ such that the $\omega_{g,n}^N$ constructed from $\mathcal{S}_N$ converge to the $\omega_{g,n}$ constructed from $ \mathcal{S} $. The sequence of curves $\mc{S}_N$ will be global if $\mc{S}$ is.
\end{theorem}
\begin{proof}
	By admissibility assume the form \Cref{p:perform} for $dx$. Choose coordinates on $\Sigma\cong \C\cup\{\infty\}$ and $\P^1$ so that $x(0)=0$ and define
	\begin{equation}
    dx_N(z) \coloneq C e^{\lambda z} \prod_{j=1}^{r}(z-a_j)^{b_j} \prod_{k=1}^{N} \left[ 1-\left(\frac{z-a_j}{kp}\right)^2 \right]^{b_j} dz \,, \qquad x_N(z) \coloneq \int_{0}^{z} dx_{N} \,.
	\end{equation}
	Note that any holomorphic $1$-form on $\C$ is exact so there is no ambiguity in defining the integration. Then, if the curve is $\mathbb{C}$-admissible take the sequence of spectral curves to be 
	\begin{equation}
		\mathcal{S}_N = \left( \Sigma,\, x_N,\, y,\, \omega_{0,2} \right)\,,
	\end{equation}
	whereas if the curve is $\mathbb{C}^*$-admissible take the sequence to be
	\begin{equation}
		\mathcal{S}_N = \left( \Sigma,\, x_N,\, y_N \coloneq xy/x_N, \omega_{0,2}  \right)\,.
	\end{equation}
	In the following the $\mathbb{C}$-admissible case will be assumed and the appropriate adjustments for the $\mathbb{C}^*$-admissible case will be discussed afterwards.
	
	Now make two key observations. First, the finite $N$ ramification locus is a strict subset of the limiting ramification locus. Second, by \cite{BE13,BBCKS23},\footnote{In particular \cite[Theorem 5.8]{BBCKS23}. Note that the dependence on $N$ is analytic away from $z=\infty$ for large enough $N$.} in small punctured disks around each finite ramification point of $x_N$ the $\omega_{g,n}^N$ constructed from $\mathcal{S}^N$ converge locally uniformly as $x_N$ and $\omega_{0,1}^N = yx_N$ converge locally uniformly. The Weirerstrauss $M$-test will then be used to commute the limit in $N$ with the sum over ramification points. For notational simplicity, assume that $p=1$ and up to periodicity there is only one ramification point $a$. Then, inductively claim that, for $2g+n-2\geq 1$
  \begin{equation} \label{e:per-ind-ass}
    \omega_{g,n+1}^N(z_0,z_{\llbracket n \rrbracket}) = \sum_{l=2}^{L_{g,n+1}} \sum_{k=-\infty}^\infty a_{l,k}^N(z_{\llbracket n \rrbracket}) d\xi^{-l}_{a+k}(z_0) dz_1\cdots dz_n \,,
	\end{equation}
  where $L_{g,n}$ are some positive integers (they exist by \Cref{t:origprop}) and the $a_{l,k}^N(\llbracket n \rrbracket)$ are uniformly bounded in $k$ and $N$ in the following sense: for every $1>D>\delta>0$ there exists $M(z_{\llbracket n \rrbracket}) \in \mathbb{R}_{\geq 0}$ such that 
  \begin{equation}
    \left| a_{l,k}^N(z_{\llbracket n \rrbracket}) \right| < M(z_{\llbracket n \rrbracket}) \quad \forall \ z_1,\dots,z_n \in \Sigma \cap \left( \bigcup_{k \in \Z} \mathrm{ann}(a+k;\delta,D) \right) \,, 
  \end{equation}
  where $\mathrm{A}(a+k;\delta,D)$ is the annulus of inner radius $\delta$ and outer radius $D$ centred at $k+a$. Note that the $a_{l,k}^N$ do, of course, depend on $g$ and $n$ even though this is suppressed in the above notation. 

The induction argument will be started with $\omega_{0,2}$, and as nothing is claimed for $\omega_{0,2}$ proceed to the induction step. The claim follows if the following integral is bounded independently of $k,m,N\in\mathbb{Z}$\footnote{In a fixed coordinate $z$, there is no ambiguity in referring to a bound on a differential as one may always divide by $dz$.}
	\begin{equation}
		\oint_{|z-a-k|=\epsilon}(z-a-k)^{m} \sum_{\emptyset \neq Z \subseteq \mf{f}'_{a+kp}(z)} \left(\frac{1}{dx_N(z)}\right)^{|Z|} \prod_{z_0\in Z}\left(M_2(z)-M_2(z_0)\right)^{-1} \mc{W}_{g,n,|Z|+1}(z, Z \mid z_{\llbracket n \rrbracket})\,,
	\end{equation}
	where $M_2=\omega_{0,1}^N/dx_N$.\footnote{Of course, in this case, this is just $y$. But a similar argument for the $\mathbb{C}^*$-admissible case will be used, and this more general notation will make this later argument clearer. In either case, the function does not depend on $N$, by the definition of the sequence of curves.} The bound on the above integral will follow from a bound on the integral below
	\begin{equation}
		\epsilon^m \max_{|z-a-k|=\epsilon} \sum_{\emptyset \neq Z \subseteq \mf{f}'_{a+kp}(z)} \left| \frac{dz}{dx_N(z)} \right|^{|Z|} \prod_{z_0 \in Z} \left|M_2(z)-M_2(z_0)\right|^{-1} \left| \frac{\mc{W}_{g,n,|Z|+1}(z, Z \mid z_{\llbracket n \rrbracket})}{dt^{|Z|+1}} \right|\,.
	\end{equation}
  Obviously the pre-factor $|z-a-k|^{m}=\epsilon^m$ is bounded in the desired manner. Next check the $\left| \mc{W}_{g,n,|Z|+1}(z, Z \mid z_{\llbracket n \rrbracket}) \right|$ factor; if the $\omega_{0,2}$ factors within $\mc{W}_{g,n,|Z|+1}$ could be shown to behave, this factor would be bounded in the desired manner (recall \Cref{d:combcomb}) provided that $(\sigma_{a+k}^{m,N}(a+k+\epsilon e^{i\theta})-a-k)^{-l}$ for $l\in \Z_{>0}$ is bounded independently of $N$ and $k$ and that $\omega_{g',n'+1}$ (with $2g'+n'-1 < 2g+n-1$) factors are well-behaved by the induction assumption ($\omega_{0,2}$ is special as the induction assumption says nothing about these correlators). 
	
  To see $(\sigma_{a+k}^{m,N}(a+k+\epsilon e^{i\theta})-a-k)^{-l}$ for $l\in \Z_{>0}$ is bounded in the desired manner argue as follows. Fix an open neighbourhood $U$ of $a$ such that each $\sigma_{a+k}^{m,N}$ is univalent on $U+k$ (it is possible to choose one such $U$ as there is a fixed minimal distance between ramification points). Then consider the set of functions $\mc{N} \coloneq \{f:U\mapsto\mathbb{C}\,|\, f(z)=\sigma_{a+k}^{m,N}(z+k)-a-k\}$, where $N$ and $k$ are integers such that $|k|\leq N \geq 0$ and $l=0,\dots,b-1$. Now claim, for any $f\in\mathcal{N}$, that $\infty,\pm 1$ are lacunary values for all functions in $\mc{N}$. That any $f$ can not be infinite is clear. The fact that the values $\pm 1$ are missed is because, for $N\geq 1$, these correspond to ramification points and all non-ramification points that are mapped to ramification points by deck transformations correspond to branchpoints of those deck transformations. As all the $f$ are single-valued on $U$, these values must be missed. Thus, by the Montel-Carathéodory theorem $\mathcal{N}$ is a normal family.
	
  Now proceed by contradiction and presume there exists a sequence $\{(k_i,N_i)\}_{i=1}^\infty$ such that $\sigma_{a+k_i}^{m,N_i}(a+k+\epsilon e^{i\theta})-a-k \to 0$ as $i\to\infty$ for some $l$ and $\theta$. Presuming $\epsilon$ is chosen small enough, one can apply the above result on $\mc{N}$ to conclude that this sequence of functions has a convergent subsequence. This is an immediate contradiction of the Hurwitz theorem.
	
  Therefore, it is indeed the case that $(\sigma_{a+k}^{m,N}(a+k+\epsilon e^{i\theta})-a-k)^{-l}$ is bounded independently of $N$ and $k$. Examine now $\omega_{0,2}$. These factors will appear in $\mc{W}_{g,n,|Z|+1}$ in two ways. The first is where one argument is in $Z$ and the other is in $z_{\llbracket n \rrbracket}$. In this case
\begin{equation}
  \omega_{0,2}^N(\sigma^{m,N}_{a+k}(z),z_i) = \omega_{0,2}(\sigma^{m,N}_{a+k}(z),z_i) = \sum_{l \in \Z_{\geq 1}} l d\xi^{l}_{a+k}(\sigma^{m,N}_{a+k}(z)) d\xi^{-l}_{a+k}(z_i) \,,
\end{equation}
thus the result follows if $(\sigma_{a+k}^{m,N}(a+k+\epsilon e^{i\theta})-a-k)^{l}$ is bounded uniformly in $k,N$ for $l\in \Z_{>0}$. Luckily, the same argument as above, but using the analogous statement to Hurwitz's theorem for poles rather than zeroes, shows that this is the case. The second case is where both arguments are in $Z$. Here recall that $\omega_{0,2}$ can be written as
\begin{equation}
  \omega_{0,2}(z_1,z_2) = \frac{dz_1 dz_2}{(z_1 - z_2)^2} + \phi_{0,2}(z_1,z_2)dz_1dz_2 \,,
\end{equation}
as it is tamed. That the $\phi_{0,2}(\sigma^{m_1,N}_{a+k}(z),\sigma^{m_2,N}_{a+k}(z))$ portion is bounded is also follows from $\omega_{0,2}$ being tamed. Ergo examine
\begin{equation}
  \frac{ d\big( \sigma^{m_1,N}_{a+k}(z)-a-k\big) d\big( \sigma^{m_2,N}_{a+k}(z)-a-k\big) }{ \left( \big( \sigma^{m_1,N}_{a+k}(z)-a-k\big) - \big( \sigma^{m_2,N}_{a+k}(z)-a-k\big) \right)^2 }
\end{equation}
The only issue with such a term that has not been fairly directly addressed by prior arguments is whether the denominator will converge to zero for $z\neq a+k$. However, this is impossible, again by the normality of $\mc{N}$.

As $M_2$ is independent of both $N$ and $k$, and by a nearly identical boundedness argument as that shown in \cref{e:y-kernel-bound}, the $\left|M_2(z_0)-M_2(z)\right|^{-1}$ factors are bounded in the desired manner.
	
	Thus, all that is left is the factor from the recursion kernel
	\begin{equation}
		\left|\frac{dz}{dx_N(z)}\right|^{|Z|} \,,
	\end{equation} 
	and this is the most involved part. One wishes to show $dx_N(z)/dz$ is bounded below independently of $k$ and $N$ for $z = a+k+\epsilon e^{i\theta}$ for some $0\leq\theta<2\pi$. Explicitly, notice
	\begin{equation}
		dx_N(z) = \prod_{l=1}^{N} \left[ 1-\frac{(k+\epsilon e^{i\theta})^2}{l^2} \right]^b dz \,,
	\end{equation}
	for some $b\in\mathbb{Z}_{\geq 2}$. Now note
	\begin{multline}
		\left| 1-\frac{(k+\epsilon e^{i\theta})^2}{l^2} \right|^2 = 1 + \frac{k^2}{l^2}\left( -2-\frac{4\epsilon\cos\theta}{k} \right)\\
		+ \frac{k^4}{l^4} \left( 1 + \frac{4\epsilon\cos\theta}{k} + \frac{2\epsilon^2}{k^2}( 1+2\cos^2\theta ) + \frac{4\epsilon^3\cos\theta}{k^3} \right) + \mathcal{O}\left( \frac{\epsilon^2k^0}{l^2} \right)\\
		= \frac{\left( l^2-(k^2+\mathcal{O}(k^0l^0\epsilon))^2 \right)^2}{l^4} + \mathcal{O}\left(\frac{\epsilon^2k^0}{l^2}\right) \,.
	\end{multline}
	Ergo, one can choose $\epsilon\ll 1$ small enough to get arbitrarily close to having this expression be $(l^2-k^2)^2/l^4$ independently of $k$ and $l$, provided $k < l$ so the leading order behaviour does not cancel out. Then note
	\begin{equation}
		\prod_{l=1}^{k-1}\left(\frac{k^2-l^2}{l^2}\right)^b=\prod_{l=1}^{k-1}\left(1+\frac{k}{l}\right)^b=\left(\frac{(k+1)\raisingfactorial{k}}{(k-1)!}\right)^b \,,
	\end{equation}
	where $x\raisingfactorial{n}=x(x+1)\cdots(x+n-1)$ is the rising factorial. Now observe
	\begin{equation}
		\prod_{l=k+1}^{N}\left(\frac{l^2-k^2}{l^2}\right)^b\geq \prod_{l=k+1}^{\infty}\left(\frac{l^2-k^2}{l^2}\right)^b=\left(\frac{(k!)^2}{(2k)!}\right)^b \,,
	\end{equation}
	putting these two together
	\begin{equation}
		\left(\prod_{l=1}^{k-1}\frac{k^2-l^2}{l^2}\right)^b\left(\prod_{l=k+1}^{N}\frac{l^2-k^2}{l^2}\right)^b\geq k^b \,.
	\end{equation}
	Finally, the $k=l$ term is $\mathcal{O}_s\left((\epsilon/k)^b\right)$ so the $k$-independent lower bound exists.
	
	For the $\mathbb{C}^*$-admissible case the proof is very similar except the following replacement must be made
	\begin{equation}
		\left|\frac{dz}{dx_N(z)}\right|^{|Z|} \to \left|\frac{x_N(z)dz}{dx_N(z)}\right|^{|Z|}.
	\end{equation}
	Examine the large $N$ behaviour of $x_N(z)/dx_N(z)$ through repeatedly integrating by parts to obtain an asymptotic series
	\begin{equation}\label{e:x/dxexp}
		\begin{split}
			\frac{\int_0^z dx_N(w)}{dx_N(z)} &= \frac{1}{dx_N(z)} \int_0^z \frac{ dx(w) }{ \prod_{l=N+1}^{\infty}\left[ 1-\frac{(w-a)^2}{l^2} \right]^b } \\
			&= \frac{x(z)}{dx(z)} - \frac{1}{dx_N(z)}\int_{0}^{z} \left( \sum_{l=N+1}^{\infty} \frac{ 2bx(w)(w-a)dw }{ l^2 \left(1-\frac{(w-a)^2}{l^2}\right)^b } \right) \prod_{l=N+1}^{\infty} \left[ 1-\frac{(w-a)^2}{l^2} \right]^{-b} \\
			&= \frac{x(z)}{dx(z)} - \frac{1}{dx(z)}\left(\int_0^zx(w)dw\right)\left( \sum_{l=N+1}^{\infty} \frac{ 2b(z-a) }{ l^2\left(1-\frac{(z-a)^2}{l^2}\right)^b } \right) \\
			&+ \frac{1}{dx_N(z)} \int_0^z \left(\int_0^w x(w')dw'\right) \prod_{l=N+1}^{\infty} \left[ 1-\frac{(w-a)^2}{l^2} \right]^{-b} \\
			&\times \left[ \left( \sum_{l=N+1}^{\infty} \frac{ 2b(w-a) }{ l^2 \left(1-\frac{(w-a)^2}{l^2}\right)^b } \right)^2 + \sum_{l=N+1}^{\infty} \frac{ 2b }{ l^2 \left(1-\frac{(w-a)^2}{l^2}\right)^b }\left(1+\frac{2b}{l^2\left(1-\frac{(w-a)^2}{l^2}\right)^b}\right) \right]dw \,.
		\end{split}
	\end{equation}
	Examine now the order of the behaviour of the corrections when $z = a+k+\epsilon e^{i\theta}$ where $|k|\leq N$. Evidently, as $x$ is periodic, the factor that is just the integral of $x$ is order one, and the factor that is just $dx$ is order one (or, more precisely, $x'$ is order one). Then, observe
	\begin{equation}
		\sum_{l=N+1}^{\infty} \frac{ 2(a+k+\epsilon e^{i\theta}) }{ l^2\left(1-\frac{(k+\epsilon e^{i\theta})^2}{l^2}\right)^b } = \mathcal{O}\left(\frac{k}{N^2}\right) = \mathcal{O}\left(\frac{k^0}{N}\right) \,,
	\end{equation}
	as $k\leq N$. From \eqref{e:x/dxexp} applying similar estimates will yield the next term in the series is $\mathcal{O}(k^0/N^2)$. By continuing to integrate by parts, more factors of $l$ in the denominator will accumulate and fewer powers of $w$ (or $z$) in the numerator will appear and the asymptotic series may thereby be calculated to any finite order. 
	
	However, what is needed here is not the whole asymptotic series but just the fact $x_N(z)/x'_N(z)=x(z)/x'(z)+\mathcal{O}(N^{-1})$ uniformly in $z$ when $z=a+k+\epsilon e^{i\theta}$ for an integer $k \leq N$. This indeed follows, as $x(a+k+\epsilon e^{i\theta}) = x(a+\epsilon e^{i\theta})$ and $dx(a+k+\epsilon e^{i\theta}) = dx(a+\epsilon e^{i\theta})$.
\end{proof}

\begin{remark}\label{r:nintlam}
	Notice that in the above proof, for the $\mathbb{C}^*$-admissible curves, one does not actually need that $\lambda p \in 2\pi i\mathbb{Z}$. Rather, any $\lambda$ such that $\lambda p\in 2\pi i\mathbb{R}$ would have been sufficient. The main difference in defining such curves is that the $s_l^m$ of \eqref{e:sigexp} will depend on $k$, but will still satisfy a boundedness property. Although, here, TR will not be formally define in general for such curves, all of the results should go through. Even, perhaps unintuitively, a version of \Cref{p:periodic}, as the non-periodicity in $dx/x$ cancels.
\end{remark}

\begin{corollary}\label{c:perprop}
	The correlators $\omega_{g,n}$ constructed by the topological recursion defined by \Cref{d:TRinflocus} satisfy the following properties.
	\begin{itemize}
		\item \emph{Symmetry:} the $\omega_{g,n}$ are symmetric in their $n$ variables;\\
		\item \emph{Residueless:} for $2g+n-2\geq 0$ the $\omega_{g,n}$ have vanishing residue; \\
		\item \emph{Homogeneity:} under the rescaling $\omega_{0,1} \mapsto f\omega_{0,1}$ it follows that $\omega_{g,n} \mapsto f^{2-2g-n} \omega_{g,n}$;\\
		\item \emph{Pole Structure:} the $\omega_{g,n}$ only have poles in $R_0$, and for every $a\in R_0$ and $2g+n-2\geq 0$, $\omega_{g,n}$ will have a pole of order no more than $(s_a-1)(2g+n-2)+2g$ at $a$ in each of its $n$ variables, where $s_a$ is defined in \Cref{d:CondSCAdm}; \\
		\item \emph{Formula for $\omega_{0,3}$:} the simplest stable correlator $\omega_{0,3}$ is given by the following formula
		\begin{equation}
			\omega_{0,3}(z_1,z_2,z_3) = \sum_{a\in R_0} \Res_{z=a} \frac{ \omega_{0,2}(z,z_1) \omega_{0,2}(z,z_2) \omega_{0,2}(z,z_3) } { dx(z) dy(z) }\,.
		\end{equation}
	\end{itemize}
\end{corollary}
\begin{proof}
	These properties hold for each curve in the sequence by \Cref{t:origprop}, and therefore for the limiting curve. 
\end{proof}

\begin{corollary}\label{c:perx-y}
  The correlators $\omega_{g,n}$ constructed by the topological recursion defined by \Cref{d:TRinflocus} satisfy the $x$-$y$ swap relation \Cref{t:x-yswap}.\footnote{This isn't quite true in full generality. See the discussion following \Cref{t:x-yswap}; an easy sufficient condition is that the sequences $dx_N$ and $dy_N$ should (for generic large $N$) have simple zeroes.}
\end{corollary}
\begin{proof}
  This holds for each element in the sequence constructed in \Cref{t:origprop}. As \Cref{t:x-yswap} is a purely algebraic relation between the correlators, there is no issue with taking the limit.
\end{proof}

Now consider an example to illustrate how the above results may be put to work on the physics applications.

\begin{example}\label{ex:Lstring-prop}
  Consider the periodic spectral curve $\mc{S}$ with $\Sigma = \big\{ \Re\{z\} > 0 \big\} \subset \C$, $x = -2\cos(\pi b^{-1} z)$, $y = 2\cos(\pi b z)$, and $\omega_{0,2} = \frac{4z_1z_2dz_1dz_2}{(z_1^2-z_2^2)^2}$ where $b\in\C\setminus(\R\cup i\R)$. This curve will be shown to produce the same correlators as \Cref{ex:Lstring} in \Cref{s:pullback} and, for $b\in e^{i\pi/4} \R$, produces correlators related to the string amplitudes of the complex Liouville string \cite{CEMR24,CEMR25}. Here, it will be checked that it is $\C$-admissible. 
  First notice that
  \begin{equation}
    \omega_{0,2}(z_1,z_2) = \frac{4z_1z_2dz_1dz_2}{(z_1^2-z_2^2)^2} = \frac{dz_1dz_2}{(z_1-z_2)^2} + \frac{dz_1dz_2}{(z_1+z_2)^2} \,,
  \end{equation}
  is tamed on the set $R_0 = b\Z_{\geq 1}$ given the domain $\Sigma^2$. Now consider $y=2\cos(\pi b z)$; the goal is to show this function has a pole-like singularity along $b\Z_{\geq 1}$. Put $b = r_be^{i\theta_b}$, and, as $b\notin (\R\cup i\R)$, note that $\sin(2\theta_b) \neq 0$. Let $0 < \epsilon < (2r_b)^{-1}$ and choose $\phi = -\theta_b + \sgn(\sin(2\theta_b))\pi (kr_b^2\cos(2\theta_b)+3/2)$. Now consider $z=re^{i\theta}$ with $r<\epsilon$ and $kb \in b\Z = R_0$
  \begin{equation}
    \begin{split}
      \Re \left\{ e^{i\phi} y'(kb+z) \right\} &= 2r_b \Re \left\{ e^{i\sgn(\sin(2\theta_b))\pi (kr_b^2\cos(2\theta_b)+3/2)} \sin(\pi b(kb+z)) \right\} \\
      &= r_b \Re \left\{ e^{3i\pi/2 + i\sgn(\sin(2\theta_b))\pi (kr_b^2\cos(2\theta_b)+3/2)} \left( e^{i\pi(b^2 k + bz)} + e^{-i\pi(b^2 k + bz + 1)} \right) \right\} \,.
    \end{split}
  \end{equation}
  Next, if $\sin(2\theta_b) > 0$ the first term will go to zero as $k\to \infty$ whereas if $\sin(2\theta_0) < 0$ the second term will go to zero. Assume $\sin(2\theta_b) < 0$ so the second term may be neglected as $k\to\infty$. Then
  \begin{equation}
    \Re \left\{ e^{i\phi} y'(kb+z) \right\} = r_b e^{-\pi k r_b\sin(2\theta_b) - \pi r r_b \sin(\theta+\theta_0)}\cos\big( \pi r_b r \cos(\theta+\theta_0)\big) + \mc{O}(e^{\pi k r_b \sin(2\theta_0)}) \,.
  \end{equation}
  Thus, as $0 < \epsilon < (2r_b)^{-1}$ we have that $\cos\big( \pi r_b r \cos(\theta+\theta_0)\big) > \cos\big(\pi\epsilon / (4 r_b) \big) > 0$ is uniformly bounded below for $r < \epsilon$. The $\sin(2\theta_b) > 0$ case proceeds \textit{mutatis mutandis}.
\end{example}

\begin{proposition}\label{p:periodic}
	For regular curves, taking $z = x^ly$ where $l=0,1$ is chosen so that $z$ is an affine coordinate, it follows that 
	\begin{equation}
		\omega_{g,n}(z_1+p,z_2+p,\dots,z_n+p) = \omega_{g,n}(z_1,\dots,z_n) + px(z_1)^{-l}dx(z_1)\delta_{g,0}\delta_{n,1}.
	\end{equation}
\end{proposition}
\begin{proof}
	Shifting all coordinates $z\mapsto z+p$ it is easy to observe that $\omega_{0,2}$ is invariant and $\omega_{0,1}\mapsto\omega_{0,1} + px^{-l}dx$. As the correlators are unaffected by shifting $\omega_{0,1}$ by a pure function of $x$ (this property is immediately clear from the local definition of topological recursion, \Cref{d:TRinflocus}), the result follows immediately.
\end{proof}

In latter sections, to construct the wavefunction \eqref{e:wfuncconj} of periodic curves, one must obtain antiderivatives of the correlators $\omega_{g,n}$ by integrating each variable from a base point $b$. In the cases consider, it will turn out to be most natural to take this base point to be the essential singularity. This could, in principle, raise questions of whether one may commute limits with the integrals and sums, and, perhaps more importantly, whether the integral itself even converges. The following lemma assures us that these issues do not occur.

\begin{lemma}\label{l:wavefuncper}
	Let $\mathcal{S} = \left(\Sigma,\, x,\, \omega_{0,1},\, \omega_{0,2}\right)$ be a periodic regular spectral curve. Identifying $\Sigma$ with $\mathbb{C}\cup R_{\infty}$ where $R_\infty = \{\infty\}$, take any complex number $v\in\mathbb{C}^*\setminus \Delta$, where $\Delta$ is the convex hull of $R_0$;\footnote{$\Delta$ can not be all of $\C$, as $x$ has only a single period $p$ and has only finitely many ramification points in $\C/p\Z$ by \Cref{d:per}.} then, for all all stable correlators $\omega_{g,n}$, the following anti-derivatives are well-defined
	\begin{equation}
		W_{g,n}(w_{\llbracket n \rrbracket}) \coloneq \int_{v\infty}^{w_1}\int_{v\infty}^{w_2} \cdots \int_{v\infty}^{w_n} \omega_{g,n}(z_{\llbracket n \rrbracket}) \,,
	\end{equation}
	where the symbol $\int_{v\infty}^{w_i}$ means that one integrates along the shortest line segment from $\infty$ to $v\neq 0$, and then integrates from $v$ to $w_i$ in the regular, contour independent (as the correlators are residueless), manner. Furthermore, taking $\mathcal{S}_N$ to be the sequence of spectral curves constructed in \Cref{t:seqper} and the $W_{g,n}^N$ to be the antiderivatives obtained by integrating the $\omega_{g,n}^N$, it follows that
	\begin{equation}
		\lim\limits_{N\to\infty}W_{g,n}^N(w_{\llbracket n \rrbracket}) = W_{g,n}(w_{\llbracket n \rrbracket}) + \delta_{n,1}C_g(v) \,,
	\end{equation}
	where $C_g(v)\in\C$ is a constant depending on $v$ and $g$.
\end{lemma}
\begin{proof}
	For simplicity assume only one ramification point up to periodicity, and choose coordinates on $\C$ so that $x$ has a ramification point at zero and the period of $x$ is unity. Combining \Cref{p:periodic} with the form \eqref{e:convergeform} in \Cref{l:converge}, one can expand any stable correlator $\omega_{g,n+1}$ in the following manner
	\begin{equation}\label{e:badform}
		\omega_{g,n+1}(z_0,z_{\llbracket n \rrbracket}) = \sum_{l_0,\dots,l_{n}=2}^{L_{g,n+1}} \sum_{k_0,\dots,k_n=-\infty}^{\infty} \frac{a_{g;k_1,\dots,k_n}^{l_0,\dots,l_n}dz_0\cdots dz_n}{(z_0-k_0)^{l_0}(z_1-k_0-k_1)^{l_1}\dots(z_n-k_0-k_n)^{l_n}} \,,
	\end{equation}
	where $a_{g;k_1,\dots,k_n}^{l_0,\dots,l_n} = \mc{O}(k_1^0\cdots k_n^0)$, and the bound $L_{g,n+1}<\infty$ exists by \Cref{c:perprop}. Notice that $a_{g;k_1,\dots,k_n}^{l_0,\dots,l_n}$ does not depend on $k_0$, which exactly implements the periodicity property of \Cref{p:periodic}.
	
	However, this form is not quite good enough for the present purposes. Instead, claim inductively on $2g+(n+1)-2\geq 1$ that
	\begin{equation}\label{e:claimvan}
		a_{g;k_1,\dots,k_n}^{l_0,\dots,l_n} = \mc{O}\left(k_i^{2\delta_{k_i,0}+2\sum_{j\neq i}\delta_{k_i,k_j}-2n}\right) \,,\quad k_i\to\infty \,, \quad i\in\llbracket n \rrbracket \,.
	\end{equation}
	For $(g,n+1) = (0,3)$, the special formula for $\omega_{0,3}$ given in \Cref{c:perprop} implies that
	\begin{equation}
		a_{0;k_1,k_2}^{l_0,l_1,l_2} \propto \delta_{k_1,0}\delta_{k_2,0} \,,
	\end{equation}
	so the claim is clear in this case. Furthermore, for $n+1=1$ the claim is obvious, as there is no further claim beyond the already established form \eqref{e:badform}; in particular, the claim holds for $(g,n+1) = (1,1)$. Ergo, the claim holds for $2g+(n+1)-2 = 1$, and the proof proceeds to the induction step.
	
	Recall
	\begin{equation}
		\omega_{g,n+1}(z_0,z_{\llbracket n \rrbracket}) = \sum_{k=-\infty}^{\infty} \Res_{z = k} \sum_{\emptyset\neq Z\subseteq f'_k(z)} \frac{(k-z)dz_0}{(z_0-z)(z_0-k)}\frac{1}{\prod_{z'\in Z}\omega_{0,1}(z')-\omega_{0,1}(z)}\mc{W}_{g,n+1,|Z|}(z,Z,z_{\llbracket n \rrbracket}) \,.
	\end{equation}
	Then, by \eqref{e:badform} and \Cref{d:combcomb} (the definition of the $\mc{W}$), $z$ will appear in terms of $\mc{W}_{g,n+1,|Z|}(z,Z,z_{\llbracket n \rrbracket})$ via factors of the form (the first will come from stable correlators, whereas the last two will result from $\omega_{0,2}$)
	\begin{equation}
		\frac{d\sigma(z)}{(\sigma(z)-k')^l}\,,\qquad \frac{d\sigma(z)d\sigma'(z)}{(\sigma(z)-\sigma'(z))^2}\,,\qquad \frac{d\sigma(z)dz_i}{(\sigma(z)-z_i)^2}\,,
	\end{equation}
	where $l\in \Z_{\geq 2}$, $i\in\llbracket n \rrbracket$, and $\sigma,\sigma'$ are a deck transformations of $x$ such that $\sigma(k)=k$ and $\sigma'(k')=l'$. In the first sort of factor, if $k\neq k'$, then this will result in an $\mc{O}((k-k')^{-l'})$ contribution to the coefficient when the residue is taken, where $l'\in \Z_{\geq l}$. As $l'\geq l \geq 2$ it will, in particular, result in a contribution of order $\mc{O}(k^{-2})$. This precisely gives the contributions described in \eqref{e:claimvan}, where the delta functions take care of the fact that this does not occur for factors where $k=k'$. The second factor will give contributions that behave like the $k=k'$ terms, whereas the third will give contributions that behave somewhat like the $k\neq k'$ scenario and add an $\mc{O}\left((z_i-k)^2\right)$ factor to the final result, consistent with the induction claim. There will never be any contributions of positive powers of $k$, by the arguments given in \Cref{l:converge}.
	
	Now that \eqref{e:claimvan} has been established, that one can integrate in the desired manner will follow readily. First consider the $n=1$ case. Here
	\begin{equation}
		\sum_{l_0=3}^{L_{g,1}}\sum_{k_0=-\infty}^{\infty}\int_{v\infty}^{w_0}\left|\frac{a_g^{l_0}}{(z_0-k_0)^{l_0}}\right|dz_0<\infty\,,
	\end{equation}
	so the integral may be commuted with the sums when $l_0\geq 3$ and the answer will be finite. The $l_0=2$ terms require special care. Indeed, note that
	\begin{equation}
		\int_{v\infty}^{w_0}\sum_{k_0=-\infty}^{\infty}\frac{a_g^2dz_0}{(z_0-k_0)^2} = \pi^2a_g^2\int_{v\infty}^{w_0}\csc^2(\pi z_0)dz_0 = -\pi a_g^2 \cot(\pi w_0) + C_g(v)\,,
	\end{equation}
	so this term is also well-defined. That this $C_g(v)$ is precisely the one appearing in the statement of the lemma will be justified momentarily.
	
	Now consider the $n\neq 1$ case. First note that
	\begin{equation}
		\int_{v\infty}^{w_i}\left|\frac{1}{(z-k_0-k_i)^{l_i}}\right| dz = \mathcal{O}(|k_0-k_i|^{-l_i+1})\,.
	\end{equation}
	Then, recalling \eqref{e:badform} and the fact that the sums
	\begin{equation}
		\sum_k \frac{1}{k^2},\quad \sum_{k_1\neq k_2}\frac{1}{k_1k_2(k_1-k_2)^2},\dots\,,
	\end{equation}
	converge absolutely, one concludes the integration and the sum may be interchanged in these cases, and that all integrals are well-defined.
	
	Now consider the curves $\mc{S}_N$ defined in \Cref{t:seqper}; the goal is to prove $W^N_{g,n+1}\to W_{g,n+1}$. The proof presented in \Cref{t:seqper} that the coefficients $a_{l,k}^N$ are uniformly bounded in $k$ and $N$ is precisely the proof that the present $a_{g;k_1,\dots,k_n}^{l_0,\dots,l_n}$ are uniformly bounded in $k$ and $N$. The argument for the additional `factors' of $k_i^{-2}$ is identical for finite $N$. For $n\neq 1$ the Weierstrauss $M$-test may then be used to commute the limit in $N$ with the sum over $k$, as in the proof of \Cref{t:seqper}, and dominated convergence again commutes the limit with the integrals. For $n=1$, the $l_0\geq 3$ terms proceed exactly as in the $n > 1$ case.
	
	However, for $l=2$ and $n=1$, one instead notes (going back, temporarily, to the limiting curve)
	\begin{equation}
	\sum_{k_0=-\infty}^{\infty}\frac{a_g^2dz_0}{(z_0-k_0)^2} = a_g^2\left(\frac{dz_0}{z_0^2}+2\sum_{k=1}^{\infty}\frac{(z_0^2+k^2)dz_0}{(z^2-k^2)^2}\right)\,,
	\end{equation}
	evaluating the integral first, and then the sum
	\begin{equation}
		a_g^2\left(\int_{v\infty}^{w_0}\frac{dz_0}{z_0^2}+2\sum_{k=1}^{\infty}\int_{v\infty}^{w_0}\frac{(z_0^2+k^2)dz_0}{(z_0^2-k^2)^2}\right) = -a_g^2\left(\frac{1}{w_0}+2\sum_{k=1}^{\infty}\frac{1}{w_0^2-k^2}\right) = -\pi a_g^2\cot(\pi w_0)\,,
	\end{equation}
	whereas, evaluating the sum first, then the integral (the calculation is the same as before)
	\begin{equation}
		\int_{v\infty}^{w_0}a_g^2\left(\frac{dz_0}{z_0^2}+2\sum_{k=1}^{\infty}\frac{(z_0^2+k^2)dz_0}{(z_0^2-k^2)^2}\right) = -\pi a_g^2\cot(\pi w_0) + C_g(v)\,.
	\end{equation}
	
	For finite $N$ the $a_g^2$ will depend on $k$ and $N$ (although they will be uniformly bounded in both). In particular, $a_g^2$ will vanish for sufficiently large $|k|$, so the sum is truncated, and one can compute the integral first and then the sum. One can then commute the sum with the limit as $N\to\infty$ using the Weirstrauss $M$-test, and the limit with the integral using dominated convergence. Notice that the constant $C_g(\nu)$ arises as, to commute the limit in $N$ with both the sum and the integral, one has to integrate first then sum. If one tried to do the sum first, then the integrand wouldn't necessarily be bounded uniformly in $N$, and so dominated convergence couldn't be applied.
\end{proof}

\begin{remark}
	As the choice of $v$ only changes the wavefunction up to a multiplicative constant, which is largely irrelevant, this choice will often be dropped in notation, and the basepoint will merely be denoted by $\infty$ rather than $v\infty$.
\end{remark}

The section concludes with a theorem giving an explicit formula for quantum curves when the basepoint is chosen to be at infinity.

\begin{theorem}\label{t:TR/QCper}
	Let $\mathcal{S} = (\Sigma,\, x,\, y,\, \omega_{0,2})$ be a regular periodic spectral curve. Then $x$ and $y$ satisfy the relation $P(x,y) =  x - f(\tilde{y}) = 0$ for some periodic entire function $f$, where $\tilde{y} = y$ if $\mc{S}$ is $\C$-admissible and $\tilde{y} = xy$ if $\mc{S}$ is $\C^*$-admissible. Then, taking the affine coordinate $z= \tilde{y}$, the wave-function, for $v\in \C^*\setminus \Delta$ where $\Delta$ is the convex hull of $R_0$,
	\begin{equation}
		\psi(z;v\infty) \coloneq \exp\left[\sum_{n=1}^{\infty}\sum_{g=0}^{\infty}\frac{\hslash^{2g+n-2}}{n!} \int_{v\infty}^z\cdots\int_{v\infty}^z \left( \omega_{g,n}-\delta_{n,2}\delta_{g,0} \frac{dx(z_1) dx(z_2)}{(x(z_1) - x(z_2))^2} \right)\right]\,,
	\end{equation}
	satisfies the differential equation
	\begin{equation}
		\left( x - f(\frac{d}{d\tilde{x}}) \right) \psi(z;v\infty)=0 \,,
	\end{equation}
	where $\tilde{x} = x$ for $\C$-regular curves and $\tilde{x} = \log(x)$ for $\C^*$-regular curves.
\end{theorem}
\begin{proof}
	The proof of the form of the quantum curve proceeds, \textit{mutatis mutandis}, as the proof of \Cref{t:TR/QCess}. That the wavefunction is well-defined follows from \Cref{l:wavefuncper}.
\end{proof}

\section{Topological recursion on exponential curves}\label{s:exp}
\subsection{Definition}\label{ss:defexp}

There are a large number of spectral curves that take the form $P(X,Y) = 0$, where $(X,Y)\in\mathbb{C}^*\times\mathbb{C}^*$. Such curves are often related to topological strings \cite{BM08,BS12,Z12} and/or knot theory \cite{ABM12,BE12,GJKS14}, and are therefore of great interest. To get to regular $\C$-variables, rather than $\C^*$-variables, define $x,y\in\C$ through $e^x=X,\, e^y=Y$. The natural $1$-form $\omega_{0,1}$ is then $\omega_{0,1} = ydx = \log(Y)d\log(X)$. Normally, these curves are parametrised so that $X$ and $Y$ are meromorphic. However, one then has to cut the Riemann surface $\Sigma$ to define the logarithms, which ruins compactness.

To remedy this, and to obtain the curves that are referred to as exponential throughout the present work, apply the transformation $x = \log(x'),\, y = x'y'$ and note that the new $x'$ and $y'$ still satisfy $\log(Y)d\log(X) = y'dx'$. However, the polynomial equation $P(e^x,e^y) = 0$ has now become $P(x',e^{x'y'}) = 0$. To parametrise such a curve, one may take $x'y'$ to be meromorphic and allow $x'$ to have essential singularities (but no $\log$-cuts). This is the perspective adopted in the present work, and corresponds to the exponential spectral curves of \Cref{d:expc}. Of course, there is no a priori reason that changing the parametrisation of the curves yields related correlators, but some of the results with quantum curves (\Cref{s:app}) and some direct results (\Cref{s:pullback}) suggest that this is the case here. 

\begin{remark}
	In \Cref{d:expc}, the polynomial equation was written as $P(x,e^{\kappa xy}) = 0$, where $\kappa\in\C^*$ was introduced. Using the homogeneity property of the $\omega_{g,n}$ stated in \Cref{t:origprop}, one can always get rid of $\kappa$ by rescaling $y$. However, it is sometimes more convenient to keep $\kappa$ in the definition of the curve, rather than scaling it away.
\end{remark}

Notice that, given the above discussion, $x'$ should be a variable in $\mathbb{C}^*$ rather than $\P^1$ (in fact, $x'=X$). It turns out that exponential curves which, in some sense, respect the $(x,e^{xy}) \in \mathbb{C}^*\times\mathbb{C}^*$ structure are especially nice to work with. This is the source of the following definition.

\begin{definition}\label{d:natexp}
	An exponential spectral curve is called \emph{natural} if, on the curve $P(x,e^{\kappa xy}) = 0$ whenever $e^{\kappa xy}$ approaches zero or infinity it follows that $x$ approaches zero or infinity (not necessarily respectively, e.g. $(x,e^{\kappa xy})\to (0,\infty)$ is acceptable).
\end{definition}

Although most of the main results in the present work are only proved for natural exponential curves, it is the belief of the author that this is more of a technical assumption, and that the main theorems should go through in the `unnatural' case. Indeed, it is not hard to see there exists a sequence of natural exponential curves converging to any unnatural exponential curve (just rescale $x\to e^{N^{-1}x}$), but proving this limit is well-behaved turns out to be subtle.

A key result that will allow progress to be made with exponential curves is the following.
\begin{lemma}\label{l:simppole}
	For curves $P(x,e^{\kappa xy}) = 0$ the poles of $xy:(x,y)\mapsto xy$ are simple.
\end{lemma}
\begin{proof}
	By the Newton-Puiseux theorem there exists a non-zero integer $n$ such that, for all $xy$ with $|e^{\kappa xy}|<\epsilon$
	\begin{equation}\label{e:xiny}
		x(xy)=\sum_{m=-M}^{\infty}a_me^{\kappa \frac{m}{n}xy}.
	\end{equation}
	Assume $xy$ has a non-simple pole at $p\in\Sigma:=\overline{\{P(x,e^{\kappa xy})=0\}}$ and let $z$ be a coordinate chart near $p$ such that $z^r=(\kappa xy)^{-1}$ is valid in some open neighbourhood $U$ of $p$. Then the expansion \eqref{e:xiny} takes the form
	\begin{equation}
		x(z)=\sum_{m=-M}^{\infty}a_me^{\frac{m}{n}z^{-r}},
	\end{equation}
	which is valid in the open set $V=U\cap\{|e^{z^{-r}}|<\epsilon\}\neq\emptyset$, where the set is non-empty as $V$ must contain $p$ as a limit point by the great Picard theorem. Note that the biholomorphic map $z\mapsto e^{2\pi i/r}z$ preserves both $x$ and $xy$ on $V$, so by the identity theorem it preserves both $x$ and $xy$ on $U$. Ergo, the map $z:\Sigma\to\mathbb{P}^1$ maps the same point $(x(q),y(q))$, for $q\in U$, to $r$ different values in $\mathbb{P}^1$, which implies $r=1$.
\end{proof}

One should observe that the poles of $xy$ are exactly the elements of $R_\infty$, and away from these points the curve $P(x,e^{\kappa xy})=0$ has only finitely many ramification points, all of which are in $R_0$. Near the poles of $xy$, as the poles are simple, $xy$ is an injective function and thus a valid local coordinate choice on $\Sigma$. The next lemma establishes that the function $x$ is periodic in this coordinate.

\begin{lemma}\label{l:expperx}
	Let $\mathcal{S} = \left(\Sigma,\, x,\, y, \omega_{0,2}\right)$ be an exponential spectral curve. Let $U\subset\Sigma$ be an open neighbourhood of a pole of $z \coloneq xy$. Then $x$, considered as a function of $z$, is periodic.
\end{lemma}
\begin{proof}
	In the last lemma, it was established that, for $|e^{\kappa z}|$ sufficiently small and in $U$
	\begin{equation}
		x(z)=\sum_{m=-M}^{\infty}a_me^{\kappa \frac{m}{n}z}\,.
	\end{equation}
	Thus, locally near this pole and for $|e^{\kappa z}|$ sufficiently small 
	\begin{equation}
		x(z) = x\left(z + \frac{2\pi i n}{m\kappa} \right)\,.
	\end{equation}
	By the identity theorem, this relation must hold on all of $U$.
\end{proof}

Before defining topological recursion for exponential curves, the issue of admissibility must be addressed.

\begin{definition}
	An exponential spectral curve is called \emph{admissible} if, for every finite ramification point the regular definition of admissibility, \Cref{d:CondSCAdm}, holds.
\end{definition}

With the condition for admissibility established, it is now an appropriate time to define the topological recursion on exponential curves.

\begin{definition}\label{d:TRexp}
	Given an exponential spectral curve $\mathcal{S}$, the topological recursion is defined to be that of \Cref{d:CndTR}, except that the sum over $R_0$ is now, possibly, infinite.
\end{definition}

One should immediately check that the above definition makes sense, i.e. the sum over $R_0$ converges. This is the content of the lemma below.

\begin{lemma}
	Given the above definition of the topological recursion, the sum over $R_0$ converges and produces well-defined multidifferentials $\omega_{g,n}$.
\end{lemma}
\begin{proof}
	This proof proceeds in a nearly identical manner to that of \Cref{l:converge}. The argument regarding the periodicity property of the deck transformations carries over with the slight modification that it must be done locally in a neighbourhood of every pole of $z \coloneq xy$. Fixing an open neighbourhood $U$ of a pole of $z$ such that $z$ is a valid local coordinate on $U$, claim inductively that the stable correlators take the form \eqref{e:convergeform} for $z_0\in U$
	\begin{equation}
		\omega_{g,n}(z_0,z_{\llbracket n-1 \rrbracket})=\Omega^U_{g,n}(z_0 | z_{\llbracket n-1 \rrbracket})+\sum_{|k| \geq K^U_{g,n}}\sum_{l=2}^{N^U_{g,n}}\frac{a_l^{k,U}(z_{\llbracket n-1 \rrbracket})}{(z_0-a^U-kp^U)^l}\,,
	\end{equation}
	where for simplicity we assumed all ramification points in $U$ are related to $a^U$ via a period $p$, $\Omega^U$ is holomorphic in $z_0$ on $U$ for fixed $z_1,\dots,z_{n-1}$ and $a^{k,U}_l = \mc{O}(k^0)$. Then proceed with the argument of \Cref{l:converge} for each open set $U$ to get the desired claimed.
\end{proof}

Next, regularity is defined.

\begin{definition}\label{d:regexp}
	An exponential spectral curve $\mc{S}$ is called \emph{regular} if it is genus zero and $xy\in\Aut(\Sigma)$.
\end{definition}

This definition is no surprise, based on the definition of regularity presented for periodic curves. As with periodic curves, for a justification of this definition, consult the discussion just before \cite[Definition 5.6]{BKW24}. As with periodic curves, in the language of \Cref{d:CondSCAdm}, $s_a=1,r_a+1$ for all $a\in R_0$ for regular curves, and regular, perforce, implies admissible. Furthermore, regular exponential curves take on a very simple form.

\begin{proposition}\label{p:formexp}
	A regular exponential spectral curve must take the form
	\begin{equation}
		\mc{S} = \left(\P^1,\, x(z) = \frac{P_1(e^{\kappa z})}{P_2(e^{\kappa z})},\, y = z/x(z),\, \omega_{0,2}(z_1,z_2) = \frac{dz_1dz_2}{(z_1-z_2)^2}\right)\,,
	\end{equation}
	where $z\in\Aut(\P^1)$ is a coordinate on $\P^1$ and $P_1,P_2$ are coprime polynomials. The corresponding polynomial equation is then
	\begin{equation}
		P(x,e^{\kappa xy}) = xP_2(e^{\kappa xy}) - P_1(e^{\kappa xy})\,.
	\end{equation}
\end{proposition}
\begin{proof}
	The curve is genus zero, so the associated Riemann surface is isomorphic to $\P^1$. As $xy \in\Aut(\P^1)$ by \Cref{d:regexp}, the form of $\omega_{0,1}$ is clear. As $\mathcal{S}$ must be global and genus zero, the form of $\omega_{0,2}$ is fixed.\footnote{It is a standard result that in genus zero there is only one symmetric bidifferential whose only pole is a double pole on the diagonal with birsediue one. See, for example, \cite{EO07}.} Then, as $xy$ can be considered as a global affine coordinate on $\P^1$, $x$ must be a function purely of $xy$, so $P(x,e^{\kappa xy}) = 0$ must be linear in $x$. This immediately yields the claimed form of $x$.
\end{proof}
 
\subsection{Properties}\label{ss:propexp}

 This section will proceed analogously to \Cref{ss:propper}; first it will establish, for any exponential curve, the existence of a sequence of spectral curves converging to the exponential curve in question, and then use this sequence to deduce the properties of the limiting curve.

\begin{theorem}\label{t:seqexp}
	Given a natural admissible exponential spectral curve $\mathcal{S} = \left(\Sigma,\, x,\, y,\, \omega_{0,2}\right)$ there exists a sequence of algebraic curves $\mathcal{S}_N = \left(\Sigma,\, x_N,\, y_N \coloneq yx/x_N,\, \omega_{0,2}\right)$ such that the $\omega^N_{g,n}$ constructed from $S_N$ converge to the $\omega_{g,n}$ constructed from $\mathcal{S}$.
\end{theorem}
\begin{remark}\label{r:concol}
	For this proof the assumption will be made that the topological recursion is well-behaved with respect to limits when there are finitely many ramification points that all have uniformly bounded multiplicity. This is certainly true when these ramification points do not collide (one can just use dominated convergence to commute the limit with the residue), and in cases where the ramification points do collide the results of \cite{BE13} suggested that there are still no issues. However, in \cite{BBCKS23} it was realised that, in some edge cases, there can be issues; these will be ignored here, and the concerned reader should consult \cite{BBCKS23}, in particular Theorem 5.8.\footnote{Theorem 5.8 in \cite{BBCKS23} requires analytic dependence on a complex parameter. This is actually $N$ here; even though the dependence of $x_N$ on $N$ is not analytic, the dependence of $\omega_{0,1}^N = yx d\log(x_N)$ on $N$ is.} 
\end{remark}
\begin{proof}
	As $\mathcal{S}$ is an exponential curve $x$ and $y$ must satisfy, for some polynomial $P$, an equation of the form
	\begin{equation}
		P(x,e^{\kappa xy}) = 0 \,.
	\end{equation}
	By rescaling $y$, one can set $\kappa =1$ while only rescaling the correlators $\omega_{g,n}$ by the constant $\kappa^{2g+n-2}$, so, without loss of generality, $\kappa = 1$. Define $\mathcal{S}_N$ as the algebraic spectral curve such that $x_N$ and $y_N$ satisfy
	\begin{equation}
		P(x_N,\exp_N(x_Ny_N)) = 0 \,,
	\end{equation}
	where $\exp_N(z) = (1+z/N)^N$ and put $\omega_{0,1}^N \coloneq y_N dx_N = x_Ny_N d\log(x_N)$. The function $z\coloneq x_Ny_N$ has only simple poles for sufficiently large $N$ by \Cref{l:simppole} and is, therefore, a valid coordinate on an open neighbourhood of each such simple pole. Outside these open neighbourhoods there are only finitely many ramification points all with bounded multiplicity so there should be no convergence issues (see \Cref{r:concol}).
	
	There are now two main problems to be overcome: first, the number of ramification points becomes infinite, and it is a priori not clear whether the limit as $N\to\infty$ may be commuted with this infinite sum; second, the contribution from the ramification points where $x_Ny_N = N,\infty$ must vanish in the limit. The first issue will be resolved first and the second, second.
	
	Consider the algebraic spectral curve $\bar{\mathcal{S}}$ defined through $P(x,xy) = 0$ with ramification locus $\bar{R}$ and meromorphic functions $\bar{x},\bar{y}$. Choose an open neighbourhood $U$ of a simple pole of $z=x_Ny_N$ such that $z$ is a well-defined coordinate on $U$. Then notice that, in the coordinate $z$ there is the relation $x_N = \bar{x}\circ \exp_N$ so the location of the ramification points of $x_N$ in $U$ (other than potentially $z=N,\infty$) are $V_N\coloneq U\cap \exp_N^{-1}(\bar{R})$; also define $V_\infty \coloneq U\cap \exp^{-1}(\bar{R})$.
	
	Pause to prove the following lemma.
	\begin{lemma}\label{l:mindist}
		There exists $\delta>0$ (independent of $N$) such that all points in $V_N$ are at least a distance of $\delta$ apart in the coordinate $z$ (for all $N$). Furthermore, for any such $\delta>0$ and for every $z\in U \setminus \bigcup_{a\in V_N}B_{\delta}(a)$ it follows that $\exp_N^{-1}(\exp_N(z)) \cap \bigcup_{a\in V_N}B_{\delta}(a) = \emptyset$ . 
	\end{lemma}
	\begin{proof}
		It will first be established that there exists a $\delta>0$ (independent of $N$) such that all points in $V_N$ are at least $\delta$ apart in the coordinate $z$. Proceed by contradiction and assume such a delta doesn't exist; there then are two sequences of points $p_N,q_N$ such that the distance between $p_N,q_N\in V_N$ goes to zero as $N\to\infty$, but $p_N\neq q_N$ for all $N$. Put
		\begin{equation}
			p_N = N\left(\exp\left(\frac{2\pi i m_N}{N}\right) a_N^{1/N}-1\right),\qquad q_N = N\left(\exp\left(\frac{2\pi i n_N}{N}\right)b_N^{1/N}-1\right) \,,
		\end{equation}
		where $0\leq m_N,n_N\leq N-1$ and $a_N,b_N\in \bar{R}$. Observe
		\begin{equation}
			p_N-q_N = N\exp\left(\frac{2\pi i m_N}{N}\right)-N\exp\left(\frac{2\pi i n_N}{N}\right) + \mathcal{O}(N^{0}) \,,
		\end{equation}
		so, defining $l_N = m_N-n_N\mod N$ such that $0\leq l_N\leq N-1$, it must be the case that $l_N/N=\mathcal{O}(N^{-1})$ (recall that the roots of unity are evenly distributed around a circle of fixed circumference) just for the limit of $p_N-q_N$ to be finite. Then (as $U$ is an open neighbourhood about $z=\infty$ it is safe to assume that neither $a_N$ nor $b_N$ are zero),
		\begin{equation}
			\exp\left(-\frac{2\pi i n_N}{N}\right)(p_N-q_N) = 2\pi i l_N + \log(a_N/b_N) + O(N^{-1}) \,,
		\end{equation}
		from which one may conclude $a_N=b_N$ for sufficiently large $N$ as $\bar{R}$ is finite. Now,
		\begin{equation}
			\exp\left(-\frac{2\pi i n_N}{N}\right)(p_N-q_N) = l_Na_N^{1/N} + O(N^{-1}) \,,
		\end{equation}
		so $l_N$ must also be zero for sufficiently large $N$. Thus, $p_N = q_N$ for sufficiently large $N$, so such a $\delta>0$ must exist.
		
		Finally, given any $\sigma$ such that $\exp_N \circ \sigma = \exp_N$, one can find a corresponding $N^\text{th}$ root of unity $\theta$ so that $\sigma(w) = -N + \theta ( w + N )$. Then notice
		\begin{equation}
			|\sigma(z)-a| = |\theta z + \theta N - N - a| = |z + N - \bar{\theta} N - \bar{\theta} a| = |z-\sigma^{-1}(a)| \geq \delta \,,
		\end{equation}
		as $\sigma^{-1}(a) \in V_N$.
	\end{proof}
	
	For each $a\in V_\infty$ choose an open neighbourhood $U_a\subset B_{\delta/2}(a)$ of $a$ such that $x$ restricts to a fully ramified Galois cover on each $U_a$. Then choose $\epsilon>0$ so small that $B_\epsilon(a)\subset U_a$ for every $a$. That such an $\epsilon>0$ exists despite the fact that $V_\infty$ is not of finite cardinality is guaranteed by the fact that $x$ is periodic (in coordinate $z$). Note that such an $\epsilon>0$ is a valid 'delta' in the above lemma.
	
	Now write, for the stable correlators
	\begin{equation}\label{e:claimform}
		\omega_{g,n+1}^N(z_0,z_{\llbracket n \rrbracket}) = H_{g,n+1}^N(z_0,z_{\llbracket n \rrbracket}) + \sum_{l=2}^{2g}\frac{k_{l,g,n+1}^N(z_{\llbracket n \rrbracket})}{(1+z_0/N)^l} + \sum_{a \in V_N} \sum_{l=2}^{L_{g,n+1}} \frac{c_{l,g,n+1}^{a,N}(z_{\llbracket n \rrbracket})}{(z_0-a)^l}dz_0\,,
	\end{equation}
	where $H_{g,n}^N$ is holomorphic and uniformly convergent in $z_0$ on $U$, $k_{l,g,n}^N = \mathcal{O}(N^{1-2g-2n})$, and 
	\begin{equation}\label{e:cbound}
		\left|\frac{c_{l,g,n}^{a,N}(z_{\llbracket n \rrbracket})}{dz_1\cdots dz_n}\right|\leq M_{l,g,n}(z_{\llbracket n \rrbracket})\,,
	\end{equation}
	for some non-negative functions $M_{l,g,n}(z_{\llbracket n \rrbracket})$. The claim is that the stable correlators take this form and that $\omega_{g,n+1}^N\to\omega_{g,n}$ in the limit as $N\to\infty$. This claim will be proven via induction on $2g+n-2$. For the induction beginning, note that the claim is obvious for the unstable correlators $(g,n) = (0,1),(0,2)$ as no specific form is claimed, and the convergence is clear. Proceed, then, to the induction step.
	
	First, as has been discussed, there are no issues at the finite ramification points away from the poles of $xy$. Next, examine the ramification points near the poles of $x_Ny_N$ other than the ones that occur at $x_Ny_N=-N$ using the form \eqref{e:claimform}. Focusing in on one such simple pole of $x_Ny_N$ and defining the open set $U$ as before, fix a branch of $\log z$ such that the branch cut does not pass through any element of the set $\bar{R}$. Then, for each $N$ and a given $\bar{a} \in \bar{R}$ fix a sequence $a_N\in V_N$ such that $\exp_N(a_N) = \bar{a}$ and $a_N\to\log(\bar{a})$ as $N\to\infty$. For each $N$ denote by $E_N\coloneq \{\sigma \mid \exp_N\circ\sigma = \exp_N\}$ the set of deck transformations of $\exp_N$. Fix $\sigma\in E_N$ and $m\in\mathbb{Z}_{\geq 1}$. The current goal is to show \eqref{e:cbound} by showing the following can be bounded independent of $\sigma$ and $N$
	\begin{equation}\label{e:prbound}
		\oint_{|z-a_N|=\epsilon}(\sigma(z)-\sigma(a_N))^{m} \sum_{\emptyset \neq Z \subseteq \mf{f}'_{a_N}(z)} \left(\frac{x_N(z)}{dx_N(z)}\right)^{|Z|} \prod_{z_0\in Z}\left(\sigma(z_0)-\sigma(z)\right)^{-1} \mc{W}_{g,n,|Z|+1}(\sigma(z), \sigma(Z) \mid z_{\llbracket n \rrbracket})\,.
	\end{equation}
	Notice, for any $z,z'\in U$ and some $N^\text{th}$ root of unity $\theta$
	\begin{equation}
		|\sigma(z)-\sigma(z')| = |\theta z + \theta N - N - (\theta z' + \theta N - N)| = |z-z'| \,,
	\end{equation}
	so $|\sigma(z)-\sigma(a_N)| = \epsilon$ and $|\sigma(z)-\sigma(z_0)| = |z-z_0|$ are bounded above and below independently of $N$, $\sigma$, and $a$. $x_N/dx_N$ does not depend on the choice of $\sigma$. As $a_N\to\log(\bar{a})$ the sequence $a_N$ is bounded so the closure of the set $\bigcup_N\{|z-a_N|=\epsilon\}$ is compact and $x_N/dx_N$ will therefore converge uniformly on this set; ergo, $x_N/dx_N$ can be bounded independently of $N$.
	
	The only remaining obstacle to the current goal, then, is the $\mc{W}_{g,n,|Z|+1}(\sigma(z), \sigma(Z) \mid z_{\llbracket n \rrbracket})$ factor. Given the form \eqref{e:claimform}, the $H_{g,n}^N$ portion is bounded independently of $N$ by definition. The $c_{l,g,n+1}^{a,N}$ coefficients are themselves bounded independently of $N$ and $|\sigma(z_0)-a|\geq \epsilon$ for all $a\in V_N$ so none of the individual $(z_0-a)^{-l}$ factors may blow up. Thus, every term in the sum over $V_N$ is bounded in the desired manner, but it is not clear if the entire sum is, as the cardinality of the set $V_N$ grows without bound.
	
	To ameliorate this problem first notice
	\begin{equation}
	N\left[\exp\left(\frac{1}{N}(\log\bar{a}+2\pi i m)\right)-1\right] = (\log(\bar{a})+2\pi i m)\left(1+\mathcal{O}(\frac{m}{N})\right)\,.
	\end{equation}
	The sum over $V_N$ may be divided into a sum over $E_N$ and a sum over $\bar{R}$; as there are only finitely many elements in the set $\bar{R}$, it suffices to focus on one specific element, say $\bar{a}\in \bar{R}$, and consider a term, for $l\in\mathbb{Z}_{\geq 2}$
	\begin{equation}\label{e:sumNexp}
		\begin{split}
		\sum_{a\in \exp_N^{-1}(\bar{a})}\frac{c_{l,g,n}^{a,N}}{(z_0-a)^l} &= \sum_{m=1-\lceil N/2 \rceil}^{\lfloor N/2 \rfloor}\frac{c_{l,g,n}^{a_m,N}} {\big( z_0-[\log \bar{a} + 2\pi i m][1+\mathcal{O}(m/N)] \big)^l} \\ 
		&= \sum_{m=1-\lceil N/2 \rceil}^{\lfloor N/2 \rfloor}\left[\frac{c_{l,g,n}^{a_m,N}} {\left( z_0 - \log \bar{a} - 2\pi i m \right)^l} + \mathcal{O}\left(\frac{m^{1-l}}{N}\right)\right] \\ 
		&= \mathcal{O}\left(\frac{(\log N)^{\delta_{l,2}}}{N}\right) + \sum_{m=1-\lceil N/2 \rceil}^{\lfloor N/2 \rfloor}\frac{c_{l,g,n}^{a_m,N}} {\left( z_0 - \log \bar{a} - 2\pi i m \right)^l}\,,
		\end{split}
	\end{equation}
	where $a_m = \log\bar{a} + 2\pi i m +\mathcal{O}(m^2N^{-1})$. Thus, there is indeed no issue with the infinite sum. Before concluding that the the $\mc{W}_{g,n,|Z|+1}(\sigma(z), \sigma(Z) \mid z_{\llbracket n \rrbracket})$ is bounded, the term
	\begin{equation}
		\sum_{l=2}^{2g}\frac{k_{l,g,n+1}^N}{(1+z_0/N)^l},
	\end{equation}
	must be examined. As the $k_{l,g,n+1}^N$ are order $N^{-1}$, the only possible way this term does not remain bounded is if $z_0$ gets close to $-N$. In terms of \eqref{e:prbound}, this means $\sigma(z)$ must get close to $-N$ as $N\to\infty$ while $|z-a_N|=\epsilon$. However, notice
	\begin{equation}
		|\sigma(z)+N| = |\sigma(z) - \sigma(-N)| = |z+N| \geq |a_N + N| - |z-a_N| = |\log(a_0) + \mathcal{O}(N^{-1}) + N| + \epsilon = \mathcal{O}(N) \,,
	\end{equation}
	so the above scenario can not happen. 
	
	\eqref{e:prbound} is therefore bounded. As it is, up to $(2\pi i)^{-1}$, just a formula for the $c_{l,g,n}^{a,N}$, it establishes, for the induction step, the desired property of the $c_{l,g,n+1}^{a,N}$ in \eqref{e:claimform}. The fact that the sum over $l$ begins at two is simply due to the fact the correlators must be residueless, and the fact that there exists some bound on the order of the pole independent of $N$ follows from the fact that, for a given $a_0 \in \bar{R}$, $\omega_{0,1}$ will have the same order at all ramification points in $\exp_N^{-1}(a_0)$ for sufficiently large $N$, and the order of the pole of the correlators may be bounded based on this \Cref{t:origprop}. Putting these facts together, it has been shown that the contributions to the topological recursion from ramification points in $U$ other than $z=-N$ are compatible with the induction assumption \eqref{e:claimform}.
	
	Turn now to the pole at $z=-N$. That the contributions here fit with the claimed form of
	\begin{equation}
		\sum_{l=2}^{2g}\frac{k_{l,g,n+1}^N}{(1+z_0/N)^l}\,,
	\end{equation}
	with $k_{l,g,n+1}^N=\mathcal{O}(N^{1-2g-2n})$ is given by equation (78) in \cite[Corollary 4.10]{BKW24} (with the variables $m_1$ and $m_2$ set to one). It is here the naturalness assumption is used as \cite[Lemma 4.8]{BKW24} requires the order of $\omega_{0,1}^N(z)-\omega_{0,1}^N(\sigma(z)) = (z-\sigma(z)) d\log(x_N)$ to be bounded. This will only occur if the branchpoint corresponding to $z=-N$ is either zero or infinity, which is ensured by \Cref{d:natexp}.
	
	The pole at $z=\infty$ can not contribute as, by naturalness, $d\log(x_N)$ will have a simple pole at this point, $x_Ny_N = z$ also has a simple pole at this point, so $\omega_{0,1} = y_Ndx_N = x_Ny_N d\log(x_N)$ will have a double pole at $z=\infty$, killing all contributions to topological recursion for all $N$.
	
	Finally, ramification points outside of $U$ may also contribute. If these are away from any simple poles of $xy$, then there are no issues with convergence (or at least, this will be assumed; see \Cref{r:concol}) and these contributions can be put in $H_{g,n+1}^N$. If the ramification points are near a simple pole of $xy$ another open set, say $U'$, may be chosen and a coordinate $z' = x_Ny_N$ may be defined. The above arguments may be repeated, but the analogous terms to
	\begin{equation}
		\sum_{l=2}^{2g}\frac{k_{l,g,n+1}^N}{(1+z_0/N)^l} + \sum_{a \in V_N} \sum_{l=2}^{L_{g,n+1}} \frac{c_{l,g,n+1}^{a,N}(z_{\llbracket n \rrbracket})}{(z_0-a)^l}dz_0\,,
	\end{equation}
	that come from $U'$ will simply be holomorphic on $U$ and thus should also be put in $H_{g,n+1}^N$. That everything converges to the right limiting correlators is now clear; the Weierstrauss M-test may be used to commute the limit in $N$ with the infinite sum over ramification points in $V_N$, and each individual term converges. This concludes the induction step, and by extension, the the proof of the theorem.
\end{proof}

As with \Cref{t:seqper} for periodic curves, the above theorem allows one to conclude that the key properties of the topological recursion carry over in the limit.

\begin{corollary}\label{c:expprop}
	The correlators $\omega_{g,n}$ constructed by the topological recursion defined by \Cref{d:TRexp} satisfy the following properties for natural spectral curves.
	\begin{itemize}
		\item \emph{Symmetry:} the $\omega_{g,n}$ are symmetric in their $n$ variables;\\
		\item \emph{Residueless:} for $2g+n-2\geq 0$ the $\omega_{g,n}$ have vanishing residue; \\
		\item \emph{Homogeneity:} under the rescaling $\omega_{0,1} \mapsto f\omega_{0,1}$ it follows that $\omega_{g,n} \mapsto f^{2-2g-n} \omega_{g,n}$;\\
		\item \emph{Pole Structure:} the $\omega_{g,n}$ only have poles in $R_0$, and for every $a\in R_0$ and $2g+n-2\geq 0$, $\omega_{g,n}$ will have a pole of order no more than $(s_a-1)(2g+n-2)+2g$ at $a$ in each of its $n$ variables, where $s_a$ is defined in \Cref{d:CondSCAdm}; \\
		\item \emph{Formula for $\omega_{0,3}$:} the simplest stable correlator $\omega_{0,3}$ is given by the following formula
		\begin{equation}
			\omega_{0,3}(z_1,z_2,z_3) = \sum_{a\in R_0} \Res_{z=a} \frac{ \omega_{0,2}(z,z_1) \omega_{0,2}(z,z_2) \omega_{0,2}(z,z_3) } { dx(z) dy(z) }\,.
		\end{equation}
	\end{itemize}
\end{corollary}
\begin{proof}
	These properties hold for each curve in the sequence by \Cref{t:origprop}, and therefore for the limiting curve. Note that $A=R$ so $A\cap R_0 = R_0$.
\end{proof}

\begin{remark}
	For unnatural exponential curves one could consider the sequence of spectral curves $\mc{S}_N = \left(\Sigma,\, x,\, (A_N,\,\omega_{0,2}),\,y,\,\omega_{0,2}\right)$ where $A_N\subset R_0$ is an increasing sequence of subsets of the ramification locus $R_0$ such that $\bigcup_N A_N = R_0$. One can easily prove the limit commutes in this situation and thus the properties in \Cref{c:expprop} hold. However, properties (in particular, quantum curves) requiring $A=R$ may not hold.
\end{remark}

There is an analogous proposition for exponential curves, to \Cref{p:periodic} for periodic curves.
\begin{proposition}\label{p:expper}
For an exponential spectral curve, fix an open neighbourhood $U$ of a simple pole of $z \coloneq xy$ such that $z$ is a valid coordinate on $U$. Then $\omega_{g,n}(z_1+p,z_2+p,\dots,z_n+p) = \omega_{g,n}(z_1,\dots,z_n) + p\delta_{g,0}\delta_{n,1}d\log x(z_1)$ for some $p\in 2\pi i\Q$.
\end{proposition}
\begin{proof}
By \Cref{l:expperx} $x$ is periodic in the coordinate $z$; denote this period by $p$. Observe that, upon a shift in coordinates $z\mapsto z+p$, $\omega_{0,2}$ is invariant and $\omega_{0,1}\mapsto\omega_{0,1} + pd\log x$. As the correlators are unaffected by shifting $\omega_{0,1}$ by a pure function of $x$ (this property is clear from the local definition of topological recursion, \Cref{d:TRinflocus}), the result follows immediately.
\end{proof}

In genus zero, to construct the wavefunction \eqref{e:wfuncconj} for exponential spectral curves, one must perform iterated integrals of the correlators from a fixed basepoint $b$. In \Cref{ss:C3} it will turn out that taking $b$ to be an exponential singularity of $x$ will produce the `nicest' quantum curve. The proceeding lemma, in analogy with \Cref{l:wavefuncper} for periodic curves, resolves the potential issues that may arise when making this choice.

\begin{lemma}\label{l:wavefuncexp}
	Let $\mathcal{S} = \left(\Sigma,\, x,\, y,\, \omega_{0,2}\right)$ be an exponential regular spectral curve. Identifying $\Sigma$ with $\mathbb{C}\cup R_{\infty}$ where $R_\infty = \{\infty\}$, take any complex number $v\in\mathbb{C}^*\setminus \Delta$, where $\Delta$ is the convex hull of $R_0$;\footnote{$\Delta$ can not be all of $\C$, as $x$ has only a single period $p$ and has only finitely many ramification points in $\C/p\Z$, by \Cref{p:formexp}.} then, for all all stable correlators $\omega_{g,n}$, the following anti-derivatives are well-defined
	\begin{equation}
		W_{g,n}(w_{\llbracket n \rrbracket}) \coloneq \int_{v\infty}^{w_1}\int_{v\infty}^{w_2} \cdots \int_{v\infty}^{w_n} \omega_{g,n}(z_{\llbracket n \rrbracket}) \,,
	\end{equation}
	where the symbol $\int_{v\infty}^{w_i}$ means that one integrates along the shortest line segment from $\infty$ to $v\neq 0$, and then integrates from $v$ to $w_i$ in the regular, contour independent (as the correlators are residueless), manner. Furthermore, taking $\mathcal{S}_N$ to be the sequence of spectral curves constructed in \Cref{t:seqper} and the $W_{g,n}^N$ to be the antiderivatives obtained by integrating the $\omega_{g,n}^N$, it follows that
	\begin{equation}
		\lim\limits_{N\to\infty}W_{g,n}^N(w_{\llbracket n \rrbracket}) = W_{g,n}(w_{\llbracket n \rrbracket}) + \delta_{n,1}C_g(v) \,,
	\end{equation}
	where $C_g(v)\in\C$ is a constant depending on $v$ and $g$.
\end{lemma}
\begin{proof}
	The statement for the limiting curve $\mc{S}$ is proven in the exact same manner as \Cref{l:wavefuncper}. However, commuting the limit here is more subtle than in \Cref{l:wavefuncper} for two reasons: first, one has to consider whether the principal parts of $\omega_{g,n}^N$ at $z = -N$ (where $z \coloneq xy = x_Ny_N$) can survive in the limit once integrated; second, the location of all the ramification points change with $N$. To tackle the first problem recall \eqref{e:claimform} from the proof of \Cref{t:seqexp}, which gives the principal part of $\omega_{g,n+1}^N$ at $z=-N$
	\begin{equation}
		\omega_{g,n+1}^N(z_0,z_{\llbracket n \rrbracket}) = \sum_{l=2}^{2g}\frac{k_{l,g,n+1}^N(z_{\llbracket n \rrbracket})dz_0}{(1+z_0/N)^l} +\mathcal{O}\left((1+z_0/N)^0\right)\,,\quad k_{l,g,n+1}^N = \mathcal{O}\left(N^{1-2g-2n}\right)\,.
	\end{equation}
	Then notice, by direct computation
	\begin{equation}
		\int_{v\infty}^{w_0} \frac{dz_0}{(1+z_0/N)^l} = \mathcal{O}(N)\,.
	\end{equation} 
	Thus, as one performs $n+1$ integrations on $\omega_{g,n+1}$ to obtain $W_{g,n+1}$, this principal part may survive in the limit when $1-2g-2n+n+1\geq 0$, which will only occur for stable correlators when $(g,n+1) = (0,3), (1,1)$. However, the sum goes from $l=2$ to $2g$, so the contributions vanish identically for all $N$ when $g=0$, so one only need consider $\omega_{1,1}^N$. Here the contribution is just a constant in the limit $N\to\infty$, and may therefore be absorbed into $C_1(v)$.
	
	Now the second problem must be tackled. To this end, recall \eqref{e:sumNexp} from the proof of \Cref{t:seqexp}
	\begin{equation}
		\sum_{a\in \exp_N^{-1}(a_0)}\frac{c_{l,g,n}^{a,N}}{(z_0-a)^l} = \sum_{m=1-\lceil N/2 \rceil}^{\lfloor N/2 \rfloor}\frac{c_{l,g,n}^{a_m,N}} {\left( z_0 - \log a_0 - 2\pi i m \right)^l}\left(1+\mathcal{O}\left(\frac{m}{N}\right)\right)\,.
	\end{equation}
	These extra $\mc{O}(m/N)$ terms will not be of concern for $l\geq 3$. By the reasoning, \textit{mutatis mutandis}, as in the proof of \Cref{l:wavefuncper}, the $c_{l,g,n}^{a_m,N}$ will pick up extra factors of $\mathcal{O}(m^{-2})$, and the only actually problematic terms are for $l=2$ and $n=1$ (note that $c_{l,g,n}^{a_m,N}$ depends on $z_{\llbracket n \rrbracket}$, even though this is being suppressed here). For $n=1$, by the periodicity property, $c_{l,g,n}^{a_m,N}$ does not depend on $m$, up to $o(N^0)$. Then, for $l=2$ and $n=1$ consider
	\begin{equation}\label{e:w1nexpint}
		\lim\limits_{N\to\infty}\int_{v\infty}^{w_0}\sum_{m=1-\lceil N/2 \rceil}^{\lfloor N/2 \rfloor}\frac{c_{l,g,n}^{a_m,N}dz_0} {\left( z_0 - \log a_0 - 2\pi i m \right)^2}\left(1+\mathcal{O}\left(\frac{m}{N}\right)\right)\,.
	\end{equation}
	Now recall the following facts:
	\begin{itemize}
		\item the derivative of \eqref{e:w1nexpint} with respect to $w_0$ converges to the correct correlator by \Cref{t:seqexp};
		\item derivatives commute with limits of sequences of compactly convergent holomorphic functions;
		\item the extra $\mathcal{O}(m/N)$ factor in the second term of the summand may contribute in the limit as $N\to\infty$, but it will not blow up, even after integration.
	\end{itemize}
	Putting these recalled facts together, the integral and the limit as $N\to\infty$ in \Cref{e:w1nexpint} must commute up to a constant, as claimed.
\end{proof}

\begin{remark}
	As the choice of $v$ only changes the wavefunction up to a multiplicative constant, which is largely irrelevant, this choice will often be dropped in notation, and the basepoint will merely be denoted by $\infty$ rather than $v\infty$.
\end{remark}

With the previous lemma at hand, one can construct quantum curves for regular exponential curves, choosing the basepoint to be at infinity. The result is the content of the following theorem.

\begin{theorem}\label{t:TR/QCexp}
	Let $\mathcal{S} = (\Sigma,\, x,\, y,\, \omega_{0,2})$ be a regular exponential spectral curve. Denote the corresponding irreducible polynomial equation as $P(x,y) = P_2(e^{xy}) x - P_1(e^{xy})$ for some polynomials $P_1,P_2$. Then, taking the affine coordinate $z= xy$ the wave-function, for $v\in \C^*\setminus \Delta$ where $\Delta$ is the convex hull of $R_0$,
	\begin{equation}
		\psi(z;v\infty) \coloneq \exp\left[\sum_{n=1}^{\infty}\sum_{g=0}^{\infty}\frac{\hslash^{2g+n-2}}{n!} \int_{v\infty}^z\cdots\int_{v\infty}^z \left( \omega_{g,n}-\delta_{n,2}\delta_{g,0} \frac{dx(z_1) dx(z_2)}{(x(z_1) - x(z_2))^2} \right)\right]\,,
	\end{equation}
	satisfies the differential equation
	\begin{equation}
		\left( P_2\big(e^{\hslash x\frac{d}{dx}}\big) x - P_1\big(e^{\hslash x\frac{d}{dx}}\big) \right) \psi(z;v\infty)=0 \,.
	\end{equation}
\end{theorem}
\begin{proof}
	The proof of the form of the quantum curve proceeds, \textit{mutatis mutandis}, as the proof of \Cref{t:TR/QCess}. That the wavefunction is well-defined follows from \Cref{l:wavefuncexp}.
\end{proof}

\section{The old quantum curve curiosity shop}\label{s:app}
\subsection{Gromov-Witten theory of $\mathbb{P}^1$}\label{ss:P1}
It is well-known that topological recursion produces generating functions for the Gromov-Witten invariants of $\mathbb{P}^1$ from the initial data of the following spectral curve \cite{Z12,NS14,DOSS14}
\begin{equation}\label{e:GWnorm}
  \mathcal{S} = \left(\mathbb{P}^1,\, x(z)=z+\frac{1}{z},\, y(z)=\log(z),\, \omega_{0,2}(z_1,z_2) = \frac{dz_1dz_2}{(z_1-z_2)^2}\right)\,,
\end{equation}
which is a parametrisation of the equation
\begin{equation}
	P(x,y) = x-2\cosh(y) = 0 \,,
\end{equation}
which itself can be recognised as the curve from \Cref{ex:GWnorm}. Here a different parametrisation of $P(x,y) = 0$ is studied, namely that of \Cref{ex:GWus}
\begin{equation}\label{e:GWus}
  \mathcal{S}_\infty = \left(\P,\, x(z)=2\cosh(z),\, y(z)=z,\, \omega_{0,2}(z_1,z_2) = \frac{dz_1dz_2}{(z_1-z_2)^2}\right).
\end{equation}
Denote by $\omega_{g,n}$ the correlators constructed from \eqref{e:GWnorm} and $\omega_{g,n}^\infty$ the correlators constructed from \eqref{e:GWus}.

Initially, one might na\"ively guess that since $x$ and $y$ in the curve \eqref{e:GWus} are just the pullbacks of $x$ and $y$ in the curve \eqref{e:GWnorm} under the exponential map $\pi=\exp$ that the correlators $\omega_{g,n}$ constructed from the initial data of \eqref{e:GWus} will be the pullbacks of the correlators constructed from the initial data of \eqref{e:GWnorm} under $\pi$. This, however, is not the case as $\omega_{0,2}$ is the same for both curves. However, rather surprisingly, the two different spectral curves have the same quantum curve; the quantum curve of \eqref{e:GWnorm} was computed in \cite{DMNPS17} and the quantum curve of \eqref{e:GWus} is shown to agree in the following proposition.

\begin{proposition}\label{p:QCGWP1}
	The quantum spectral curve of \Cref{ex:GWus} is
	\begin{equation}
		\left[\hat{x} - 2\cosh(\hat{y})\right] \psi_\infty(z;\infty) = 0\,,
	\end{equation}
	where $\hat{x}=x=2\cosh(z)$, $\hat{y}=\hslash{d}/{d}x$, and $\psi_\infty(z;\infty)$ is both the limit of $\psi_N(z;\infty)$ and (up to a multiplicative constant) the wavefunction constructed directly in the limit with the integration described in \Cref{l:wavefuncper}.
\end{proposition}
\begin{proof}
	This immediately follows from \Cref{t:TR/QCper}.
\end{proof}

Now turn to the corresponding wavefunctions. To better facilitate the discussion define the $\hslash^\chi$ coefficients in the logarithm of the wavefunction for the two curves with the following notation: the $S^\infty_\chi$ correspond to the curve \eqref{e:GWus} (with $x_\infty(z) = 2\cosh(z)$ and $\omega_{g,n}^\infty$ denoting the correlators of this curve) and the $S_\chi$ to the curve \eqref{e:GWnorm}.
\begin{equation}
	\begin{split}
		\chi &= -1: \\
		S^\infty_{-1}(z) &\coloneq -dx_\infty(z)/dz + zx_\infty(z), \qquad
		S_{-1}(z) \coloneq z^{-1}-z + (z^{-1}+z)\log(z) \\
		\chi &= 0: \\
		S^\infty_{0}(z) &\coloneq \frac{1}{2!} \int_{-\infty}^z\int_{-\infty}^z \left( \frac{dz_1dz_2}{(z_1-z_2)^2} - \frac{dx_\infty(z_1)dx_\infty(z_2)}{(x_\infty(z_1)-x_\infty(z_2))^2} \right), \qquad
		S_{0}(z) \coloneq S^\infty_{0}(\log(z)) \\
		\chi &\geq 1: \\
		S^\infty_{\chi}(z) &\coloneq \sum_{ \substack{ 2g+n-2 = \chi \\ (g,n) \in \mathbb{Z}_{\geq 0} \times \mathbb{Z}_{\geq 1} } } \frac{1}{n!} \stackrel{ n\text{-integrals} } { \overbrace{\int_{-\infty}^z\dots\int_{-\infty}^z} } \omega^\infty_{g,n}, \qquad
		S_{\chi}(z) \coloneq \sum_{ \substack{ 2g+n-2 = \chi \\ (g,n) \in \mathbb{Z}_{\geq 0} \times \mathbb{Z}_{\geq 1} } } \frac{1}{2^nn!} \stackrel{ n\text{-integrals} } { \overbrace{\int_{z^{-1}}^z\dots\int_{z^{-1}}^z} } \omega_{g,n}.
	\end{split}
\end{equation}
For a justification regarding the definition of the $S_\chi(z)$ see \cite{DMNPS17}. It is worth noting that the definition of $S_0(z)$ does not agree with what one would expect from \eqref{e:wfuncconj}, and is instead defined in such a way that it does indeed satisfy the relation $S^\infty_0(z) = \pi^*S_0(z)$.\footnote{The definition given in \cite{DMNPS17} actually differs by a constant, but this is irrelevant.} In fact, the following proposition holds.

\begin{proposition}\label{p:GWagree}
	The wavefunctions of the spectral curves $\mc{S}$ and $\mc{S}_\infty$ agree up to a pulback, \textit{i.e.}
	\begin{equation}
		e^C \exp\left[\sum_{\chi = -1}^{\infty}\hslash^{\chi}S^\infty_{\chi}(z)\right] = \pi^*\exp\left[\sum_{\chi = -1}^{\infty}\hslash^{\chi}S_{\chi}(z)\right]\,,
	\end{equation}
	where $\pi$ is the exponential map and $C\in \C[[\hslash]]$.
\end{proposition}
\begin{proof}
	The quantum curve that kills both wavefunctions is the same, up to pullback via $\pi$. Furthermore, the relation $\omega_{0,1}^\infty = \pi^*\omega_{0,1}$ holds so the $\hslash^{-1}$ term in both expansions coincides.
\end{proof}

In \Cref{Comp_GWp1} the $\chi=1$ case is computed explicitly for both $\mc{S}$ and $\mc{S}_\infty$; it is indeed found that $S_1^\infty(z) = \pi^*S_1(z)$. To the knowledge of the author, this is the first time two different collections of correlators have been shown to have the same quantum curve and wavefunction. 

Now, a natural question to ask is whether the expansion coefficients of the $\omega_{g,n}^\infty$ have a simple expression in terms of Gromov-Witten invariants of $\P^1$ as do the $\omega_{g,n}$. Although some interesting relations between the $\omega_{g,n}^\infty$ and the $\omega_{g,n}$ are presented in \Cref{s:pullback}, a good answer to this question has not been obtained.

\subsection{Gromov-Witten theory of framed $\C^3$}\label{ss:C3}

Consider the family of curves depending on an integer $f\in\mathbb{Z}$ called the `framing'
\begin{equation}
	P(x,y)=-e^{(f+1)xy}+e^{fxy}-x=0 \,.
\end{equation}
The parametrisation of this curve
\begin{equation}\label{e:CYthem}
	\mathcal{S} = \left(\C\setminus[0,\infty),\, x(z) = z^f(1-z),\, y(z) = \log(z)/x(z) ,\, \omega_{0,2}(z_1,z_2) = \frac{dz_1dz_2}{(z_1-z_2)^2}\right),
\end{equation}
is known to calculate the Gromov-Witten invariants of framed $\C^3$ when expanded around $z=1$ in the variable $e^{x(z)}$ \cite{BM08,BS12,Z12}. As in the previous section, in order to apply the results of this paper, a different parametrisation is used, one without $\log$-cuts, namely that of \Cref{ex:GWtorus}
\begin{equation}\label{e:CYus}
  \mathcal{S}_\infty = \left(\P^1,\, x(z)=e^{fz}(1-e^{z}),\, y(z)=z,\, \omega_{0,2}(z_1,z_2) = \frac{dz_1dz_2}{(z_1-z_2)^2}\right)\,,
\end{equation}
so that $\mathcal{S}_\infty$ is a periodic $\C^*$-admissible curve for $f\in\Q\setminus\{0,-1\}$.\footnote{In the $f=0,-1$ cases the only ramification point of $x$ is the one at infinity, and this point does not contribute to the topological recursion. Curiously, taking the $f\to 0,-1$ limit results in non-zero correlators (see \Cref{Comp_C3}).}

Inspired by the `miracle' that occurred in \Cref{ss:P1}, the quantum curve with basepoint $b=\infty$ will be computed. This is the content of the following proposition.
\begin{proposition}\label{p:QCGWtor}
  The quantum spectral curve of \eqref{e:CYus} (\Cref{ex:GWtorus}) is
	\begin{equation}
		\left[1-e^{\hat{x}\hat{y}}-\hat{x}e^{-f\hslash}e^{-f\hat{x}\hat{y}}\right]\psi(z;\infty) = 0 \,,
	\end{equation}
	where the integration in the wavefunction is defined in \Cref{l:wavefuncexp}.
\end{proposition}
\begin{proof}
	First note that, as $f\notin \{0,-1\}$, the curve is indeed natural. It is then clear that the curve is a regular exponential curve and thus falls under \Cref{t:TR/QCexp}. One then immediately obtains\footnote{Technically, one should treat the case where $f>0$ separately from the case $f<0$ and multiply $P$ $e^{-fxy}$ in the $f<0$ case to obtain a polynomial in $x$ and $e^{xy}$. However, it is easy to see this gives the same result.}
	\begin{equation}
		\left[-e^{(f+1)\hat{x}\hat{y}}+e^{f\hat{x}\hat{y}}-\hat{x}\right]\psi(z;\infty) = 0 \,,
	\end{equation}
	from which the desired results follows upon multiplying on the left by $e^{-f\hat{x}\hat{y}}$. 
\end{proof}

This is not in agreement with the results of \cite{Z12}\footnote{\cite{Z12} writes this result slightly differently due to a different convention regarding the definitions of $\hat{x}$ and $\hat{y}$.}
\begin{equation}
	\left[1-e^{\hat{x}\hat{y}}-\hat{x}e^{\frac{1}{2}\hslash}e^{-f\hat{x}\hat{y}}\right]\tilde{\psi} = 0 \,,
\end{equation}
for any integer $f$. However, notice that
\begin{equation}
	e^{(f+\frac{1}{2})\hat{x}\hat{y}}\left[1-e^{\hat{x}\hat{y}}-xe^{-f\hslash}e^{-f\hat{x}\hat{y}}\right]e^{-(f+\frac{1}{2})\hat{x}\hat{y}} = \left[1-e^{\hat{x}\hat{y}}-xe^{\frac{1}{2}\hslash}e^{-f\hat{x}\hat{y}}\right]\,,
\end{equation}
so there is a simple relation between the two curves, and the two curves coincide, at least formally, for $f=-1/2$. Furthermore, 
\begin{equation}
  \psi(x) = e^{(f+1/2)\hat{x}\hat{y}} \tilde{\psi(x)} e^{-(f+1/2)\hat{x}\hat{y}} = \tilde{\psi}(e^{\hslash(f+1/2)}x) \,,
\end{equation}
so the two wavefunctions will also coincide for $f=-1/2$, and a very simple relation between them exits for all $f$.

There are no issues defining TR for $f=-1/2$ for the curve \eqref{e:CYus} (see \Cref{r:nintlam}), so one would expect that in this case the two wavefunctions do coincide, presuming the results of \cite{Z12} extend to the non-integer case. Although half-integer framing may seem somewhat irrelevant, it has appeared in the physics literature \cite{DF05}. Indeed, in the situation at hand $f=-1/2$ is somewhat special, as $x(z) = -2\sinh(z/2)$ has `more symmetry' than the generic case. In \Cref{Comp_C3}, it is shown explicitly that the two wavefunctions $\psi$ and $\tilde{\psi}$ do coincide in the $f=-1/2$ case up to order $\hslash^1$, where the $f=-1/2$ correlators of the finite curve $\mc{S}$ are defined via an analytic continuation.

The quantum curve with basepoint $b=0$ of the spectral curve \eqref{e:CYus} can be computed using \Cref{t:TR/QCzero} and the limiting procedures described in this article. The $\hslash$ corrections turn out to be very complicated and it does not appear to be interesting.

To end the discussion, a comment on the is a well-known fact that the $f\to\infty$ limit of the framed mirror curve reproduces the simple Hurwitz curve $P_H(x,y) = y - e^{-xy} = 0$ \cite{CGMPS07,BKMP08,BM08}. At the level of spectral curves, one can the consider the substitution in $P(x,y)=-e^{(f+1)xy}+e^{fxy}-x$ of $(x,y)\to (f^{-1}x,-y)$ so $P(x,y) = f^{-1}(e^{-xy}f(1-e^{-f^{-1}xy})-x)$. Multiplying $P$ by $(f/x)e^{xy}$ and taking the limit $f\to\infty$ one obtains precisely the curve of $P_H$. Interestingly, this limiting procedure seems to jive with the quantum curve of \Cref{p:QCGWtor}. Indeed, performing the substitution $(\hat{x},\hat{y})\to (f^{-1}\hat{x},-\hat{y})$ and taking the $f\to\infty$ limit of the quantum curve in \Cref{p:QCGWtor}, one obtains the quantum curve of the simple Hurwitz numbers with basepoint $b=\infty$ that was found in \cite{W24}.

\section{Relations to spectral curves with finite ramification loci}
\label{s:pullback}

In the prior section, various quantum curve computations suggested a relation between the newly defined topological recursion on spectral curves with infinite ramification loci, and the traditional Eynard-Orantin topological recursion on spectral curves with finite ramification Loci. Here, this relationship is pursued purely from the perspective of topological recursion, without invoking the quantum curve/topological recursion correspondence.

As a guiding example consider the spectral curves of \Cref{ex:GWnorm}
\begin{equation}
  \label{e:nopullback}
	\mathcal{S} = \left(\P^1,\, x(z) = z+1/z,\, y(z) = \log(z),\, \omega_{0,2}(z_1,z_2) = \frac{dz_1dz_2}{(z_1-z_2)^2}\right) \,,
\end{equation}
and \Cref{ex:GWus}
\begin{equation}
  \label{e:yespullback}
	\mathcal{S}_\infty = \left(\P^1,\, x(z) = 2\cosh(z),\, y(z) = z,\, \omega_{0,2}(z_1,z_2) = \frac{dz_1dz_2}{(z_1-z_2)^2}\right)\, .
\end{equation}
The regular Eynard-Orantin topological recursion of \cite{EO07} defines topological recursion on $\mathcal{S}$ whereas \Cref{s:per} defines topological recursion on $\mathcal{S}_\infty$. However, both these curves identically satisfy the entire equation
\begin{equation}
	P(x,y) = x - 2\cosh(y) = 0 \,,
\end{equation}
and, in fact, $\mathcal{S}_\infty$ looks as if it is the \emph{pullback} of $\mathcal{S}$ under the exponential map. However, not all the objects in the spectral curve \eqref{e:nopullback} are pulled back under the exponential map. Indeed, the bidifferential $\omega_{0,2}$ is identical for both $\mathcal{S}$ and $\mathcal{S}_\infty$. The following (essentially obvious) proposition gives a relation between two curves where $\omega_{0,2}$ is also pulled back.

\begin{proposition}\label{p:pullback}
	Let $\mathcal{S} = \left(\Sigma,\, x,\, (A,B),\, y,\, \omega_{0,2}\right)$ be an admissible spectral curve. Then let $\pi$ be a branched covering from a finite disjoint union of Riemann surfaces $\tilde{\Sigma}$ to $\Sigma$ such that $A$ does not contain any branchpoints of $\pi$. Given a set $\tilde{A}\subset\tilde{\Sigma}$ such that $\pi|_{\tilde{A}}:\tilde{A}\to A\cap R_0$ is a bijection define the spectral curve
	\begin{equation}
		\tilde{\mathcal{S}} \coloneq \left(\tilde{\Sigma},\, \tilde{x}\coloneq \pi^*x,\, (\tilde{A},\tilde{B}\coloneq\pi_1^*\pi_2^*B),\, \tilde{y} \coloneq \pi^*y,\, \tilde{\omega}_{0,2}\coloneq\pi^*_1\pi^*_2\omega_{0,2}\right) \,,
	\end{equation}
	where the notation $\pi_i^*$ means the pullback in the $i$th argument of an multi-differential. Then denote by $\{\omega_{g,n}\}_{(g,n)\in\Z_{\geq 0}\times\Z_{\geq 1}}$ the system of correlators constructed by topological recursion on $\mc{S}$ and $\{\tilde{\omega}_{g,n}\}_{(g,n)\in\Z_{\geq 0}\times\Z_{\geq 1}}$ the system of correlators constructed by topological recursion on $\tilde{\mc{S}}$. The two systems of correlators are then related via the pullback
	\begin{equation}
		\tilde{\omega}_{g,n} = \pi_1^*\pi_2^*\cdots\pi_n^*\omega_{g,n} \,.
	\end{equation}
\end{proposition}
\begin{proof}
	This is proven via induction on the negative of the Euler characteristic $2g+n-2$. The cases $2g+n-2=-1,0$ hold by definition, and so the proof proceeds directly to the induction step. Pick an $a\in A\cap R_0$ and then let $\tilde{a}$ be the unique element of $\tilde{A}$ such that $\pi(\tilde{a})=a$. Let $\sigma_a$ be a local deck transformation of $x$ about $a$ and note that there exists a local partial inverse of $\pi$, satisfying $\pi^{-1}(a)=\tilde{a}$, so that $\tilde{\sigma}_a\coloneq \pi^{-1}\circ\sigma_a\circ\pi$ is a local deck transformation of $\tilde{x}$ about $\tilde{a}$. Indeed, this process induces an isomorphism between the two local deck transformation groups and therefore a bijection between the sets $\mf{f}_a(w)$ and $\mf{f}_{\tilde{a}}(z)$ where $w=\pi(z)$. Furthermore, given a function or differential $g$ one can note that
	\begin{equation}
		\pi^*\big(g\circ\sigma_a\big)(z) = \big(\pi^*g\big)(\tilde{\sigma}_a(z)) \,.
	\end{equation}
	With this knowledge in mind a simple computation establishes the induction step
	\begin{equation}
		\begin{split}
			\pi^*_0\pi^*_1\cdots\pi^*_n\omega_{g,n+1}(z_0, z_{\llbracket n\rrbracket}) &= \sum_{a \in A \cap R_0} \Res_{w = a} \sum_{\emptyset \neq Z \subseteq \mf{f}'_a(w)} \pi^*_0K_{|Z|+1}(z_0,w, Z) \pi^*_1\cdots\pi^*_n \mc{W}_{g,n,|Z|+1} (w, Z \mid z_{\llbracket n \rrbracket}) \\
			&=\sum_{\tilde{a} \in \tilde{A}} \Res_{z = \tilde{a}} \sum_{\emptyset \neq Z \subseteq \mf{f}_{\tilde{a}}'(z)} \tilde{K}_{|Z|+1}(z_0,z, Z) \tilde{\mc{W}}_{g,n,|Z|+1} (z, Z \mid z_{\llbracket n \rrbracket}) \\
			&=\tilde{\omega}_{g,n+1}(z_0,z_{\llbracket n \rrbracket} )\,,
		\end{split}
	\end{equation}
	where $\tilde{\mc{W}}_{g,n,|Z|+1}$ and $\tilde{K}_{|Z|+1}$ are defined in the obvious manner for the curve $\tilde{\mc{S}}$.
\end{proof}

The preceding proposition, then, justifies the succeeding definition.

\begin{definition}\label{d:pullback}
	Given a spectral curve $\mathcal{S} = \left(\Sigma,\, x,\, (A,B),\, y,\, \omega_{0,2}\right)$ and a branched covering $\pi$ from a finite disjoint union of Riemann surfaces $\tilde{\Sigma}$ to $\Sigma$ such that $A$ does not contain any branchpoints of $\pi$, define the pullback of $\mc{S}$ under $\pi$ as
	\begin{equation}
		\pi^*\mc{S} \coloneq \left(\tilde{\Sigma},\, \tilde{x}\coloneq \pi^*x,\, (\tilde{A},\tilde{B}\coloneq\pi_1^*\pi_2^*B),\, \tilde{y} \coloneq \pi^*y,\, \tilde{\omega}_{0,2}\coloneq\pi^*_1\pi^*_2\omega_{0,2}\right)\,,
	\end{equation}
	where $\tilde{A}\subset\tilde{\Sigma}$ has the property that $\pi|_{\tilde{A}}:\tilde{A}\to A\cap R_0$ is a bijection. Note that, the topological recursion on $\pi^*\mc{S}$ will be independent of the choice of $\tilde{A}$, by \Cref{p:pullback}, and so not making this choice explicit in the notation $\pi^*\mc{S}$ is justified.
\end{definition}

However, if the goal is to relate the curves $\mathcal{S}$, of \cref{e:nopullback}, to the curve $\mc{S}_\infty$, of \cref{e:yespullback}, then, after applying the above proposition, there are still three difficulties to overcome: first, the bidifferential $B$ should not be pulled back, but should remain the same; second, rather than just include a subset of the ramification points $\tilde{A}$, the polarisation should include the whole of $\pi^{-1}(A)$; third, $\omega_{0,2}$ should remain the same. The next proposition demonstrates that the first and second issue cancel each other out.

\begin{proposition}\label{p:ncan_pullback}
	Let $\mathcal{S} = \left(\Sigma,\, x,\, (A,B),\, y, \omega_{0,2}\right)$ be an admissible globally polarised spectral curve and denote by $\{\omega_{g,n}\}_{(g,n)\in\Z_{\geq 0}\times\Z_{\geq 1}}$ the system of correlators constructed by topological recursion on $\mc{S}$. Let $\pi$ be a branched covering from $\Sigma$ to itself such that $A$ does not include any branchpoints of $\pi$ and $\pi^{-1}(w)$ is a finite set for every $w\in\Sigma$. Then define the spectral curve
	\begin{equation}
		\bar{\mathcal{S}} = \left(\Sigma,\, \bar{x}\coloneq \pi^* x,\, (\bar{A}, \bar{B})\coloneq (\pi^{-1}(A),B),\, \bar{\omega}_{0,1}\coloneq \pi^*\omega_{0,1}, \bar{\omega}_{0,2}\coloneq \pi_1^*\pi_2^*\omega_{0,2}\right)\,,
	\end{equation}
	and denote by $\{\bar{\omega}_{g,n}\}_{(g,n)\in\Z_{\geq 0}\times\Z_{\geq 1}}$ the system of correlators constructed by topological recursion on $\bar{\mc{S}}$. The following relation then holds
	\begin{equation}
		\bar{\omega}_{g,n} = \pi_1^*\pi_2^*\cdots\pi_n^*\omega_{g,n} \,.
	\end{equation}
\end{proposition}
\begin{proof}
	Let $\tilde{\mc{S}} \coloneq \pi^*\mc{S}$ be as in \Cref{p:pullback}. The strategy will be to show $\bar{\omega}_{g,n}$ and $\tilde{\omega}_{g,n}$ are, in fact, equal. The proof critically relies on the following property of a globally polarised curves bidifferential $B$ \cite{EO07}
	\begin{equation}
    \pi_1^*\pi_2^*B(z_1,z_2) = \pi_2^*{\pi_2}_*B(z_1,z_2) = \pi_1^*{\pi_1}_*B(z_1,z_2) = \sum_{i=1}^m B(\nu_i(z_1),z_2) \,,
	\end{equation}
	where we have denoted the degree of $\pi$ as $m$ and defined the $m$ distinct maps $\nu_i$ through $\pi\circ\nu_i=\pi$. Of course, these $\nu_i$ may not be individually well-defined, due to a non-trivial action of the monodromy group, but any sum over them is. Next, choosing a $\tilde{A}$ as described in \Cref{d:pullback}, note 
	\begin{equation}
		\bigcup_{i=1}^m \nu_i(\tilde{A}) = \pi^{-1}(A) = \bar{A} \,. 
	\end{equation}
	Now, claim inductively on $2g+n-2$ that $\tilde{\omega}_{g,n} = \bar{\omega}_{g,n}$ and, noting the base cases $(g,n)=(0,1),(0,2)$ are trivial, compute
	\begin{equation}
		\begin{split}
			\tilde{\omega}_{g,n+1}(z_0,z_{\llbracket n \rrbracket}) &= \sum_{a \in \tilde{A}} \Res_{z = a} \sum_{\emptyset \neq Z \subseteq \mf{f}_{\tilde{a}}'(z)} \tilde{K}_{|Z|+1}(z_0,z, Z) \tilde{\mc{W}}_{g,n,|Z|+1} (z, Z \mid z_{\llbracket n \rrbracket}) \\
			&= \sum_{a \in \tilde{A}} \Res_{z = a} \sum_{\emptyset \neq Z \subseteq \mf{f}_{\tilde{a}}'(z)} \left(\int_a^z\sum_{i=1}^m B(\nu_i(z'),z_1)\right) \frac{\bar{\mc{W}}_{g,n,|Z|+1} (z, Z \mid z_{\llbracket n \rrbracket})}{\prod_{z' \in Z}\big(\bar{\omega}_{0,1}(z)-\bar{\omega}_{0,1}(z')\big)} \\
			&= \sum_{a \in \tilde{A}}\sum_{i=1}^m \Res_{z = \nu_i(a)} \sum_{\emptyset \neq Z \subseteq \mf{f}_{\tilde{a}}'(z)} \left(\int_a^z B(z',z_1)\right) \frac{\bar{\mc{W}}_{g,n,|Z|+1} (z, Z \mid z_{\llbracket n \rrbracket})}{\prod_{z' \in Z}\big(\bar{\omega}_{0,1}(z)-\bar{\omega}_{0,1}(z')\big)} \\
			&= \bar{\omega}_{g,n+1}(z_0,z_{\llbracket n \rrbracket}) \,,
		\end{split}
	\end{equation}
	where we were able to go from the second line to the third as the integrand in $z$, except for $B$, is a pullback in $\pi$ and so has the same value at $z$ and $\nu(z)$.
\end{proof}

\begin{corollary}\label{c:exppullback}
	The above proposition holds when $\pi$ is a branched covering from $\P^1$ to $\P^1\setminus\{z=0\}$ given by $\pi(z) = e^z$ where $z$ is an affine coordinate.
\end{corollary}
\begin{proof}
	First, the only choice $B$ on $\P^1$ so that $\mc{S}$ is globally polarised is $B(z_1,z_2) = dz_1dz_2/(z_1-z_2)^2$ \cite{EO07}. Then note
	\begin{equation}
		\pi_1^*\pi_2^*B(z_1,z_2) = \frac{e^{z_1}e^{z_2}dz_1dz_2}{(e^{z_1}-e^{z_2})^2} = \frac{1}{4}\csch^2\left(\frac{z_1-z_2}{2}\right) = \sum_{i=-\infty}^{\infty} \frac{dz_1dz_2}{(z_1-z_2-2\pi i)^2} = \pi^*\pi_*B(z_1,z_2)\,.		
	\end{equation}
	Thus, the central property of $B$ still holds in this case. Noticing that convergence in the infinite sum is absolute, the interchange of the sum over $i$ with the integration over $z'$ in the proof of \Cref{p:ncan_pullback} is valid, and so the result will also hold for the exponential map. 
\end{proof}

\begin{remark}
	The above proposition should remain true for any Riemann surface $\Sigma$ with $\pi(z) = M_0(z) e^{M_1(z)}$ a transalgebraic function. This can be shown using the asymptotics of \cite[Lemma 3.10]{BKW24}. 
\end{remark}

\begin{remark}
	It is important to note a certain subtlety when $y$ has branch cuts that are `related' to the branch cuts $\pi^{-1}$. In this case one needs to take care in `cancelling' branch cuts in expressions involving $y\circ\pi^{-1}$. For example, if $\pi(z) = e^z$ and $y(z) = \log(z)$ then the recursion kernel will contain factors of the form
	\begin{equation}
		\left(\log(e^z) - \log(e^{\log(\sigma(e^z))})\right)^{-1}\,,
	\end{equation}
	where $\sigma$ is a deck transformation of the original (not pulled back) $x$ and some of the $\log$'s may be chosen to be different branches. To determine the correct branch choice, one should note that these factors have to give a pole at $a$. For example, if $x(z) = z+1/z$ so $\pi^*x = 2\cosh(z)$ around a ramification point $\pi m$, noting $\sigma(z) = 1/z$, the branches are chosen such that
	\begin{equation}
		\left(\log(e^z) - \log(e^{\log(\sigma(e^z))})\right)^{-1} = (z+z-2\pi m)^{-1}\,.
	\end{equation}
In general, this will result in $y$ having no consistent, global, definition as a meromorphic function; one will need to define it locally around each ramification point to obtain the correct branch choice.
\end{remark}

The next lemma provides, in particular, an explicit relation between the genus zero correlators the curves \cref{e:nopullback} and \cref{e:yespullback}.

\begin{lemma}\label{l:pullback}
	Using the notation of the previous proposition the following relation holds for $n\in\Z_{\geq 0}$
	\begin{equation}
		\bar{\omega}_{0,n+1}(z,z_{\llbracket n \rrbracket}) = \sum_{z'_1 \in \pi^{-1}(\pi(z_1))} \sum_{z'_2 \in \pi^{-1}(\pi(z_2))} \cdots \sum_{z'_n \in \pi^{-1}(\pi(z_n))} \hat{\omega}_{0,n+1}(z,z'_{\llbracket n \rrbracket})\,,
	\end{equation}
	where the $\hat{\omega}_{g,n}$ are constructed from the spectral curve
	\begin{equation}
		\hat{\mathcal{S}} = \left(\Sigma,\, \hat{x}\coloneq \pi^* x,\, (\hat{A}, \hat{B})\coloneq (\pi^{-1}(A),B),\, \hat{y}\coloneq \pi^*y, \hat{\omega}_{0,2}\coloneq \omega_{0,2}\right)\,.
	\end{equation} 
\end{lemma}
\begin{proof}
  The $n=0$ case holds by definition, and the $n=1$ case is a well-known property of $\omega_{0,2}$ \cite{EO07}. Proceeding via induction on $n$ observe that, by the induction assumption
	\begin{equation}
		\sum_{z'_1 \in \pi^{-1}(\pi(z_1))} \sum_{z'_2 \in \pi^{-1}(\pi(z_2))} \cdots \sum_{z'_n \in \pi^{-1}(\pi(z_n))} \bar{\mc{W}}_{0,n,|Z|+1} (z, Z \mid z'_{\llbracket n \rrbracket}) = \hat{\mc{W}}_{0,n,|Z|+1} (z, Z \mid z'_{\llbracket n \rrbracket})\,,
	\end{equation}
	as $\mc{W}_{0,n,|Z|+1} (z, Z \mid z'_{\llbracket n \rrbracket})$ only contain genus zero correlators with one argument in $z\cup Z$ and the rest in $z'_{\llbracket n \rrbracket}$ by \Cref{d:combcomb}. The desired result immediately follows from this observation and the definition of topological recursion. 
\end{proof}

There is also a simple relation that holds for $\omega_{1,1}$.

\begin{lemma}
  Using the notation of the previous proposition and lemma the following relation holds
  \begin{equation}
    \bar{\omega}_{1,1}(z_1) = \hat{\omega}_{1,1}(z_1) + \frac{1}{2} \sum_{a\in A} \Res_{w_1=a}\Res_{w_2=a} \omega_{0,3}(z_1,w_1,w_2) \phi(w_1,w_2) \,,
  \end{equation}
  with 
  \begin{equation}
    \phi(w_1,w_2) = \int_{b_1}^{w_1}\int_{b_2}^{w_2} \big( \bar{\omega}_{0,2} - \hat{\omega}_{0,2} \big) \,,
  \end{equation}
  where the choice of basepoints $b_1,b_2 \in \Sigma$ does not matter as $\omega_{0,3}$ is residueless and the choice of branch for $\phi$ does not matter for the same reason.
\end{lemma}
\begin{proof}
  \begin{equation}
    \begin{split}
      \bar{\omega}_{1,1}(z_1) =& \sum_{a\in A} \Res_{z=a} \sum_{\sigma(z) \in f'_a(z)} K_{2}(z_1,z,\{\sigma(z)\}) \bar{\omega}_{0,2}(z,\sigma(z)) \\
      =&  \sum_{a\in A} \Res_{z=a} \sum_{\sigma(z) \in f'_a(z)} K_{2}(z_1,z,\{\sigma(z)\}) \left( \hat{\omega}_{0,2}(z,\sigma(z)) + \vphantom{\frac12}\right.\\ 
       &\left.\frac{1}{2}\Res_{w_1=z}\Res_{w_2=\sigma(z)}\phi(w_1,w_2) \left[\hat{\omega}_{0,2}(z,w_1)\hat{\omega}_{0,2}(\sigma(z),w_2) + \hat{\omega}_{0,2}(z,w_2)\omega_{0,2}(\sigma(z),w_1)\right] \right) \\
      =& \hat{\omega}_{1,1}(z_1) + \sum_{a_1\in A} \Res_{w_1=a_1}\Res_{w_2=a_1} \phi(w_1,w_2) \sum_{a\in A} \delta_{a,a_1} \Res_{z=a} \sum_{\sigma(z) \in f'_a(z)} \\ 
       &\left. \frac{1}{2} K_2(z_1,z,\{\sigma(z)\}) \left[\hat{\omega}_{0,2}(z,w_1)\hat{\omega}_{0,2}(\sigma(z),w_2) + \hat{\omega}_{0,2}(z,w_2)\hat{\omega}_{0,2}(\sigma(z),w_1)\right]\right. \\
      =& \hat{\omega}_{1,1}(z_1) + \frac{1}{2} \sum_{a_1\in A}\Res_{w_1=a_1}\Res_{w_2=a_1} \phi(w_1,w_2) \hat{\omega}_{0,3}(z_1,w_1,w_2) \,,
    \end{split}
  \end{equation}
  where the first equality is the definition of $\bar{\omega}_{1,1}$, the second uses the fact that $\hat{\omega}_{0,2}$'s only pole is a double pole on the diagonal with biresidue one, the third equality uses the fact that $\bar{\omega}_{0,2}-\hat{\omega}_{0,2}$ has no pole on the diagonal (so no other poles are picked up when commuting the residues), and the fourth equality uses the fact that $B$ has no off-diagonal poles as well as the fact about $\hat{\omega}_{0,2}$'s pole structure to get rid of the Kronecker delta, $\delta_{a,a_1}$, as the terms with $a_1\neq a$ are zero anyway. Finally, there is a slight subtlety with commuting the residues involving the sum over $f'_a(z)$. To really do this correctly, one needs to put the residue at $z=a$ inside the aforementioned sum, commute the residues in $w_1,w_2$ with the residue in $z$ inside the sum, and then proceed. However, the integrand is not necessarily a well-defined function of $z$ before summing over $f'_a(z)$; this issue can easily be resolved by either going to local coordinates or treating the $\sigma(z)$ as being defined by the first couple of terms in their expansions about $a$.
\end{proof}

\begin{remark}
  One would like to try and extend the above results to all correlators. In fact, both of the above results are very redolent of the formula to go from non-blobbed TR to blobbed TR \cite{BS17}. Indeed, one can rewrite the sum over preimages as residues of correlators with the integral of the difference $\bar{\omega}_{0,2}-\hat{\omega}_{0,2}$; this difference acts like a `$\phi_{0,2}$' holomorphic part of $\omega_{0,2}$. However, naïvely trying to extend the formulas for blobbed TR to this case does not work, not even for the genus zero case, and the appropriate adjustments for the higher genus case are not clear to the author.

  Similarly, the above result is also redolent of the change of polarisation formula given in \cite{EO15}. Again, naïve application of this formula does not work and the required modifications are not clear to the author.
\end{remark}

\begin{remark}\label{r:qas}
  One can rephrase this process, at least for simple curves, in the ABCD language of \cite{KS17,ABCO24}. Changing just the $\omega_{0,2}$ with one argument in $z_{\llbracket n \rrbracket}$ corresponds to symmetrising over indices in the $A$ and $B$ tensors. Changing the $\omega_{0,2}$ evaluated at $z$ and an element of $Z$ (so for simple curves: just the $\omega_{0,2}$ appearing in the formula for $\omega_{1,1}$) corresponds to a shift in the $D$ tensor and therefore a change in the quantisation of the classical Airy structure. The $C$ tensor never changes in these processes. 
\end{remark}

\begin{remark}
	In \cite{AC24}, a spectral curve was pulled back via the map $\varpi(z) = (\pi b^2)^{-1}\sin(\pi b z)$\footnote{A suitable generalisation of \Cref{p:ncan_pullback}, proved in an identical manner to \Cref{c:exppullback}, would have to be applied here.} where $b^2 = 2/(2p+1)$ and $p\in\N$. The results here would allow one to restore the canonical form of $B$ at the cost of enlarging the $A$ to include all ramifications points of $\varpi^* x$.
\end{remark}

\section{Conclusion and outlook}

In this paper two new classes of spectral curves were defined and studied: the periodic curves and the exponential curves. The topological recursion was then defined for such curves and it was established that many of the most salient features of the original Eynard-Orantin framework carried over to this more general case. This more general framework was then applied to the study of quantum curves, where it was shown that the QC/TR connection could be proven for a large collection of periodic and exponential curves. Curiously, in multiple cases, this produces quantum curves which can also be obtained from curves that fall under the traditional Eynard-Orantin framework.

Many open questions remain (there is some intersection here with the 'Further Directions' section of \cite{AH26}).

\begin{itemize}
	\item In two examples of periodic and exponential curves, the quantisation produced by the parametrisation used here gives the identical results to the standard parametrisation used elsewhere in the literature. The results in \Cref{s:pullback,Npol} give some indications as to the connections between these curves, but if there is a nicer, general, result it would be good to derive it.
	\item Only genus zero curves were quantised here, but many interesting exponential curves that are believed to be related to knot theory are higher genus curves . In principle, \Cref{t:seqexp} allows one to construct a sequence of curves $\mc{S}_N$ that can be quantised for each $N$ according to the results of \cite{EGMO21}; in principle the limit as $N\to\infty$ can then be taken as in the genus zero case. In practice, this is a non-trivial computation even for simpler examples. It would be intriguing to see if this program could carried out successfully for any higher genus exponential curve. Using the traditional, Eynard-Orantin, parameterisation of such curves there is a conjectured relationship with knot theory when the spectral coincides with the $A$-polynomial of a given knot \cite{ABM12,BE12,GJKS14}. However, quantisations of these curves (with the traditional parameterisation) are not known to exist. Given the genus zero results (or coincidences) in the present work, this could be an interesting line of study.
	\item None of the curves studied in this work were explicitly given enumerative interpretations. As the wavefunctions of some coincide with wavefunctions that have known enumerative interpretations\footnote{Up to pulling back by the exponential map.} (see \Cref{ss:P1,ss:C3}), it is intuitively plausible that there should be one.
\item The concept of a non-canonically polarised spectral curve was introduced, but not really explored. Some initial results were obtained in \Cref{Npol}, but this is a far cry from a systematic study. In particular, a relation to quantum airy structures and existing classification results \cite{KS17,ABCO24,BBCCN24,BKS24} should be undertaken.
\end{itemize}

\appendix

\section{Non-canonically polarised spectral curves}\label{Npol}

Prior to the present work, no study of non-canonically polarised spectral curves has been conducted in the literature. It is likely that such curves are, in principal, already defined through the framework of quantum airy structures (see \Cref{r:qas}) and a complete understanding merely involves correctly relating non-canonically polarised TR to known classification results of quantum airy structures (see \cite{ABCO24,BBCCN24,BKS24} for such classification results). Although it is not the main point of focus here, this appendix provides some rudimentary steps (mainly examples) towards understanding these curves from a purely TR-based perspective, using what has been established in \Cref{s:pullback}.

\begin{proposition}
	Let
	\begin{equation}
		\mc{S} = \left(\Sigma,\, x,\, (A,B=\omega_{0,2}),\, y,\, \omega_{0,2}\right)\,,
	\end{equation}
	be an admissible, canonically polarised, spectral curve. Let $\pi$ be a branched covering from $\Sigma$ with a ramification locus disjoint from $A$. Then the correlators computed by TR of the curve $\pi^* \mc{S}$ satisfy all the properties of \Cref{t:origprop}.
\end{proposition}
\begin{proof}
	Symmetry, vanishing residues, and homogeneity, follow directly from \Cref{p:pullback}. The special formula for $\omega_{0,3}$ is just a special case of the general formula for taking the residue of a pullback.
\end{proof}

Spectral curves that are the pullback of an admissible curve will result in TR that computes a solution of the higher abstract loop equations in the sense of \Cref{t:loop}. Thus, these curves should also, presumably, be defined as admissible.

Using \Cref{p:ncan_pullback} and the previous proposition, one can actually derive large classes of admissible curves which have $\tilde{A}$ as the entire ramification locus ($\tilde{A}$ being the same as in \Cref{d:pullback}), but are not canonically polarised. For example, consider the following easy corollary.

\begin{corollary}
	Let $\pi$ be a finite degree branched covering from a Riemann surface $\Sigma$ to itself, $f:\Sigma\to\Sigma$ be a meromorphic function with poles at all branchpoints of $\pi$, $g:\Sigma\to\Sigma$ be a meromorphic function such that $dg$ is non-zero at all branchpoints of $\pi$, and pick a symmetric bidifferential $B$ on $\Sigma^2$ whose only pole is a double pole on the diagonal with biresidue one. Denote the ramification points of $f$ as $R_f$ and the ramification points of $\pi$ as $R_\pi$. Then the topological recursion applied to the globally polarised spectral curve
	\begin{equation}
		\mc{S} = \left(\Sigma,\, x = \pi^*f,\, (R_f\cup R_\pi, B),\, y = \pi^*g,\, \omega_{0,2} = \pi_1^*\pi_2^*\omega_{0,2}\right)\,,
	\end{equation}
	produces correlators that satisfy the properties of \Cref{t:origprop}, save for the special formula for $\omega_{0,3}$.
\end{corollary}
\begin{proof}
	Follows directly from the previous proposition, \Cref{p:ncan_pullback}, and noting that the ramification points in $R_\pi$ do not contribute to the topological recursion due to both the non-vanishing condition on $dg$ and the requirement on the location of poles of $f$ so $\pi^*df$ has at least a pole of order $\Mult_a(x)\geq 2$ at $a\in R_\pi$.
\end{proof}

This corollary lets one quickly formulate large classes of natural-looking non-canonically polarised curves that produce correlators with all the properties one would expect from TR, as demonstrated in the following example.

\begin{example}
	Fix $n\in\Z\setminus\{0,1\}$, let $R_n = \{\exp[\pi i m/(n-1)]\, |\, m=0,1\dots,2n-3 \}\cup\{0,\infty\}$, take $m\in Z$ such that $m+1$ is coprime with $n$ and $|m+1|<|n|$, and let $c\in\C^*$. Then the topological recursion applied to the spectral curve
	\begin{equation}
		\mc{S} = \left(\P^1,\, x(z) = z^n+z^{-n},\, \left(R_n,\, \frac{dz_1dz_2}{(z_1-z_2)^2}\right),\, y(z) = cz^n,\, \omega_{0,2} = \frac{n^2z_1^{n-1}z_2^{n-1}dz_1dz_2}{(z_1^n-z_2^n)^{2}}\right) \,,
	\end{equation}
  where $z$ is an affine coordinate, produces correlators that satisfy the properties of \Cref{t:origprop} (except the special formula for $\omega_{0,3}$, which would have to be modified). Indeed, consider the pullback of
	\begin{equation}
		\tilde{\mc{S}} = \left(\P^1,\, x(z) = z+z^{-1},\, \left(\{\pm 1\},\, \frac{dz_1dz_2}{(z_1-z_2)^2}\right),\, y(z) = (c/n)z,\, \omega_{0,2} = \frac{dz_1dz_2}{(z_1-z_2)^{2}}\right)\,,
	\end{equation}
	under the map $\pi(z) = z^n$.
\end{example}

Of course, the above example looks somewhat contrived; it's essentially a pullback. One could start with an appropriate choice of $y(z) = z^{1/n}$, then after the pullback $y$ would be $z$ multiplied by an $n^{\text{th}}$ root of unity; the root of unity would have to change depending on which ramification point $z$ was near to match the original branch structure of $z^{1/n}$.

Curiously, other related examples such as
\begin{equation}
	\mc{S} = \left(\P^1,\, x(z) = z^2+z^{-2},\, \left(R_2,\, \frac{dz_1dz_2}{(z_1-z_2)^2}\right),\, y(z) = cz,\, \omega_{0,2} = \frac{2(z_1^2+z_2^2)dz_1dz_2}{(z_1^2-z_2^2)^{2}}\right) \,,
\end{equation}
also seem to produce symmetric correlators even though $\omega_{0,2}$ is not a pullback of $B$.

\section{Computations of correlators}\label{Comp}

Here some correlators for a variety of admissible periodic and exponential spectral curves are computed. The calculations were done partially by hand, and partially with SageMath \cite{Sage}. The chosen examples are then checked for consistency with various results obtained in this work.

\subsection{$ P(x,y) = x - 2\cosh(y) $}
\label{Comp_GWp1}
Consider the periodic, admissible, and regular spectral curve of \Cref{ex:GWus}
\begin{equation}
	\mathcal{S}_\infty = \left(\mathbb{P}^1,\, x(z) = 2\cosh(z),\, y(z) = z,\, B(z_1,z_2) = \frac{dz_1dz_2}{(z_1-z_2)^2}\right),
\end{equation}
where $z$ is an affine coordinate on $\mathbb{P}^1$. The unstable correlators are
\begin{equation}
\omega_{0,1}(z_1) = y(z_1)dx(z_1) = 2z_1\sinh(z_1),\qquad \omega_{0,2}(z_1,z_2) = B(z_1,z_2) = \frac{dz_1dz_2}{(z_1-z_2)^2}.
\end{equation}
Begin by a calculation of the simplest stable correlator $\omega_{0,3}$. Luckily, by \Cref{c:perprop} there exists a simple formula for $\omega_{0,3}$
\begin{equation}
	\begin{split}
		\omega_{0,3}(z_1,z_2,z_3) &= \sum_{a\in R_0} \Res_{z=a} \frac{ \omega_{0,2}(z,z_1) \omega_{0,2}(z,z_2) \omega_{0,2}(z,z_3) } { dx(z) dy(z) }\\
		&= \frac{dz_1dz_2dz_3}{2} \sum_{k=-\infty}^{\infty} \Res_{z = k\pi i} \frac{\csch(z)dz}{(z-z_1)^2(z-z_2)^2(z-z_3)^2}\\
		&= \frac{dz_1dz_2dz_3}{2} \sum_{k=-\infty}^{\infty} \frac{(-1)^k}{(z_1-k\pi i)^2(z_2-k\pi i)^2(z_3-k\pi i)^2}
	\end{split}
\end{equation}

Next, compute $\omega_{1,1}$
\begin{equation}
	\begin{split}
		\omega_{1,1}(z_1) &= \sum_{k=-\infty}^{\infty} \Res_{z=k\pi i} \left(\frac{ dz_1 } { z_1-z } - \frac{ dz_1 }{z_1-k\pi i}\right) \frac{1}{(y(2k\pi i-z)-y(z))dx(z)} \frac{dzd(2k\pi i-z)}{(z-(2k\pi i-z))^2}\\
		&= -\frac{dz_1}{16} \sum_{k=-\infty}^{\infty} \Res_{z=k\pi i} \frac{ dz } { (z_1-z)(z_1-k\pi i) } \frac{1}{(z-k\pi i)^2\sinh(z)} \\
		&= \frac{dz_1}{96} \sum_{k=-\infty}^{\infty} (-1)^k\left(\frac{6}{(z_1-k\pi i)^4}-\frac{1}{(z_1-k\pi i)^2}\right) \\
		&= \frac{dz_1}{16} \coth(z_1)\csch^3(z_1).
	\end{split}
\end{equation}
Observe that, in agreement with \Cref{p:periodic}, both these correlators are invariant under simultaneous shifts in all variables by the period of $x$. Now, with all of the $2g+n-2=1$ correlators at hand the order $\hslash^1$ contribution to the wavefunction is easily computed
\begin{equation}
	\begin{split}
		\frac{1}{3!}\int_{-\infty}^{z}\int_{-\infty}^{z}\int_{-\infty}^{z}\omega_{0,3}(z_1,z_2,z_3)+\int_{-\infty}^{z}\omega_{1,1}(z_1) &= -\sum_{k=-\infty}^{\infty} \frac{(-1)^k}{12(z-k\pi i)^3} - \frac{1}{48}\csch^3(z)\\
		&=-\frac{1}{48}\csch^3(z)\big(\cosh(2z)+4\big).
	\end{split}
\end{equation}
Finally, by making the replacement $z\mapsto\log(z)$ one finds
\begin{equation}
	\label{e:Nh1}
	\pi^*\left(-\frac{(z^4+8z^2+1)z}{12(z^2-1)^3}\right)=-\frac{1}{48}\csch^3(z)\big(\cosh(2z)+4\big),
\end{equation}
where $\pi$ is the exponential map.

To compare this answer to the usual one obtained for
\begin{equation}
	\mathcal{S} = \left(\mathbb{P}^1,\, x(z) = z+\frac{1}{z},\, y(z) = \log(z),\, \omega_{0,2}(z_1,z_2) = \frac{dz_1dz_2}{(z_1-z_2)^2}\right)\,,
\end{equation}
where $\log(z)$ is computed using the series 
\begin{equation}
	\log(z) = -\frac{1}{2}\sum_{k=0}^{\infty}\frac{(1-z^2)^k}{k}\,,\quad |1-z^2|<1\,,
\end{equation}
the $-\chi_{g,n}=1$ correlators must be computed. This is easily done and the following answer is obtained
\begin{equation}
	\begin{split}
		\omega_{0,3}(z_1,z_2,z_3) &= \frac{dz_1dz_2dz_3}{2(z_1-1)^2(z_2-1)^2(z_3-1)^2} + \frac{dz_1dz_2dz_3}{2(z_1+1)^2(z_2+1)^2(z_3+1)^2}\,, \\
		\omega_{1,1}(z_1) &= -\frac{(z_1^2+4z_1+1)(z_1^2-4z_1+1)(z_1^2+1)dz_1}{24(z_1^2-1)^4}\,.
	\end{split}
\end{equation}
Now compute the $\mathcal{O}(\hslash^1)$ term in the wavefunction for this spectral curve
\begin{multline}
		\frac{1}{2^3 \cdot 3!} \int_{z^{-1}}^z\int_{z^{-1}}^z\int_{z^{-1}}^z \omega_{0,3}(z_1,z_2,z_3) + \frac{1}{2}\int_{z^{-1}}^z \omega_{1,1}(z_1) \\ 
		= -\frac{\left(3z^{2} + 1\right) \left(z^2+3\right) z}{24 \left(z^2 - 1\right)^{3}} + \frac{\left(z^{2} + 2 z - 1\right) \left(z^{2} - 2 z - 1\right) z}{24 \left(z^2 - 1\right)^{3}} = -\frac{{\left(z^{4} + 8z^2 + 1\right)} z}{12 {\left(z^2 - 1\right)}^{3} }\,,
\end{multline}
which agrees with \eqref{e:Nh1}, as it must, by \Cref{p:GWagree}.

\subsection{$ P(x,y) = -e^{(f+1)xy} + e^{fxy} - x $}
\label{Comp_C3}
Consider now the periodic admissible spectral curve
\begin{equation}
	\mathcal{S}_\infty = \left(\P^1,\, x(z)=e^{fz}(1-e^{z}),\, y(z) = z,\, \omega_{0,2}(z_1,z_2) = \frac{dz_1dz_2}{(z_1-z_2)^2}\right)\,,
\end{equation}
where $z$ is an affine coordinate on $\P^1$ and $f\in\Q\setminus\{0,-1\}$. The finite ramification points are located at $a_k = -\Log(1+f^{-1})+2\pi i k$, where $\Log$ is the branch of the logarithm with the negative imaginary axis cut out and the imaginary part of the logarithm taking values in $(-\pi/2,3\pi/2)$. Then define the local deck transformation $\sigma_k$ through the equalities $x\circ \sigma_k = x$ and $\sigma_k(a_k)=a_k$. The first couple of terms in the expansion of $\sigma_k$ about $a_k$ are easily computed
\begin{equation}
	\sigma_k(z) = a_k - (z-a_k) - \frac{1}{3}(2f+1) (z-a_k)^2 - \frac{1}{9}(2f+1)^2 (z-a_k)^3 + \mathcal{O}\left((z-a_k)^4\right)\,.
\end{equation}
The $2g+n-2=1$ correlators may now be computed. First, calculate $\omega_{0,3}$ using \Cref{c:expprop}, defining $y(z) \coloneq \omega_{0,1}(z)/dx(z) = ze^{-fz}(1-e^z)^{-1}$
\begin{equation}
	\begin{split}
		\omega_{0,3}(z_1,z_2,z_3) &= \sum_{k=-\infty}^{\infty} \Res_{z = a_k} \frac{\omega_{0,2}(z,z_1)\omega_{0,2}(z,z_2)\omega_{0,2}(z,z_3)}{dx(z)dy(z)} \\
		&=-\frac{dz_1dz_2dz_3}{f(f+1)}\sum_{k=-\infty}^{\infty}\frac{1}{(z_1-a_k)^2(z_2-a_k)^2(z_3-a_k)^2}\,.
	\end{split}
\end{equation}
Next, $\omega_{1,1}$ may be computed. For this, SageMath \cite{Sage} was used and the result is
\begin{equation}
	\begin{split}
		\omega_{1,1}(z_1) &= \sum_{k=-\infty}^{\infty} \Res_{z=a_k} \left(\frac{dz_1}{z_1-z} - \frac{dz_1}{z_1-a_k}\right) \frac{1}{\sigma_k(z)-z} \frac{x(z)}{dx(z)} \frac{dz d\sigma_k(z)}{(z-\sigma_k(z))^2} \\
		&= \sum_{k=-\infty}^{\infty}\left( \frac{dz_1}{24(z_1-a_k)^2} + \frac{(2f+1)dz_1}{24f(f+1)(z_1-a_k)^3} - \frac{dz_1}{8f(f+1)(z_1-a_k)^4}\right)\,.
	\end{split}
\end{equation}
This will allow for the computation of the order $\hslash^1$ contribution to the wavefunction. Integrating $\omega_{0,3}$
\begin{multline}
\frac{1}{3!}\int_{\infty}^z\int_{\infty}^z\int_{\infty}^z\omega_{0,3}(z_1,z_2,z_3) = \frac{1}{6f(f+1)}\sum_{k=-\infty}^{\infty}\frac{1}{(z_1-a_k)^3} \\ = \frac{1}{48f(f+1)}\coth\left(\frac{z+\Log(1+f^{-1})}{2}\right)\csch^2\left(\frac{z+\Log(1+f^{-1})}{2}\right)\,,
\end{multline}
and then $\omega_{1,1}$
\begin{equation}
	\begin{split}
		\int_{\infty}^z \omega_{1,1}(z_1) &= -\frac{1}{24}\left(\frac{1}{z+\Log(1+f^{-1})}+2\sum_{k=1}^{\infty}\frac{z+\Log(1+f^{-1})}{(z+\Log(1+f^{-1}))^2+4\pi^2k^2}\right)\\ 
		&- \sum_{k=-\infty}^{\infty}\left(\frac{2f+1}{48f(f+1)(z+\Log(1+f^{-1})-2\pi i k)^2} - \frac{1}{24(f+1)f(z+\Log(1+f^{-1})-2\pi i k)^3}\right) \\
		&= -\frac{1}{48}\coth\left(\frac{z+\Log(1+f^{-1})}{2}\right) - \frac{2f+1}{192(f+1)}\csch^2\left(\frac{z+\Log(1+f^{-1})}{2}\right)\\
		&+ \frac{1}{192f(f+1)}\coth\left(\frac{z+\Log(1+f^{-1})}{2}\right)\csch^2\left(\frac{z+\Log(1+f^{-1})}{2}\right)\,.
	\end{split}
\end{equation}
Then, denoting the order $\hslash^1$ term in the wavefunction as $S_1(z)\coloneq (3!)^{-1}\int_{\infty}^z\int_{\infty}^z\int_{\infty}^z\omega_{0,3} + \int_{\infty}^z \omega_{1,1}$
\begin{equation}
	S_1(\log z) = \frac{z((f+1)z+f)}{12((f+1)z-f)^3} - \frac{(f+1)z+f}{48((f+1)z-f)} - \frac{f(2f+1)z}{48((f+1)z-f)^2} + \frac{z((f+1)z+f)}{48((f+1)z-f)^3}\,.
\end{equation}
It's interesting to note that neither the $\omega_{g,n}$, nor the coefficients $S_{-1},S_0,S_1$ vanish in the limits $f\to 0,-1$, as this limit does not commute with the infinite sum in hte correlators. Indeed, one can calculate
\begin{equation}
	\begin{split}
		\lim\limits_{f\to 0}\omega_{1,1}(z_1) = -\frac{dz_1}{48}e^{-z_1}\,,\qquad \lim\limits_{f\to -1}\omega_{1,1}(z_1) = -\frac{dz_1}{24}e^{z_1}\,.
	\end{split}
\end{equation}
In \cite{GS11} there was some confusion regarding the $f=0,-1$ cases, as the correlators do na\"ively vanish, but from the physics perspective should not, which suggests that there might be a `better' definition of TR in these unramified cases; this better definition is now generally understood to be the so-called $\log$-TR first systematically developed in \cite{ABDKS24} (see also \cite{H23b} for some early results), so the author is slightly hesitant to claim the above result is anything beyond a coincidence.

Now focus on the case $f=-1/2$ where $x(z) = -2\sinh(z/2)$, $\omega_{0,1}(z) = (z/2)\coth(z/2)dz$, and the deck transformation $\sigma_k$ takes the simple form $\sigma_k(z) = 2\pi i(2k-1) - z$. In this case $S_1(\log z)$ simplifies dramatically
\begin{equation}\label{e:GWTorhbar1}
	S_1(\log z) = \frac{z(z-1)}{3(z+1)^3} - \frac{z-1}{48(z+1)} + \frac{z(z-1)}{12(z+1)^3} = -\frac{(z-1)(z^2-18z+1)}{(z+1)^3}\,.
\end{equation}
It would be nice to check that this matches the order $\hslash^1$ contributions of the finite curve
\begin{equation}
	\mathcal{S} = \left(\mathbb{C}\setminus [0,-\infty),\, x(z) = z^f(1-z),\, y(z) = \log(z) / x(z),\, \omega_{0,2}(z_1,z_2) = \frac{dz_1dz_2}{(z_1-z_2)^2}\right)\,,
\end{equation}
for $f=-1/2$, as the quantum curves for $\mc{S}$ and $\mc{S}_\infty$ agree in this case. However, for $f=-1/2$, $\mc{S}$ would have $x(z) = z^{-1/2} + z^{1/2}$, which has branch cuts. Furthermore, by taking the $f \to -1/2$ limit, it is easy to convince oneself that the local deck transformation about the ramification point $z=-1$ should be $\sigma(z) = 1/z$, which only satisfies $x\circ \sigma = x$ if one starts to get very creative about interpreting branch cuts.

To avoid these difficulties, one can instead compute the correlators $\omega_{1,1}$ and $\omega_{0,3}$ for all integer $f$, and then analytically continue to all complex $f$ by requiring that the dependence on $f$ is rational (for explicit expressions of these correlators for all $f\in \Z$, see \cite{BM08}\footnote{\cite{BM08} uses a different sign convention where $\omega_{g,n}\to (-1)^{2g+n-2}\omega_{g,n}$.}). For $f=-1/2$ one obtains
\begin{equation}
	\begin{split}
		\omega_{0,3}(z_1,z_2,z_3) &= -\frac{4dz_1dz_2dz_3}{(z_1+1)^2(z_2+1)^2(z_3+1)^2}\,, \\
		\omega_{1,1}(z_1) &= \frac{dz_1}{24(z_1+1)^2} + \frac{dz_1}{2(z_1+1)^3} - \frac{dz_1}{2(z_1+1)^4}\,.
	\end{split}
\end{equation}
It is now easy to compute the corresponding order $\hslash^1$ term in the wavefunction\footnote{Note that the integration would give the same answer if the basepoint $b=1$ was chosen; this follows from the linear loop equations.}
\begin{equation}
	\begin{split}
		\frac{1}{3!2^3} \int_{z^{-1}}^{z}\int_{z^{-1}}^{z}\int_{z^{-1}}^{z} \omega_{0,3} + \frac{1}{2} \int_{z^{-1}}^{z} \omega_{1,1} &= -\frac{(z-1)^3}{12(z+1)^3} + \frac{z-1}{48(z+1)} + \frac{z-1}{8(z+1)} - \frac{z^3-1}{12(z+1)^3} \\
		&= -\frac{(z-1)(z^2-18z+1)}{(z+1)^3}\,,
	\end{split}
\end{equation}
which agrees with \eqref{e:GWTorhbar1}, as expected.
	
\subsection{$ P(x,y) = x^2 - \frac{4}{\pi^4}\sin^2(\pi\sqrt{-y}) $}
Consider the periodic and admissible, but not regular, spectral curve
\begin{equation}\label{e:WP}
	\mathcal{S} = \left(\mathbb{P}^1,\, x(z) = \frac{2}{\pi^2}\sin(\pi z),\, y(z) = -z^2,\, \omega_{0,2}(z_1,z_2) = \frac{dz_1dz_2}{(z_1-z_2)^2}\right)\,,
\end{equation}
where $z$ is an affine coordinate on $\P^1$. This is obtained from the well-known curve
\begin{equation}
	\mathcal{S}^{\vee} = \left(\mathbb{P}^1,\, x^{\vee}(z) = z^2,\, y(z) = \frac{2}{\pi^2}\sin(\pi z),\, \omega_{0,2}^{\vee}(z_1,z_2) = \omega_{0,2}(z_1,z_2) = \frac{dz_1dz_2}{(z_1-z_2)^2}\right)\,,
\end{equation}
via the $x$-$y$ swap $(x,y)\mapsto(y,-x)$.\footnote{The sign here is essentially irrelevant as it just rescales the correlators $\omega_{g,n} \mapsto (-1)^{2g+n-2}\omega_{g,n}$; the resulting correlator computed here looks slightly nicer with the chosen sign convention. \cite{ABDKS22} does not flip this sign, which results in \eqref{e:x-y-explicit} having slightly different signs than shown in \cite{ABDKS22}.} By \Cref{c:perx-y} the correlators of these two curves will be related by an algebraic formula given in \cite{ABDKS22}. Here, we will check this for $\omega_{1,1}$. For results about $\omega_{0,3}$ and $\omega_{0,4}$, see \cite{A25}.


First, for a ramification point $k\in\mathbb{Z}+1/2$ the local Galois conjugate is $\sigma_{k}(z) = 2k-z$. With this compute
\begin{equation}\label{e:w11JT}
	\begin{split}
		\omega_{1,1}(z_1) &= \sum_{k\in\mathbb{Z}+1/2} \Res_{z=k}\, \frac{1}{2} \left(\frac{ dz_1 } { z_1-z } - \frac{ dz_1 }{z_1-(2k-z)}\right) \frac{1}{(y(2k-z)-y(z))dx(z)} \frac{dzd(2k-z)}{(z-(2k-z))^2}\\
		&= -\frac{dz_1}{32} \sum_{k\in\mathbb{Z}+1/2} \Res_{z=k} \frac{ dz } { (z_1-z)(z_1+z-2k) } \frac{ \pi } { \cos(\pi z) k (z-k)^2 } \\
		&= \frac{1}{32} \sum_{k\in\mathbb{Z}+1/2} \frac{(-1)^{k-1/2}}{k} \left( \frac {dz_1} { (z_1-k)^4} + \frac {\pi^2 dz_1} {6(z_1-k)^2} \right)\,.
	\end{split}
\end{equation}

Now compare the above to the $x$-$y$ dual
\begin{equation}
		\omega_{0,1}^{\vee}(z_1) = \frac{\pi^3 dz_1}{192 z_1^2} + \frac{\pi dz_1}{32 z_1^4}\,.
\end{equation}
Then (see \cite{ABDKS22}) 
\begin{equation}\label{e:x-y-explicit}
	\begin{split}
		\omega_{1,1}(z_1) & = \omega_{1,1}^{\vee}(z_1) + d_{z_1}\left(\frac{\bar{\omega}^{\vee}_{0,2}(z_1,z_1)}{2dx^{\vee}(z_1)dy^{\vee}(z_1)} - \frac{1}{24}\frac{d^2}{dy^{\vee}(z_1)^2}\frac{dy^{\vee}(z_1)}{dx^{\vee}(z_1)}\right)\,, \\
		\bar{\omega}^{\vee}_{0,2}(z_1,z_2) & \coloneq \omega^{\vee}_{0,2}(z_1,z_2) - \frac{dx^{\vee}(z_1)dx^{\vee}(z_2)}{(x^{\vee}(z_1)-x^{\vee}(z_2))^2}\,,
	\end{split}
\end{equation}
where $i',i''$ are chosen so that $\{i,i',i''\} = \llbracket 3 \rrbracket$ (the ordering does not matter). Applying these formulae to the case considered here one obtains
\begin{equation}
  \omega_{1,1}(z_1) = \frac{\pi {\left(\pi^{2} z_1^{2} + 2 \, {\left(\pi^{2} z_1^{2} + 3 \, \pi {\left(\pi {\left(\pi z_1 \tan\left(\pi z_1\right) - 1\right)} z_1 \sec\left(\pi z_1\right)^{2} + \tan\left(\pi z_1\right)\right)} z_1 - 3\right)} \sec\left(\pi z_1\right) + 6\right)}}{192 \, z_1^{4}} \,,
\end{equation}
and it is straightforward to check that the above expression has no pole at $z_1=0$ and correctly reproduces the principal part shown in \eqref{e:w11JT} at $z_1 = k+1/2$ for every $k\in\Z$. A similar procedure may be carried out for $\omega_{0,3}$ and the author has checked on SageMath the same agreement of the two expressions for $\omega_{0,3}$. It is interesting to notice that, in this case, it is generally easier to calculate the correlators using the $x$-$y$ swap formula than by proceeding directly through the periodic TR formalism. 

%

{\setlength\emergencystretch{\hsize}\hbadness=10000
\printbibliography}

\end{document}